\documentclass[11pt]{article}

\usepackage[T1]{fontenc}
\usepackage[utf8]{inputenc}
\usepackage{lmodern}              
\usepackage{microtype}            

\usepackage{graphicx}
\graphicspath{{Figures/}}         

\usepackage{amsmath,amssymb,amsthm}
\usepackage{mathtools}            
\usepackage{bm}                   

\usepackage{booktabs}
\usepackage{subcaption}           
\usepackage{enumitem}
\usepackage{mathrsfs}             
\usepackage{algorithm}
\usepackage{algorithmic}
\usepackage{psfrag}               
\usepackage{tikz}
\usetikzlibrary{arrows.meta,positioning,shapes.multipart}

\usepackage[margin=1in]{geometry} 

\usepackage[hyphens]{url}
\usepackage[colorlinks=true,allcolors=blue!60!black]{hyperref}
\usepackage[nameinlink]{cleveref} 
\usepackage[round,authoryear]{natbib}

\numberwithin{equation}{section}   
\newtheorem{theorem}{Theorem}[section]
\newtheorem{lemma}[theorem]{Lemma}
\newtheorem{proposition}[theorem]{Proposition}
\newtheorem{corollary}[theorem]{Corollary}
\theoremstyle{remark}
\newtheorem{remark}[theorem]{remark}

\DeclareMathOperator{\Corr}{Corr}
\DeclareMathOperator{\Range}{Range}

\newcommand{\E}{\mathbf{E}}

\newcommand{\X}{\mathcal{X}}

\newcommand{\mT}{\mathcal{T}}
\newcommand{\mF}{\mathcal{F}}
\newcommand{\mM}{\mathcal{M}}
\newcommand{\Beta}{\mathrm{Beta}}
\newcommand{\ud}{\mathrm{d}}

\newcommand{\be}{\begin{equation}}
\newcommand{\ee}{\end{equation}}
\newcommand{\bea}{\begin{eqnarray}}
\newcommand{\eea}{\end{eqnarray}}

\floatname{algorithm}{Procedure}

\title{Adversarial Obstacle Placement with Spatial Point Processes for Optimal Path Disruption}
\author{%
  Li Zhou\thanks{Department of Mathematics and Statistics, Auburn University, Auburn, AL, USA. \texttt{lzz0062@auburn.edu}} \and
  Elvan Ceyhan\thanks{Corresponding author. Department of Mathematics and Statistics, Auburn University, Auburn, AL, USA. \texttt{ceyhan@auburn.edu}} \and
  Polat Charyyev\thanks{MAP Akademi, Istanbul, Turkey. \texttt{polatturkmen@gmail.com}}
}
\date{} 

\newcommand{\funding}{\footnotetext{This work was partially supported by ONR Grant N00014-22-1-2572 and NSF Award \#2319157. 
The authors contributed equally to this work.}}

\begin{document}
\maketitle
\funding

\abstract{
We investigate the Optimal Obstacle Placement (OOP) problem under uncertainty,
framed as the dual of the Optimal Traversal Path problem in the Stochastic
Obstacle Scene paradigm. 
We consider both continuous domains, discretized for analysis, 
and already discrete spatial grids that form weighted geospatial networks using 8-adjacency lattices. 
Our unified framework integrates OOP with stochastic
geometry, modeling obstacle placement via Strauss (regular) and Matérn
(clustered) processes, and evaluates traversal using the Reset Disambiguation algorithm. 
Through extensive Monte Carlo experiments, 
we show that traversal cost increases by up to 40\% 
under strongly regular placements, 
while clustered configurations can decrease traversal
costs by as much as 25\% by leaving navigable corridors
compared to uniform random layouts. 
In mixed (with both true and false obstacles) scenarios,
increasing the proportion of true obstacles from 30\% to 70\% nearly doubles
the traversal cost. These findings are further supported by statistical
analysis and stochastic ordering, providing rigorous insights into how spatial
patterns and obstacle compositions influence navigation under uncertainty.
}



\section{Introduction} 
\label{sec:intro}
\noindent
Efficient pathfinding in uncertain or dynamic spatial environments is a central problem in 
geographic information science (GIScience), with broad applications in autonomous navigation, urban 
mobility planning, defense logistics, and environmental monitoring.  
Real-world scenarios—such as maritime navigation in mine-infested zones, cities with dynamic 
construction, or landscapes fragmented by environmental hazards—require agents to make routing 
decisions under uncertainty.  
The \emph{Stochastic Obstacle Scene} (SOS) problem and its discrete analogue, 
the \emph{Canadian Traveler's Problem} (CTP), capture this challenge by modeling settings where agents must traverse 
from a source to a destination through regions containing uncertain obstacles 
\citep{papadimitriou:1991,bar-noy:1991}.

In addition to work on the CTP, our study connects with
several broader strands of research:

\textbf{(i) Spatial point process modeling of obstacles.}
The use of stochastic geometry to model spatial uncertainty is well-established in
ecological and environmental planning. Classical references such as \cite{diggle2003spatial},
\cite{moller2004statistical}, and \cite{illian2008statistical} provide comprehensive
treatments of point process models including Strauss, Matérn, and hardcore processes.
These approaches have been applied to urban growth, forestry, and habitat modeling,
and here we adapt them to adversarial obstacle placement in navigation domains.

\textbf{ (ii) Navigation and path planning under uncertainty.}
Beyond CTP, the robotics and motion-planning literature has long addressed
navigation in partially known or dynamic environments. \cite{lavalle2006planning} surveys
motion planning algorithms such as probabilistic roadmaps and RRTs, which account
for uncertainty in obstacle fields. \cite{howard2002risk} introduce
risk-aware planning for mobile robots operating with incomplete information,
highlighting themes similar to our disambiguation-based traversal.

\textbf{(iii) Empirical studies in GIS-based routing.}
In applied GIS, several studies examine how spatial patterning of hazards influences
routing outcomes without invoking the CTP framework. For example, recent work on
flood evacuation \citep{Parajuli2023FloodEvac} and off-road path planning \citep{Lv2024OffRoad}
demonstrates the importance of integrating spatial randomness and hazard clustering
into path evaluation. Our framework contributes to this line of work by explicitly modeling
the influence of obstacle pattern on traversal cost.

While much of the literature focuses on developing effective traversal strategies for a navigating agent (NAVA), 
the inverse problem—how an adversary might strategically place obstacles to hinder 
movement—has received relatively little attention.  
This dual formulation, known as the \emph{Optimal Obstacle Placement} (OOP) problem, 
introduces a  second agent, the obstacle-placing agent (OPA), 
who aims to maximize the expected traversal cost of the NAVA.  
The OOP problem, introduced by \cite{aksakalli:2012} as the Optimal Placement with Disambiguations (OPD) problem, 
generalizes the SOS framework. 
Prior work examined specific layouts or background clutter but did not systematically study 
how spatial pattern (regularity, uniformity, clustering) shapes traversal cost.
\emph{Our contribution lies
in addressing this gap by systematically evaluating the impact of obstacle pattern
geometry and obstacle composition on navigability.
}
The importance of OOP is not limited to theoretical curiosity; 
it varies with the geospatial use case. 
In maritime defense, adversaries may deploy mines in spatially
regular patterns to maximize disruption of naval logistics. In urban environments,
construction zones and artificial blockages may act as strategically placed
obstacles that reroute traffic. 
In environmental settings, clustered hazards such as
flood debris or landslides may mimic Matérn-type patterns, creating narrow corridors for evacuation. 
These prospective applications illustrate
why understanding obstacle placement under uncertainty is useful across multiple
domains, from defense to urban planning to disaster response.
A recent flood-evacuation routing study underscores the practical importance of path planning 
under spatial risk \citep{Parajuli2023FloodEvac}.  

In this work, we develop a unified framework for the OOP problem that incorporates both continuous 
and discretized spatial representations.  
The continuous setting models obstacles as disks with uncertain status; the discrete setting 
uses an 8-adjacency spatial grid to convert the environment into a weighted geospatial network.
We assume two types of obstacles exist in the navigation domain: true obstacles (which are non-traversable
or block the traversal) such as mines and false obstacles (which are traversable) such as mine-like objects
(i.e., objects resembling mines).   
This duality supports GIS-compatible analysis alongside tractable computational methods.  
A NAVA uses a greedy strategy—called the \emph{Reset Disambiguation (RD)} algorithm 
\citep{aksakalli2011}—that re-evaluates the shortest path whenever a true obstacle is encountered, 
incurring a disambiguation cost based on sensor uncertainty.
Sensors provide probability marks (modelled with a beta-distribution), prompting disambiguation actions that influence 
final traversal costs.  
In this setting, 
NAVA can only disambiguate (but not neutralize) the obstacle and thus can determine the actual status of the obstacle as true or false
at an additional cost to traversal cost.
Recent work on grid-based and off-road path planning demonstrates continued interest in coupling 
advanced algorithms with geospatial data \citep{An2024HexAStar, Lv2024OffRoad}.
While motivated by SOS/CTP, 
our focus lies in understanding how spatial obstacle patterns influence traversal cost across a broader class of stochastic environments. 
We do not propose new CTP algorithms but instead explore how underlying spatial processes affect algorithmic performance.

Our obstacle placement strategies leverage spatial point processes to model different spatial 
patterns.  
In particular, we use the Strauss process to represent spatial regularity and the Matérn  
process to represent spatial aggregation/clustering.  
This design enables controlled comparisons of obstacle layout and composition.  
Through extensive Monte Carlo (MC) experiments, we examine how traversal cost varies under different 
spatial configurations and ratios of true to false obstacles.  
To analyze outcomes, we employ robust regression, random forests \citep{breiman2001}, and 
zero-inflated negative binomial models \citep{Zeileis2008}, 
offering both statistical rigor and flexibility.

Theoretical analysis complements these empirical findings by establishing a stochastic ordering among 
the path-weight distributions.  
Configurations consisting solely of false obstacles are stochastically dominated by mixed-obstacle 
configurations, which are in turn dominated by true-only arrangements in terms of induced traversal 
cost.  
Specifically, it shows that when obstacles are all false (i.e., not truly blocking), 
the resulting paths tend to be shorter and less costly to traverse. 
When some of the obstacles are true (i.e., actual obstructions), 
the paths become more uncertain and typically longer. 
In the most obstructive case—when all obstacles are true—the traversal cost is highest. 
This establishes a clear ordering: 
false-only configurations lead to the lowest expected traversal cost, followed by mixed obstacles, 
with true-only configurations resulting in the highest cost. 
These findings demonstrate how spatial structure and obstacle composition jointly influence 
navigability in adversarial settings.  

The OOP problem also bears conceptual similarity to the well-known \emph{network interdiction problem} 
\citep{israeli2002shortest,smith2020survey}, which models a leader–follower game where an 
interdictor disables parts of a network to increase traversal costs for an adversary.  
However, our framework differs significantly: it emphasizes partial information, spatially extended 
obstacles (e.g., disk regions), and dynamic learning (via disambiguation).  
These elements are rarely addressed in classical interdiction literature, although recent extensions 
(e.g., \cite{sundar2021credible,azizi2024shortest,sadeghi2024modified}) begin to incorporate such 
dynamics.
We illustrate the model with a naval logistics scenario and note its relevance for urban mobility, ecology, and flood evacuation.


The main contributions of this paper are as follows:
\begin{enumerate}
    \item We propose a unified OOP framework that couples obstacle placement
    with stochastic geometry via Strauss and Matérn point processes, capturing
    both regular and clustered obstacle layouts.
    \item We extend OOP analysis to compositional settings that include
    false-only, true-only, and mixed obstacle fields, thus accounting for both
    physical blockage and deceptive clutter.
    \item We conduct extensive Monte Carlo experiments across a wide range of
    parameter settings and analyze outcomes using robust regression, random
    forests, and zero-inflated models to quantify the effects of obstacle
    pattern and composition.
    \item We introduce stochastic ordering as a rigorous tool to compare
    traversal cost distributions under alternative obstacle placement strategies.
    \item We present an illustrative geospatial case study
    to demonstrate the real-world applicability of the proposed framework.
\end{enumerate}

The remainder of this paper is organized as follows.  
Section~\ref{sec:OOP-problem} formalizes the OOP problem and our assumptions.  
Section~\ref{sec:gis-implications} discusses GIS-based implications of our findings, 
highlighting applications in urban mobility, environmental modeling, and maritime navigation. 
Section~\ref{sec:meth-exp-setting} outlines the experimental design and statistical modeling approach.  
Results and insights are presented in Section~\ref{sec:MCexp-results-and-analysis}
and an illustrative geospatial case study is provided in Section~\ref{sec:llust-geospat}.
Section~\ref{sec:stoch-order} explores theoretical comparisons using stochastic ordering.  
Finally, Section~\ref{sec:disc-conc} offers conclusions and future research directions.
Proofs of theoretical results and details of the extensive Monte Carlo experiments are deferred to the Appendix.

\section{The Optimal Obstacle Placement Problem}
\label{sec:OOP-problem}
The SOS problem, introduced by \citet{papadimitriou:1991},  
originally focused on computing the \emph{Optimal Traversal Path} (OTP)  
for a NAVA in a stochastic environment with obstacles.  
Its discrete analogue, the CTP,  
has received substantial attention in both theory and applications  
\citep{bar-noy:1991, nikolova:2008, eyerich:2009}.  
A complementary and less-explored formulation considers  
an OPA whose objective is to strategically insert obstacles  
to hinder the NAVA’s movement by maximizing traversal cost.  
This formulation defines the OOP problem,  
introduced by \citep{aksakalli:2012},  
which identifies worst-case obstacle configurations (for NAVA)  
within a designated \emph{insertion window}.  
Both the OTP and OOP problems can be studied in continuous and discrete domains,  
and their interplay underpins a broader class of adversarial path planning problems.

\subsection{The Continuous OOP Problem}
\label{sec:continuous-OOP-problem}
Consider a bounded region $\Omega \subset \mathbb{R}^2$, where an OPA inserts obstacles 
modeled as disks $D_x$ centered at $x \in \mathcal{X}$ with fixed radius $r > 0$. 
Let $\mathcal{X}_\mF$ and $\mathcal{X}_\mT$ denote the centers of false and true obstacles, 
generated from spatial point processes $\mathcal{P}_\mF$ and $\mathcal{P}_\mT$, respectively.

A NAVA traverses from $s$ to $t \in \Omega$, relying on a sensor that assigns 
probabilities $p: \mathcal{X} \to [0,1]$, 
where $p(x)$ indicates the probability that obstacle $x$ is true. 
Sensor outputs are modeled with Beta distributions: $p(x) \equiv \text{Beta}(a,b)$, with $a<b$ for false 
and $a>b$ for true obstacles—ensuring that true obstacles are, on average, assigned higher probabilities. 
Increasing the gap $|a - b|$ models higher sensor discrimination.

The sensor marks are drawn independently as $p_F$ for $x \in \mathcal{X}_\mF$ and $p_T$ for 
$x \in \mathcal{X}_\mT$, according to:
\[
p(x)=\begin{cases}
 p_F, & \text{if } x \in \mathcal{X}_\mF \\ 
 p_T, & \text{if } x \in \mathcal{X}_\mT \\ 
 0, & \text{otherwise}
\end{cases}
\]
Although the NAVA observes the obstacle locations $\mathcal{X} = \mathcal{X}_\mT \cup \mathcal{X}_\mF$, 
their true status remains unknown unless disambiguated. 
Each disambiguation incurs a fixed cost $c > 0$, 
typically interpreted as time, which is added to the overall traversal cost.

The \emph{continuous OTP problem} then seeks the minimum-cost $(s,t)$ path avoiding true obstacles, 
while the \emph{continuous OOP problem} seeks to maximize this cost through strategic obstacle placement. 
To navigate the uncertain environment, the NAVA evaluates paths by balancing Euclidean distance and 
the risk based on from $p(x)$. 
Heuristic strategies such as the \textit{Risk-Aware Greedy Algorithm} 
\citep{missiuro2006adapting, aoude2013probabilistically} select paths based on a composite measure 
of length and estimated risk. These methods may be enhanced through probabilistic planners like 
Rapidly-exploring Random Trees (RRTs) or risk-weighted A* variants, enabling adaptive traversal 
in uncertain and spatially complex environments \citep{Meng2022nrrrt,Chung2019risk}.

\subsection{The Discretized OOP Problem}
\label{sec:dicrete-OOP-problem}
To facilitate computation, the continuous domain $\Omega$ is discretized into an $n \times m$ grid,  
forming an \emph{8-adjacency integer lattice} \citep{aksakalli:2012}.  
Obstacles are modeled as disks of fixed radius \citep{witherspoon1995COBRA},  
and grid resolution is chosen to closely approximate continuous traversal.  
The resulting graph $G = (V, E)$ contains vertices at grid points  
and edges connecting adjacent vertices, including diagonals.  
Each interior vertex links to eight neighbors:  
four unit-length and four diagonal ($\sqrt{2}$-length) edges.  
Additionally, edge connections are added along the grid boundary to preserve connectivity.  
A start vertex $s$ and target vertex $t$ are designated.  
The NAVA seeks a path from $s$ to $t$ while minimizing a traversal cost  
that incorporates both Euclidean distance and the risk associated with uncertain obstacles,  
disambiguated at cost $c > 0$ when necessary.  
This discrete setup corresponds to the CTP  
with spatially dependent stochastic costs  
\citep{nikolova:2008, eyerich:2009, xu:2009CTP}. 

Obstacles are placed within a designated window between source and target,
representing the adversarial region of influence. This window is sized to cover
the main traversal corridor while leaving peripheral areas open, ensuring that
placement decisions are consequential but do not trivially block all routes.

The OPA’s objective in the discretized OOP problem  
is to place obstacles within an insertion window $\Omega_{\mathcal{O}} \subset \Omega$—  
typically a homothetic subregion—  
to maximize expected traversal cost.  
A coastal defense analogy illustrates the setting:  
an OPA delays an intruding vessel (NAVA)  
navigating through a mined nearshore zone (Figure~\ref{fig:cobra-traversal}).  
The annular (i.e. ring-shaped) obstacle-free window ensures that traversal remains feasible but strategically costly.  
The grid is aligned such that $\partial(\Omega_{\mathcal{O}})$  
coincides with grid cell boundaries. 

Obstacle radius is set to $r=4.5$ \citep{witherspoon1995COBRA}. 
This value is chosen to ensure that each obstacle intersects several grid edges
(roughly spanning $9$--$10$ units on a $101 \times 101$ grid), thereby exerting
a nontrivial effect on traversal. At the same time, the radius is small enough
that feasible corridors between source and target remain available. 
This scaling follows prior studies on the SOS framework
\cite{aksakalli2011}, where obstacle radii were selected relative to grid
resolution to balance obstruction with navigability. 

In the discrete setting, traversal is modeled on a weighted graph 
where each edge $e$ has a baseline length $\ell(e)$, 
equal to the Euclidean distance between its endpoints 
($\ell(e)=1$ for horizontal/vertical, $\ell(e)=\sqrt{2}$ for diagonals). 
Obstacle disks intersecting an edge add uncertainty to its cost, 
with disambiguation determining whether the edge is blocked or available.

The NAVA follows the RD algorithm, 
which adaptively recomputes shortest paths based on perceived risk. 
For a path $\pi(s,t)$ from $s$ to $t$, the weight of edge $e$ is defined as
\begin{equation}
\label{eqn:weight-of-edges}
w(e) = \ell(e) + \tfrac{1}{2} F(e), \quad \text{where} \quad
F(e) = \sum_{x \in \X} \mathbf{1}_{{D_r(x) \cap e \neq \emptyset}}
\left( \frac{c(x)}{1 - p(x)} \right),
\end{equation}
with $p(x)$ the sensor-assigned probability that obstacle $x$ is true, $c(x)$ the disambiguation cost, 
and $\mathbf{1}_{\{\cdot\}}$ the indicator function. 
The term $F(e)$ captures the cumulative risk from all uncertain obstacles intersecting edge $e$.

The total path weight is then
\begin{equation}
\label{eqn:Lpi}
W(\pi,\X) = \sum_{e \in \pi(s,t)} w(e),
\end{equation}
used to approximate the expected traversal cost perceived by the NAVA before disambiguation.  
When no obstacles are present, $W(\pi, \X) = L_{\pi} = \sum_{e \in \pi(s,t)} \ell(e)$,  
and NAVA simply follows the shortest path.  
With uncertain obstacles, $W(\pi, \X)$ becomes stochastic (i.e. random) due to Beta-distributed $p(x)$
and also the stochastic nature of obstacle locations.  
The RD algorithm adaptively recomputes paths upon disambiguating true obstacles,  
resetting traversal from the current location.  
It extends Dijkstra’s algorithm to accommodate dynamic edge weights  
derived from spatial uncertainty \citep{Dijkstra1959}.

\begin{figure}[t]
\centering
\captionsetup{width=.95\linewidth}
\begin{minipage}{0.48\textwidth}
    \centering
\includegraphics*[scale=.425]{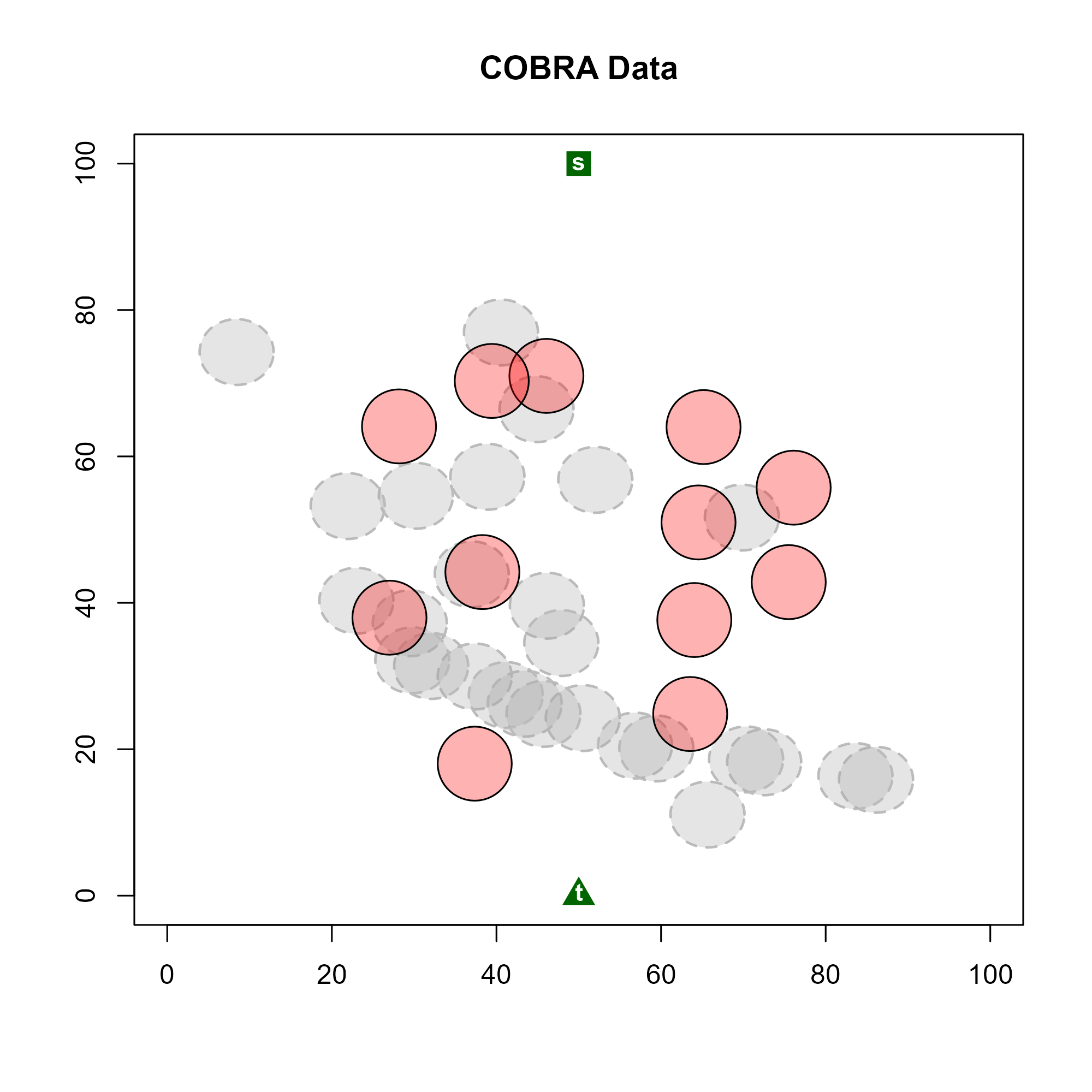}\\
    (a)
\end{minipage}\hfill
\begin{minipage}{0.48\textwidth}
    \centering
\includegraphics*[scale=.425]{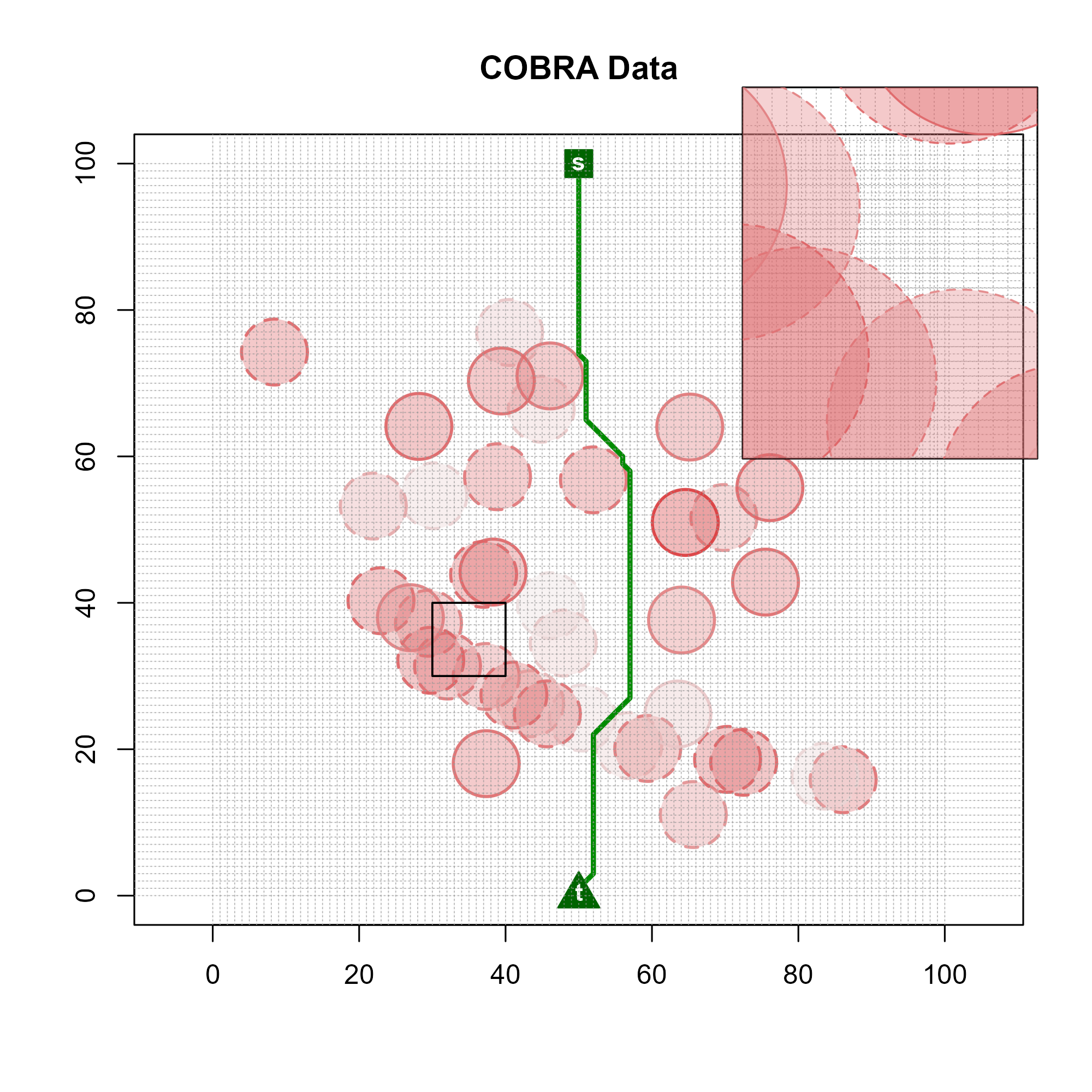}\\
    (b)
\end{minipage}
\caption{
\label{fig:cobra-traversal}
(a) Coastal battlefield reconnaissance and analysis (COBRA) data with 12 mines
(orange circles) and 27 false obstacles (light gray circles) \citep{witherspoon1995COBRA}.
(b) Sensor-derived obstacle probabilities (darker orange indicates higher $p(x)$) and corresponding NAVA traversal path computed using the RD algorithm.
}
\end{figure}

\subsection{The Distinction between Continuous vs. Discrete Traversal}
Two formulations of the OOP problem can be distinguished. 
In the \emph{continuous formulation}, the agent’s trajectory is modeled as an arbitrary curve 
in the plane that avoids true obstacles. 
This description is natural in open domains such as maritime navigation or off-road mobility, 
where paths are not confined to predefined routes. 

In the \emph{discrete formulation}, movement is represented as a walk on a graph. 
This graph may arise either from discretizing the continuous domain (e.g., an $n \times m$ lattice) 
or from a naturally occurring network such as streets, corridors, or utility grids. 
Although the underlying space may be physically continuous and traversable along edges, 
the model treats traversal as \emph{node-to-node steps}: the agent moves only between 
adjacent vertices and never halts at intermediate points on an edge. 
Thus, the discreteness stems from the representation of movement, not from the geometry of the environment itself. 
In our study we employ an 8-adjacency lattice, which permits orthogonal and diagonal moves; 
restricting to 4-adjacency is possible but typically lengthens paths by removing diagonal shortcuts. 

Application domains align naturally with these formulations: 
road networks, building layouts, and infrastructure grids lend themselves to the discrete model, 
while navigation in open terrain or sea is better captured by the continuous one. 
Figure~\ref{fig:cobra-traversal} illustrates this distinction using Coastal battlefield reconnaissance and analysis (COBRA) data with 12 mines
and 27 false obstacles (mine-like objects) \citep{witherspoon1995COBRA}.
The left panel shows the original mine and clutter field, while the right panel overlays an 
8-adjacency spatial grid on the same region and depicts a sample NAVA traversal path using RD algorithm 
computed on the corresponding graph representation.
To better illustrate the spatial grid, we zoom in a rectangular region at top right of the right panel plot.

\subsection{OOP as an Optimization Problem}
\label{sec:OOP-optimization-problem}
Let $C(\pi,\X)$ denote the realized traversal cost from $s$ to $t$ on $G$, and
recall $W(\pi,\X)$ as the perceived (pre-traversal) path weight based on sensor marks. 
Before traversal, $C(\pi,\X)$ and $W(\pi,\X)$ are distinct random
variables: $W$ anticipates cost under uncertainty, whereas $C$ includes the
actual disambiguation outcomes and their costs.

The OTP problem—continuous or discrete—can be written as
\begin{flalign}
\label{eqn:otp-formula}
\min_{\X}\; \E\!\left[C(\pi,\X)\right]
\quad \text{s.t.} \quad
\X \subseteq \Omega_{\mathcal O},\;\;
\X_\mT \cap \pi = \emptyset,\;\;
|\X| = n,
\end{flalign}
where $\E\!\left[C(\pi,\X)\right]$ is the expected traversal cost,
$\Omega_{\mathcal O}\subset \Omega$ is the insertion window, 
and $n$ the number of obstacles placed by the OPA. 
The OOP problem replaces $\min$ with $\max$ in \eqref{eqn:otp-formula}. 
The weight $W(\pi,\X)$ depends on the spatial
configuration $\X=\X_\mT\cup \X_\mF$ and the associated probability marks.

\begin{remark}
\label{rem:path-vs-walk}
\textbf{(Traversal Route: Path or Walk?)}
In the discretized setting, the NAVA may revisit vertices/edges due to
re-planning after disambiguation, so the route is a \emph{walk} in the
graph-theoretic sense \citep{west:2001}. 
For readability and consistency with routing terminology, 
we still use ``path" to refer to the traversed sequence,
which approximates a continuous trajectory in space.
\end{remark}

\section{GIS-Based Implications and Applications}
\label{sec:gis-implications}
\noindent
The proposed OOP framework integrates naturally with GIS-enabled spatial decision-support systems (SDSS) by
linking stochastic obstacle layouts, sensor-informed uncertainty, and re-planning into standard geospatial
workflows (raster surfaces, vector networks, and live sensor layers).

\subsection{Urban mobility and traffic resilience}
Transportation agencies regularly face temporary blockages 
(construction, events, incidents) that disrupt routing. 
RD-based re-planning provides a principled way to stress-test adaptive detour strategies on
partially observable street networks. 
Our finding that moderately regular obstacle patterns raise traversal costs 
complements vulnerability studies showing 
how dispersed link failures degrade performance \citep{Jenelius2015Vulnerability}. 
The setup is directly compatible with mainstream GIS network
datasets for worst-case delay analysis.

\subsection{Landscape ecology and wildlife corridors}
In fragmented habitats, uncertain permeability (roads, fences, land-use transitions) 
can be modeled as probabilistic obstacles. 
Sensor-driven disambiguation mirrors perceptual uncertainty and pairs well with
resistance surfaces and circuit-theory connectivity \citep{McRae2008Circuit}, 
enabling evaluation of corridor
designs under varying obstacle densities and sensing quality.

\subsection{Maritime and defense logistics}
Mine-suspected or cluttered waters mirror our coastal navigation setting. 
Coupling RD with acoustic/optical sensor feeds supports 
estimation of worst-case transit times and assessment of sensor placement strategies,
extending risk-aware routing for autonomous surface vessels \citep{Maidana2023}. 
The representation aligns with
standard electronic navigational charts.

\subsection{Integration with GIS platforms}
The discretized domain corresponds to raster-cell adjacency, 
while the graph abstraction maps to polyline networks. 
Sensor probabilities can be stored as raster values or edge attributes, 
enabling what-if analyses in SDSS tools (e.g., ArcGIS ModelBuilder, QGIS Processing). 
Recent work on off-road routing, flood evacuation, and
hex-grid navigation \citep{Lv2024OffRoad,Parajuli2023FloodEvac,An2024HexAStar} illustrates the utility of
spatially aware, risk-based routing.

By linking stochastic obstacle modeling with GIS analytics, 
our framework bridges theory and practice for adaptive, 
risk-aware routing across transportation, ecology, and defense.

In GIS, the obstacle window \(\Omega_{\mathcal O}\) has direct geographical meaning (e.g., road segments subject
to closure; minefields constraining shipping lanes). RD operates on the corresponding network or grid.

Dynamic settings can be handled by updating \(\Omega_{\mathcal O}\) over time (e.g., shifting flood debris,
time-varying closures). While full temporal modeling is beyond the scope here, our design accommodates such
updates. The urban evacuation case study (Section~\ref{sec:gis-implications}) exemplifies how \(\Omega_{\mathcal O}\)
and RD map to an operational GIS context.

\section{Methodology and Experimental Setting }
\label{sec:meth-exp-setting}

\subsection{Literature on Traversal Algorithms and Prior Work on Optimal Obstacle Placement}
Prior work on stochastic obstacle navigation proposes several heuristic
algorithms for the NAVA, including BAO\* \citep{aksakalli2007}, Simulated Risk Disambiguation
(SRA) \citep{fishkind2007}, Distance to Termination (DT) \citep{aksakalliari2013}, 
and Reset Disambiguation (RD) \citep{fishkind2007}. 
Each has advantages and limitations: BAO\* is exhaustive
but computationally demanding, DT underuses disambiguation, and SRA requires parameter tuning. 
RD offers a practical balance in grid-based SOS settings.
These studies largely focus on \emph{traversal} rather than \emph{obstacle placement}.

The OOP problem has been studied in
settings where an adversary seeks to maximize the traversal cost of a NAVA. 
Early work used grid-based formulations with random
clutter and evaluated heuristic traversal under different
placement strategies \citep{aksakalli2011}. 
Subsequent research explored variants combining placement with traversal heuristics such as RD, DT, SRA,
and BAO\* \citep{aksakalli2007, fishkind2007, aksakalliari2013}, 
but generally relied on fixed or simplified configurations and 
did not systematically assess how \emph{spatial point processes} shape navigability.

Our framework couples OOP with Strauss and Matérn spatial point process models
to analyze how obstacle regularity, clustering, and composition affect traversal
outcomes, and pairs this with statistical modeling for rigorous inference.
For the traversal baseline, RD has complexity
$O\!\left(k \cdot (|E| + |V|\log |V|)\right)$ with Dijkstra’s algorithm, 
where $k$ is the number of disambiguations. 
In our $101 \times 101$ grids,
$|V|\!\approx\!10^4$, $|E|\!\approx\!8|V|$, and typical runs have $k<20$;
empirically, RD completes in under one second per realization on a standard desktop.

\subsection{Proposed Framework}
We propose a unified framework for the OOP problem that couples
(i) spatial point processes (Strauss and Matérn) for obstacle layout,
(ii) stochastic sensor marks (Beta distributions), and
(iii) traversal evaluation using RD.
This enables systematic analysis of how spatial structure and composition of obstacles influence traversal cost.
Figure~\ref{fig:workflow} illustrates  our empirical evaluation approach, from obstacle generation to
statistical modeling of outcomes.

\begin{figure}[!ht]
\centering
\begin{tikzpicture}[node distance=10mm, >=stealth, thick]
\tikzstyle{flow}=[draw, rounded corners=2pt, minimum width=34mm, minimum height=12mm, align=center]
\node[flow] (gen) {Obstacle Field\\Generation\\(Strauss, Matérn)};
\node[flow, right=of gen, yshift=0mm] (sens) {Assign Sensor\\Probabilities};
\node[flow, right =of sens, yshift=0mm] (trav) {Traversal by NAVA\\(RD Algorithm)};
\node[flow, below=of sens, xshift=0mm, yshift=-5mm] (outc) {Traversal Outcomes\\(cost, disambiguations)};
\node[flow, right=of outc, yshift=0mm] (anal) {Analysis\\(Regression \& ML)};

\draw[->] (gen) -- (sens);
\draw[->] (sens) -- (trav);
\draw[->] (trav) -- (outc);
\draw[->] (outc) -- (anal);

\end{tikzpicture}
\caption{
Workflow of the empirical evaluation: 
generate obstacle fields (Strauss, Matérn), assign sensor probabilities, traverse with RD, 
and analyze outcomes via regression and machine learning.}
\label{fig:workflow}
\end{figure}
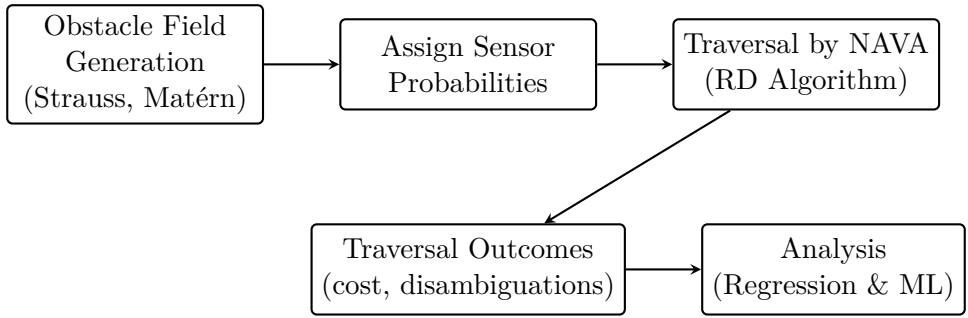

While prior work mainly examined specific layouts or validation scenarios, 
our methodological contribution is threefold.
First, we integrate OOP with stochastic geometry (Strauss, Matérn) 
to evaluate regularity, clustering, and density effects.
Second, we extend OOP to compositional settings—false-only, true-only, 
and mixed—capturing both deceptive clutter and physical blockage.
Finally, we analyze traversal cost distributions using robust regression, 
random forests, and stochastic ordering to rigorously quantify and compare obstacle impacts.


\subsection{Spatial Point Patterns for Obstacle Insertion}
\label{sec:spat-pattern-4-obs}

Obstacle placement is modeled with spatial point processes to assess 
how spatial structure affects traversal cost. 
We analyze two deviations from complete spatial randomness: 
\emph{regularity} via the Strauss process and 
\emph{clustering} via the Matérn cluster process \cite{diggle2003spatial,illian2008statistical}.

For \textbf{regular placement}, 
we use $\text{Strauss}(n, d, \gamma)$, where $n$ is the number of obstacles, $d$ the interaction distance, 
and $\gamma$ the inhibition parameter. 
A \textbf{Strauss} point process is a simple way to generate \emph{regular} (inhibitory) layouts.
It adds a soft “do-not-come-too-close” rule between points within an interaction distance $d$:
the inhibition parameter $\gamma\in[0,1]$ controls how strong that repulsion is 
($\gamma=1$ approximately yielding complete spatial randomness - CSR; 
$\gamma\to 0$ approximately yielding near–hard-core spacing). 
In obstacle terms, Strauss produces evenly spaced mines/roadblocks that blanket a corridor with few large gaps.
Greater regularity is expected to increase traversal cost through more effective corridor coverage.

For \textbf{clustered} placement, 
we use $\text{Matérn}(\kappa, r_0, \mu)$ with parent intensity $\kappa$, cluster radius $r_0$, 
and mean offspring $\mu$. 
A \textbf{Matérn cluster} process produces \emph{aggregated} layouts. 
Parent points occur sparsely (sampled from a Poisson process with intensity $\kappa$); 
each parent generates a Poisson number (with mean $\mu$) of offspring obstacles 
within a cluster radius $r_0$ (parents then discarded). 
The Matérn cluster process generates aggregated patterns with 
pockets of dense obstacles separated by relatively open areas—i.e., 
realistic “debris fields” or localized blockages.
Clustering can leave larger obstacle-free gaps, often decreasing traversal cost.

Our experiments vary $d$, $\gamma$, $r_0$, and $\kappa$, 
as well as the true–false composition ratio $\rho$, to jointly evaluate spatial arrangement and composition.

\subsection{Experimental Setting}
\label{sec:exp-setting}
We consider $\Omega=[0,100]\times[0,100]$ discretized to a $101\times101$ grid, 
yielding an $8$-adjacency graph $G=(V,E)$ with unit and diagonal ($\sqrt{2}$) edges. 
The NAVA starts at $s=(50,100)$ and targets $t=(50,1)$. 
Obstacles are disks of radius $r=4.5$ with centers sampled from the insertion window $\Omega_O=[10,90]\times[10,90]$.
This radius ensures each disk intersects multiple edges without fully blocking the corridor. 
To isolate spatial configuration effects, all disks have equal size (heterogeneous radii are a natural extension). 

Sensor marks follow Beta distributions: 
$p_F\sim\text{Beta}(2,6)$ for false obstacles and $p_T\sim\text{Beta}(6,2)$ for true obstacles. 
Disambiguation incurs a fixed cost $c=5$ \citep{aksakalli:2012,Priebe-NRL:2005}. 
Stronger discrimination can be modeled by $p_F\sim\text{Beta}(a,b)$ and $p_T\sim\text{Beta}(b,a)$ with $a < 2$ and $b > 6$.

\subsection{Key Parameters in the OOP Framework}
Regularity (Strauss $\gamma$) models deliberate spacing (e.g., minefields); 
clustering (Matérn $r_0,\kappa$) captures natural aggregation (e.g., debris); 
and the counts of true/false obstacles drive blockage and disambiguation burden.

Parameter ranges span realistic regimes 
while maintaining feasible traversal.
For Strauss: 
$d$ scales with $r=4.5$ from near overlap ($d\approx r$) to wide separation ($d\gtrsim2r$), 
with $\gamma$ from $0$ (strong inhibition) to $1$ (CSR). 
For Matérn: $r_0\in\{5,\dots,50\}$ and $\kappa\in\{2,\dots,15\}$. 
Sensor-accuracy effects generalize beyond the baseline Beta choices and 
are analyzed via stochastic ordering in Section~\ref{sec:stoch-order}.

\noindent\textbf{Obstacle Placement Strategies:}
\begin{itemize}
    \item \textbf{Regularity (Strauss):} $\gamma\in\{0.0,0.1,\dots,1.0\}$; $d\in\{0.5,1.0,\dots,15.0\}$.
    \item \textbf{Clustering (Matérn):} $\mu=10$ offspring per parent; 
    $\kappa\in\{2,4,\dots,12\}$; $r_0\in\{2.5,5,7.5,10,15,25,50\}$ (larger $r_0$ approaches uniformity).
\end{itemize}

\noindent\textbf{Obstacle Composition:}
We simulate (i) false-only with $n_F\in\{10,20,\dots,100\}$, 
(ii) true-only with $n_T\in\{10,20,\dots,100\}$, 
and (iii) mixed with $n=n_F+n_T\in\{20,30,\dots,100\}$ across varying ratios. 
In mixed settings, true obstacles drive mean cost by blocking edges; 
false obstacles primarily inflate variability via added disambiguations. 
Appendix contains further experiments (30/70, 50/50, 70/30; Strauss and Matérn) corroborate these patterns. 
Section~\ref{sec:stoch-order} formalizes these trends via stochastic ordering.

%

\noindent
\textbf{Analysis Methods:}  
Traversal cost $C(\pi, \X)$ and disambiguation behavior are analyzed via 
(i) \textbf{Robust linear regression} to quantify spatial effects,
(ii) \textbf{Random forest regression} to identify key predictors, 
and (iii) \textbf{Zero-inflated negative binomial regression}  
    for modeling disambiguation counts.
Unlike prior studies that incorporated fixed background clutter \citep{aksakalli:2012},  
our setting removes such clutter  
and allows both true and false obstacle placement by the OPA.  
This design enables a systematic exploration of how spatial structure  
and obstacle composition jointly impact traversal outcomes.

\section{Monte Carlo Experiments and Results}
\label{sec:MCexp-results-and-analysis}

We evaluate how obstacle patterns affect traversal cost via MC simulations using the RD algorithm \citep{aksakalli2011}. 
Obstacle placement follows Strauss (regular) and Matérn (clustered) processes, 
with uniform placement as a baseline \citep{baddeley2010}. 
This tripartite division is rigorous because it
captures the full spectrum of real-world scenarios, from pure decoy placement
to pure obstruction to realistic mixtures of both. Additional variations for
the mixed case, including different ratios of true to false obstacles and varied
clustering strengths, are included in the Appendix.

\begin{figure}[!ht]
\centering
\begin{minipage}{0.5\linewidth}
    \centering
    \includegraphics[width=\linewidth]{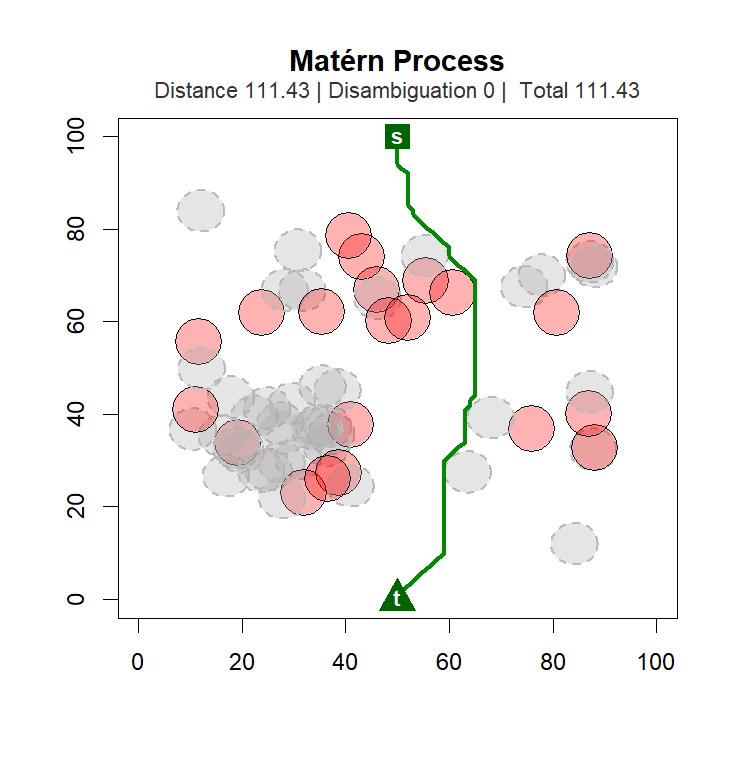}\\
    (a)
\end{minipage}\hfill
\begin{minipage}{0.5\linewidth}
    \centering
    \includegraphics[width=\linewidth]{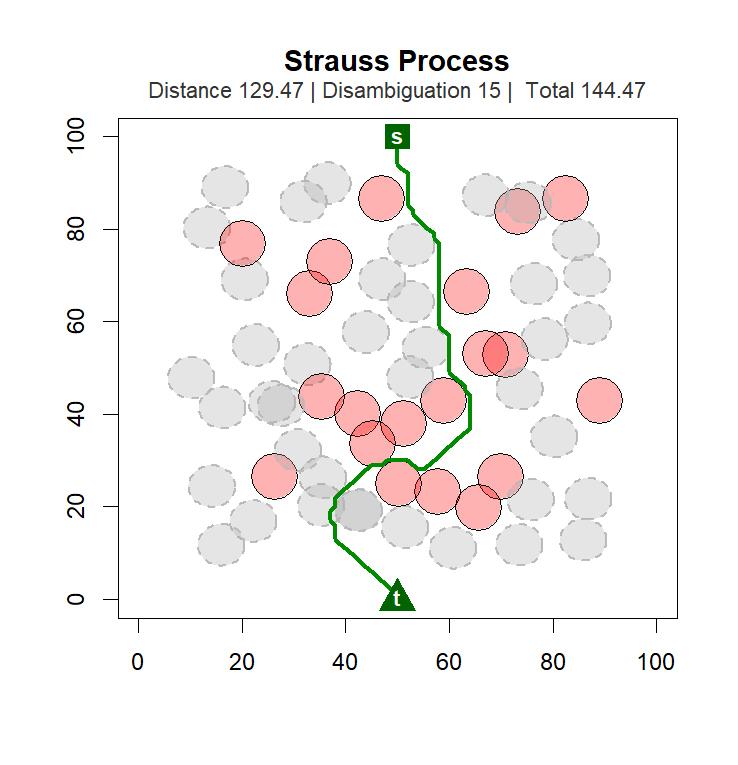}\\
    (b)
\end{minipage}
\caption{Illustrative traversals under different obstacle patterns: 
(a) Strauss (regular),
(b) Matérn (clustered). 
Red = true obstacles, dashed = false obstacles, blue = RD traversal path.
Distance: Euclidean distance, Disambiguation: Cost of disambiguation, and Total: total cost.}
\label{fig:paths}
\end{figure}

\subsection{Obstacle Pattern: Uniformity to Regularity (Strauss) - False-Obstacle Only Case}
We evaluate the effect of increasing spatial regularity on traversal outcomes by varying 
the Strauss$(n,d,\gamma)$ process parameters: $\gamma$ (repulsion strength) and $d$ (interaction distance). 

To ensure comprehensive coverage, we simulate 30 values of $\gamma$ and 22 values of $d$ 
for each selected number of false obstacles $n_F \in \{10, 20, \dots, 100\}$, 
with 100 MC replications per parameter setting. 
This results in $11 \times 30 \times 22 \times 100 = 726{,}000$ total trials. 
For each realization, we record the total traversal cost $C$ and the number of disambiguations incurred. 
The RD algorithm adaptively recomputes the path after each disambiguation based on updated obstacle information. 
A representative simulation outcome is displayed in Figure~\ref{fig:sample-Strauss-false-mix}(a).

\begin{figure}[htb]
\centering
\begin{minipage}{0.45\textwidth}
    \centering
    \includegraphics[trim=10mm 15mm 10mm 15mm, clip,height=5.1cm,width=\linewidth]{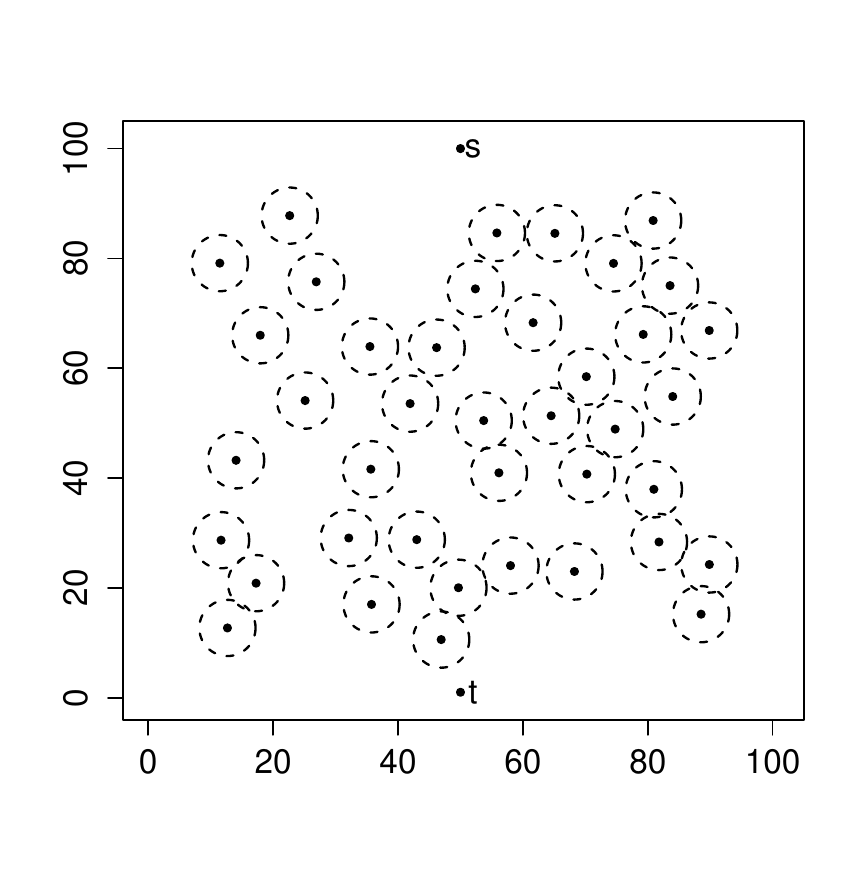}\\
    (a)
\end{minipage}\hfill
\begin{minipage}{0.45\textwidth}
    \centering
    \includegraphics[trim=10mm 15mm 10mm 15mm, clip,height=5.1cm,width=\linewidth]{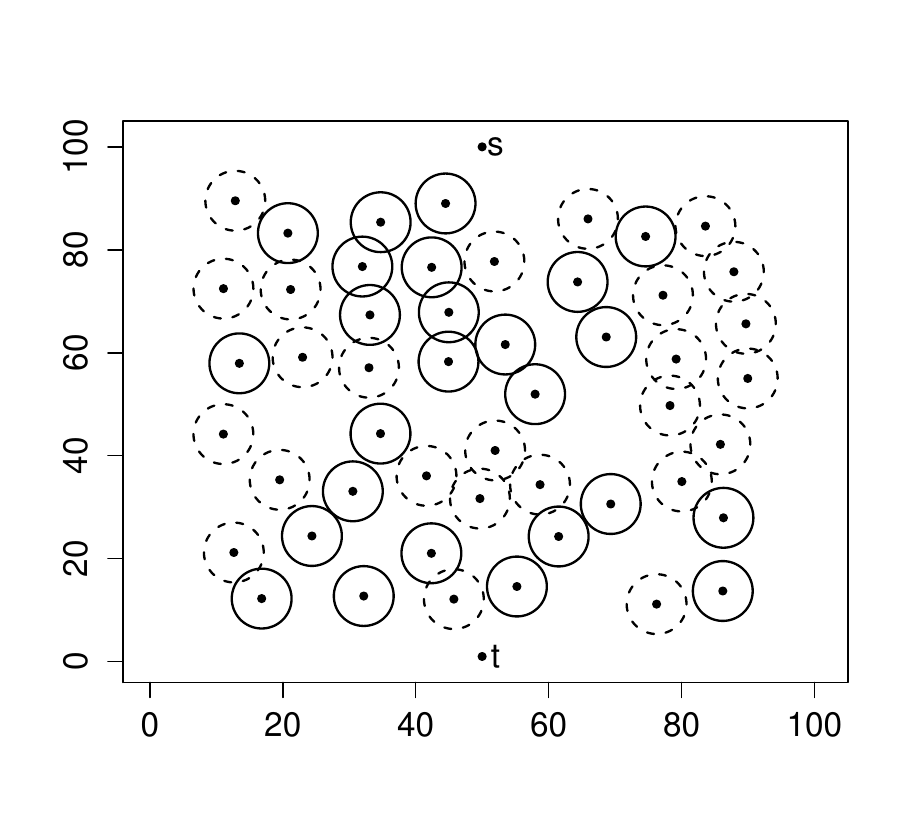}\\
    (b)
\end{minipage}
\caption{
(a) False obstacles from Strauss$(n_F=40,d=9,\gamma=0)$.  
(b) Mixed obstacles from Strauss$(n=50,d=9,\gamma=0)$ with 25 false (dashed) and 25 true (solid) obstacles.}
\label{fig:sample-Strauss-false-mix}
\end{figure}

Figure~\ref{fig:Lbar-vs-Strauss}(a) illustrates the mean traversal cost $\bar{C}$  
as a function of the Strauss inhibition parameter $\gamma$,  
across various interaction distances $d$.  
The correlation between $\bar{C}$ and $\gamma$ for each $d$ value  
is shown in Figure~\ref{fig:Lbar-vs-Strauss}(b).

\begin{figure}[!ht]
\centering
\begin{minipage}{0.45\textwidth}
    \centering
    \includegraphics[width=\linewidth]{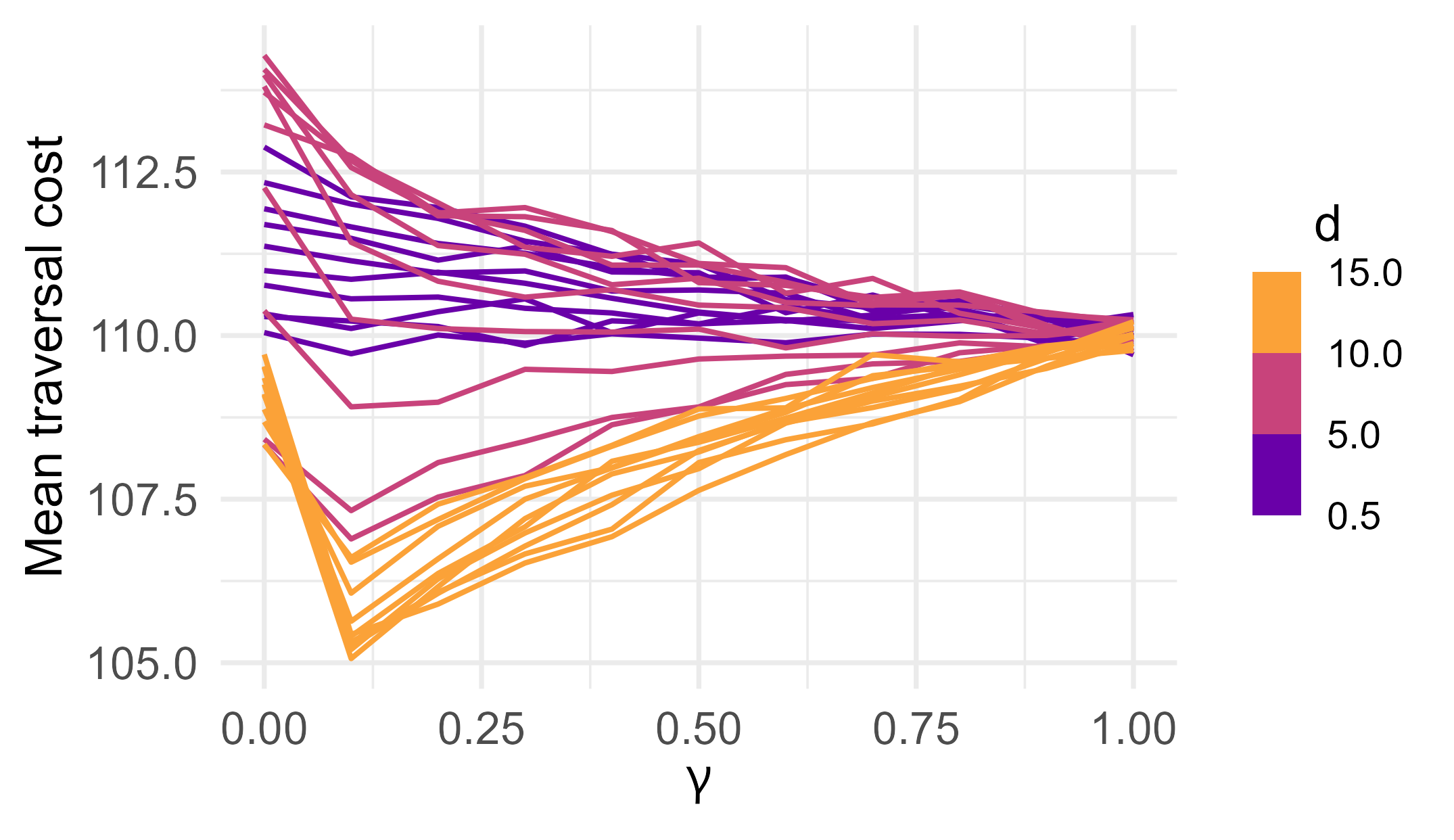}\\
    (a)
\end{minipage}\hfill
\begin{minipage}{0.45\textwidth}
    \centering
    \includegraphics[width=\linewidth]{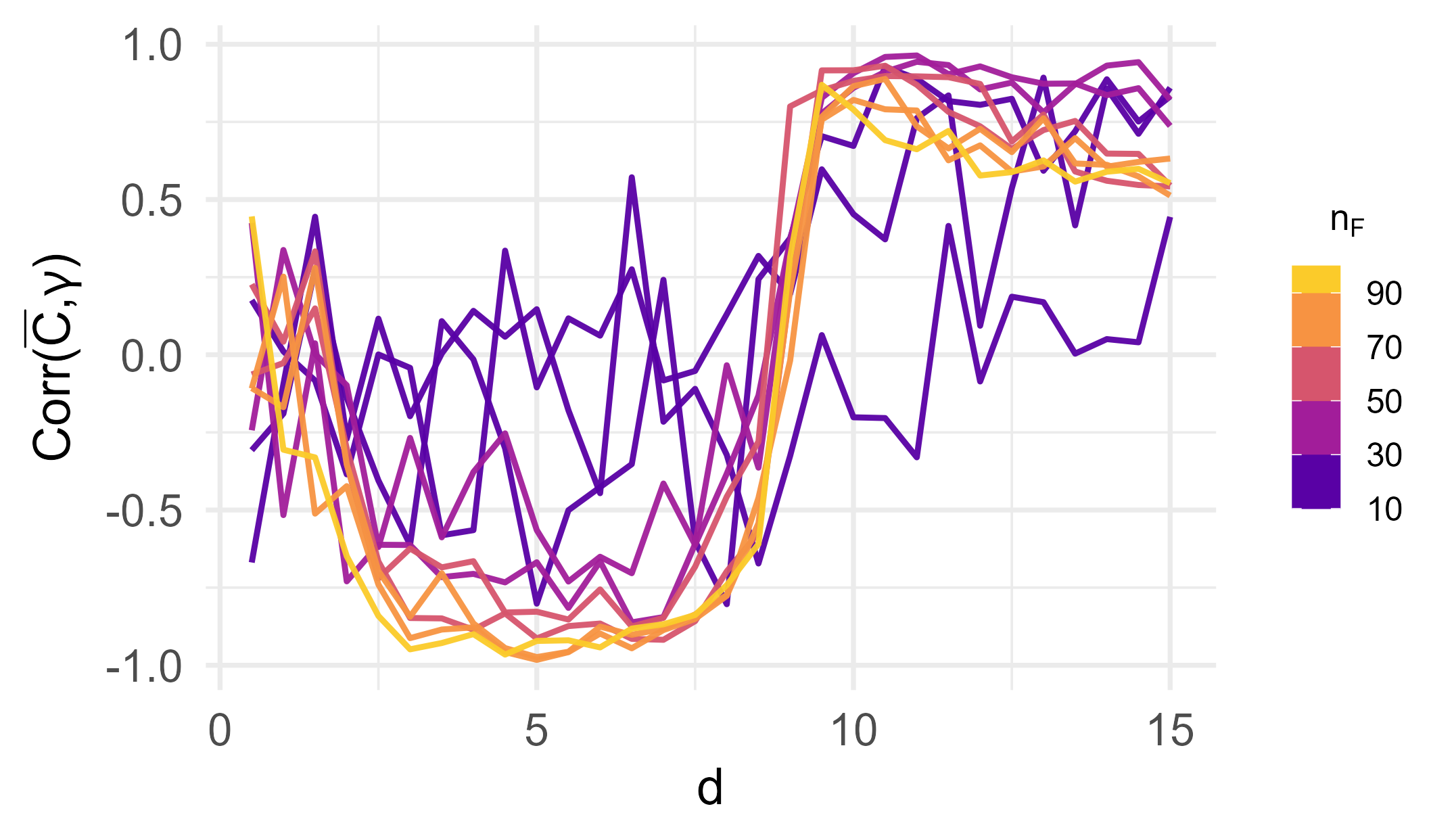}\\
    (b)
\end{minipage}
\caption{
(a) Interaction plot for the false-obstacle-only case with
the mean traversal cost $\bar{C}$ (averaged over obstacle numbers) vs.\ $\gamma$,  
for varying $d$ values under the Strauss$(n,d,\gamma)$ regularity pattern.  
(b) Correlation $\Corr(\bar{C},\gamma)$ vs.\ $d$.}
\label{fig:Lbar-vs-Strauss}
\end{figure}

For small $d$, increased regularity (i.e., decreasing $\gamma$) has little impact,  
as obstacles remain closely spaced and overlapping.  
Consequently, traversal paths do not change significantly.  
In contrast, for intermediate $d$ (around $1.5r$), regularity promotes even spacing,  
which effectively blocks direct traversal routes and increases cost.  
When $d$ exceeds $2r$, obstacles are spaced too widely to obstruct paths effectively,  
leading to lower traversal cost.

Figure~\ref{fig:Lbarvsd-at-gamma} displays the mean traversal cost $\bar{C}$  
as a function of $d$ across various $\gamma$ values.  
A unimodal (concave-down) pattern emerges:  
traversal cost peaks around $d \approx 1.5r$,  
where obstacle spacing is most disruptive.  
For small $d$, overlapping obstacles act as a single obstruction zone,  
and for large $d$, the configuration becomes too sparse  
to significantly hinder navigation.

\begin{figure} [htb]
\centering
\includegraphics[width=0.65\textwidth]{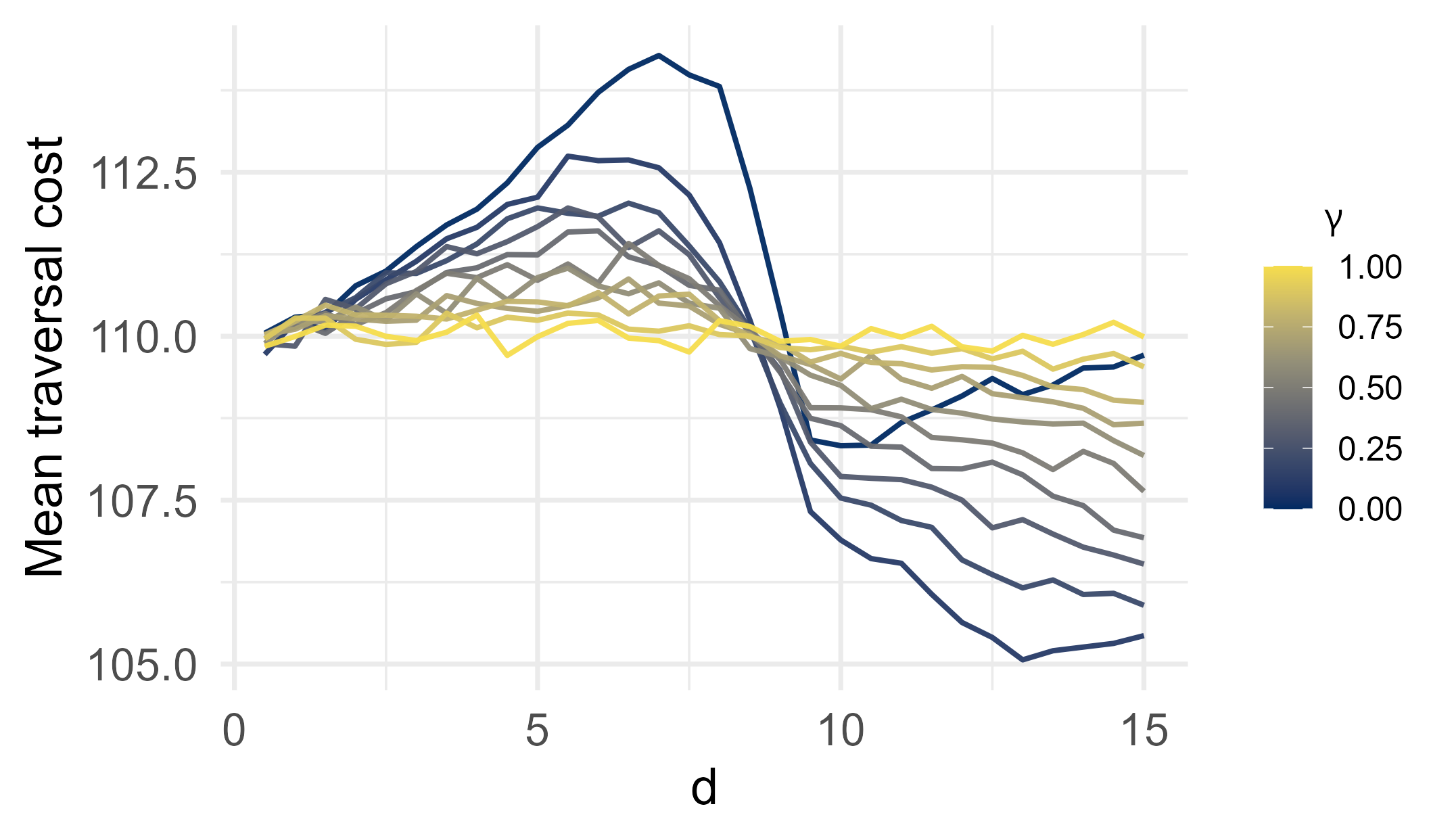}
\caption{Interaction Plot for the false obstacle only case with
the mean traversal cost $\bar{C}$ (averaged over obstacle numbers) vs. interaction distance $d$ values are plotted
for various $\gamma$ values under Strauss$(n,d,\gamma)$ regularity pattern.}
\label{fig:Lbarvsd-at-gamma}
\end{figure}

Figure~\ref{fig:contour-Lbarvsd-gamma} shows a filled contour plot  
of mean traversal cost $\bar{C}$ over the $\gamma$–$d$ space.  
The plot confirms earlier trends:  
traversal cost peaks when regularity is high (small $\gamma$)  
and spacing is moderate ($d \approx 1.5r$).  
For larger $d$ values ($\gtrsim 2r$), $\bar{C}$ becomes nearly insensitive to $\gamma$,  
effectively corresponding to uniform placement.

\begin{figure} [!htb]
\centering
\includegraphics[width=.65\textwidth]{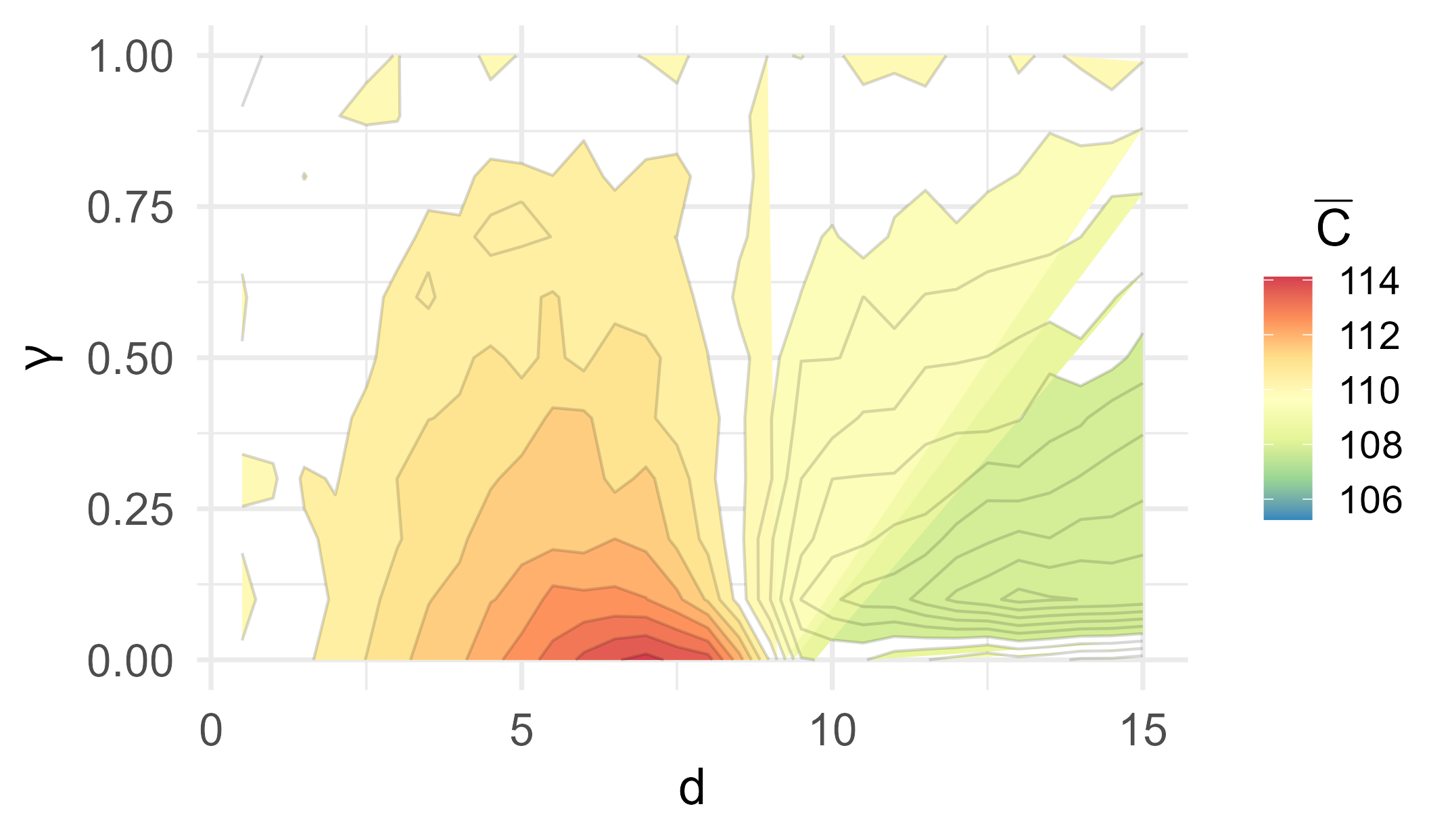}
\caption{Filled contour plot of mean traversal cost $\overline C$  (averaged over obstacle numbers) in the false obstacle case
 for $\gamma$ and $d$ values
under Strauss$(n,d,\gamma)$ regularity pattern.}
\label{fig:contour-Lbarvsd-gamma}
\end{figure}

Below are our \textbf{key findings for false obstacle only case.}
To maximize traversal cost when placing false obstacles, the OPA should:
\begin{itemize}
    \item Select moderate interaction distances ($1.3r \leq d \leq 1.8r$).
    \item Use low inhibition values ($\gamma < 0.1$) to induce strong regularity.
    \item Increase obstacle count $n_F$ to saturate the corridor.
\end{itemize}

\begin{remark}[True-Obstacle Only and Mixed-Obstacle Cases Under Regularity:]
We perform similar MC experiments under true obstacles only and mixed obstacle cases,  
and below we summarize our findings.  
See the Appendix for more details.
\end{remark}

\subsubsection{Summary of Findings for All Obstacle Compositions under Regularity}
Below we summarize the findings for all obstacle compositions under regularity cases:
\begin{itemize}
    \item \textbf{False Obstacles:} Traversal cost peaks for $d \in (6,8)$ and $\gamma < 0.1$.  
    For large $d$, regularity matters less,  
    and $\gamma \approx 1$ (uniformity) can be more effective.
    
    \item \textbf{True Obstacles:} Higher density saturates the environment;  
    regularity has limited added effect,  
    but moderate $d$ and low $\gamma$ still maximize cost.
    
    \item \textbf{Mixed Obstacles:} Increasing $n_T$ raises disambiguation frequency  
    and causes more resets, significantly raising traversal cost.
\end{itemize}

\subsection{Obstacle Pattern: Uniformity to Clustering (Matérn) - False-Obstacle Only Case}
In the false obstacle only case,  
we examine how clustering affects traversal cost  
using the Matérn process with varying parameters $\kappa$ (number of clusters),  
$r_0$ (cluster radius), and $\mu$ (mean obstacles per cluster).  
We vary $\kappa \in \{1, \dots, 10\}$, $r_0 \in \{2.5, 5, \dots, 25\}$,  
and $n_F \in \{10, 20, \dots, 100\}$.  
Each configuration is simulated with 100 MC replications,  
resulting in $11 \times 30 \times 100 = 33{,}000$ measurements per $n_F$.  
A sample realization is shown in Figure~\ref{fig:sample-Matern-false}.

\begin{figure} [!ht]
\centering
\includegraphics[width=.5\textwidth]{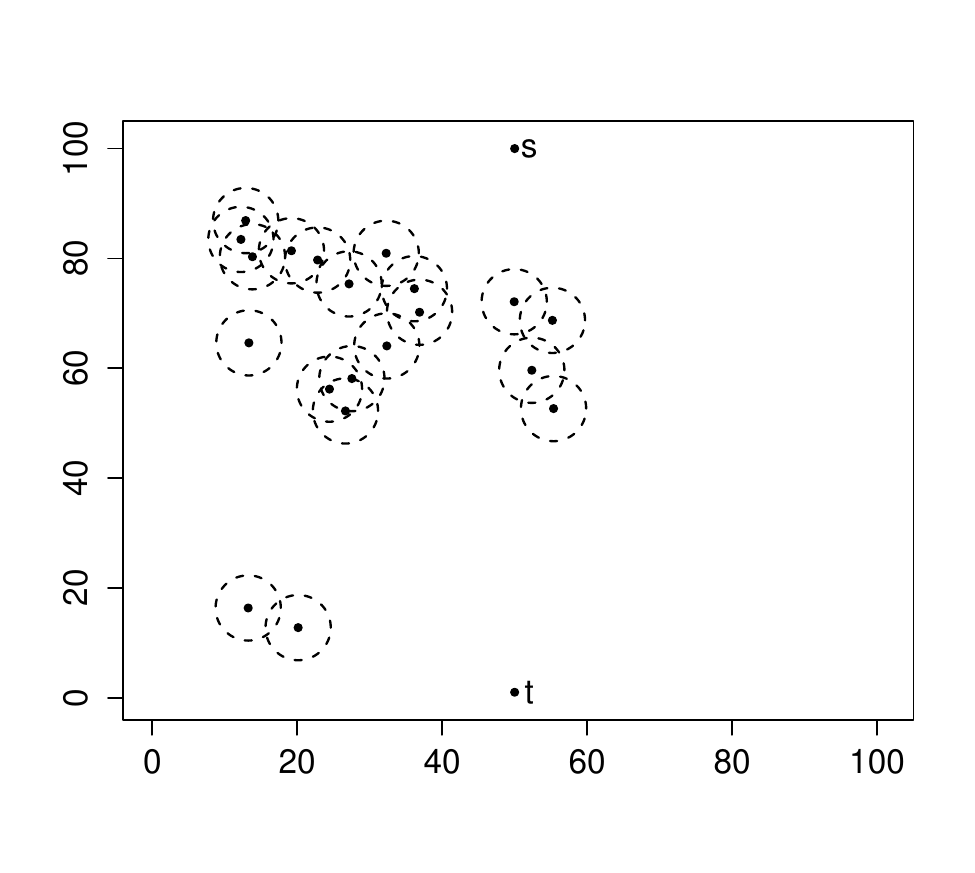}
\caption{Example realization of clustered false obstacles  
generated from Matérn$(\kappa=2,r_0=15,\mu=10)$.}
\label{fig:sample-Matern-false}
\end{figure}

Figure~\ref{fig:L-vs-Matern}(a) shows how mean traversal cost $\bar{C}$  
varies with $r_0$ for different $\kappa$ values.  
For $r_0 \gtrsim 15$, $\bar{C}$ levels off,  
as clustering no longer meaningfully alters the path.  
However, for tighter clusters ($r_0 \lesssim 15$),  
cost drops due to larger obstacle-free gaps.  
This effect intensifies with larger $n_F$.  
For fixed $r_0$, increasing $\kappa$ disperses obstacles,  
raising traversal cost by increasing obstruction in the navigation region.

\begin{figure}[!ht]
\centering
\begin{minipage}{0.45\textwidth}
    \centering
    \includegraphics[width=\linewidth]{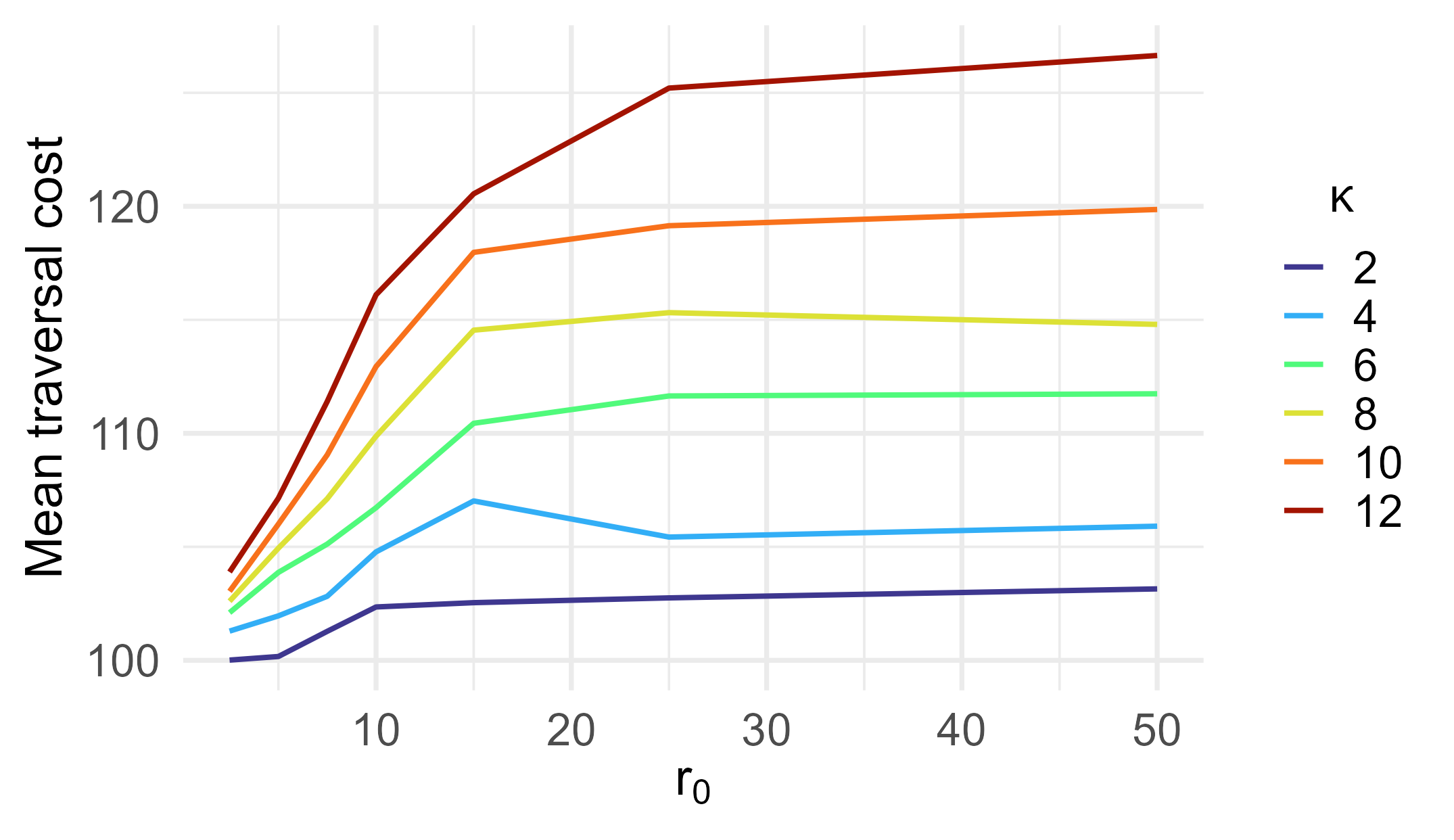}\\
    (a)
\end{minipage}\hfill
\begin{minipage}{0.45\textwidth}
    \centering
    \includegraphics[width=\linewidth]{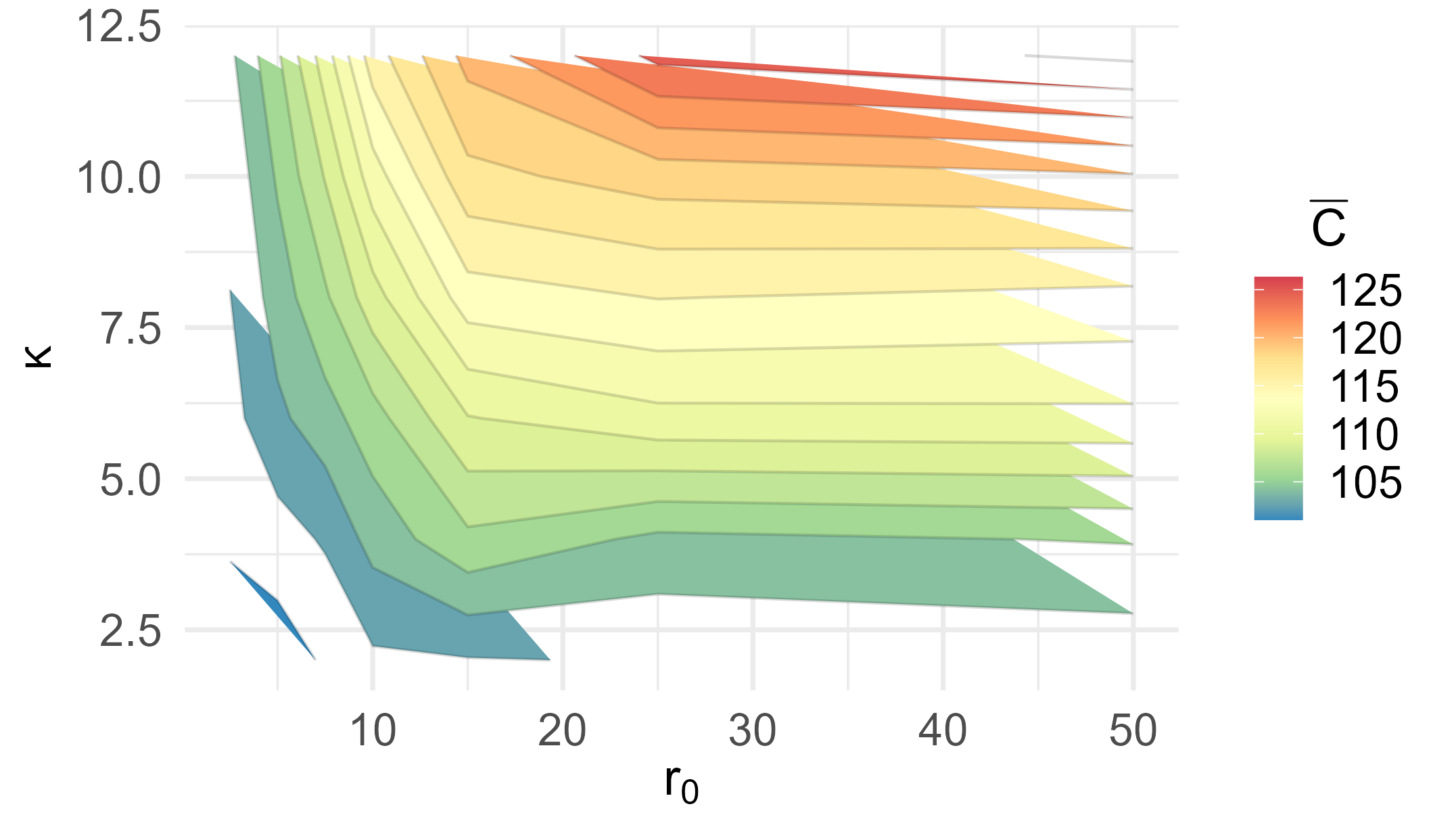}\\
    (b)
\end{minipage}
\caption{
(a) Interaction plot for the false-obstacle-only case with
the mean traversal cost $\bar{C}$ (averaged over obstacle numbers) vs.\ radius $r_0$, 
for varying $\kappa$ values under the Matérn$(\kappa,r_0,\mu)$ clustering pattern.  
(b) Contour plot of $\bar{C}$ as a function of $\kappa$ and $r_0$.}
\label{fig:L-vs-Matern}
\end{figure}

Figure~\ref{fig:L-vs-Matern}(b) indicates that traversal cost is highest  
at $\kappa \approx 12$ and $r_0 \approx 50$,  
where obstacles are most spatially dispersed.  
To maximize traversal difficulty in the clustered setting, the OPA should:
\begin{itemize}
    \item Use a high number of clusters ($\kappa$) to spread obstacles widely.
    \item Select $r_0 \gtrsim 15$ to avoid tight clustering and minimize wide open gaps.
    \item Employ a sufficiently large number of obstacles ($n_F$)  
    to fill the region effectively.
\end{itemize}

\begin{remark}[True-Obstacle Only and Mixed-Obstacle Cases Under Clustering:]
We perform similar MC experiments under true obstacles only and mixed obstacle cases,  
and below we summarize our findings.  
See the Appendix for more details.
\end{remark}

\subsubsection{Summary of Findings for All Obstacle Compositions Under Clustering.}
\begin{itemize}
    \item \textbf{False Obstacles:} Widely spaced clusters with large $\kappa$ and $r_0$  
    hinder traversal more effectively than tight clusters.

    \item \textbf{True Obstacles:} Increased spread heightens cost further  
    due to unavoidable disambiguation events.

    \item \textbf{Mixed Obstacles:} Greater proportion of true obstacles increases resets,  
    compounding the traversal burden.
\end{itemize}

Table~\ref{tab:optimal-strategies} summarizes optimal placement strategies (to maximize the traversal cost of NAVA)
under both Strauss and Matérn settings, across all obstacle types.

\begin{table}[ht]
    \centering
    \caption{Optimal Obstacle Placement Strategies for Increasing the Traversal Cost.  
    Scenario ``False Only" and ``True Only" refer to environments  
    with only false or only true obstacles, respectively;  
    ``Mixed" includes both types.}
    \label{tab:optimal-strategies}
    \begin{tabular}{lcc}
        \toprule
        \textbf{Scenario} & \textbf{Regularity (Strauss Process)} & \textbf{Clustering (Matérn Process)} \\
        \midrule
        False Only & Moderate \(d \in (6,8)\), Low \(\gamma < 0.1\) & Larger \(r_0\), High \(\kappa\) \\
        True Only  & Moderate \(d \in (6,8)\), Low \(\gamma < 0.4\) & Larger \(r_0\), High \(\kappa\) \\
        Mixed      & Same as false obs., prefer \(n_T > n_F\) & Same as true obs., prefer \(n_T > n_F\) \\
        \bottomrule
    \end{tabular}
\end{table}

\subsection{Dependence of Traversal on Obstacle Size}
\label{sec:dependence-obs-size}
We conducted a traversal comparison between obstacles of fixed size and heterogeneous radii. 
Using the same Strauss process parameters ($\gamma=0.1, d=7$) and obstacle count (n=50 with 30 false obstacles), 
we compared fixed radius, $r=4.5$, for all obstacles, against varying radii, $r\in\{3, 4.5, 6, 7.5\}$, 
randomly assigned with corresponding disambiguation costs of $\{3, 5, 7, 9\}$. 
The disambiguation cost is set in a way that disambiguation is encouraged 
when the probability of an obstacle being true is moderate or small. 
The setting with obstacles of different radii resulted in a total traversal cost of 114.63, 
while fixed-radius case led to cost of 109.77. 
Although using the same number of obstacles with similar placement patterns, 
the radius-varying obstacles created more complex edge intersection patterns, 
particularly where larger obstacles blocked critical path segments. 
As demonstrated in previous analysis, the Strauss process proves effective at creating traversal disruption. 
The additional impact observed with heterogeneous obstacle sizes suggests that 
real-world environments with varying obstacle sizes may further 
increase the disruptive effects of spatial point process, 
which presents as a future research direction. 

\begin{figure}[!ht]
\centering
\begin{minipage}{0.48\textwidth}
    \centering
    \includegraphics[width=\textwidth]{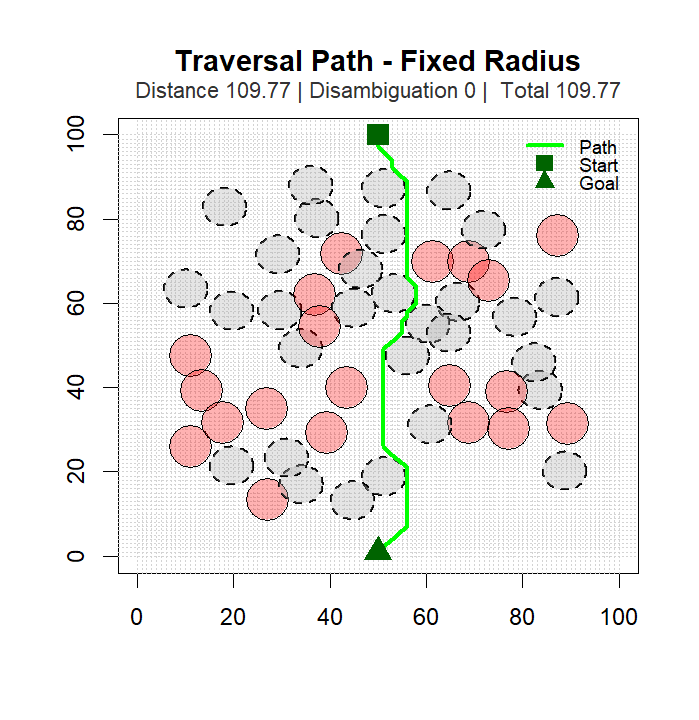}\\
    (a)
\end{minipage}\hfill
\begin{minipage}{0.48\textwidth}
    \centering
    \includegraphics[width=\textwidth]{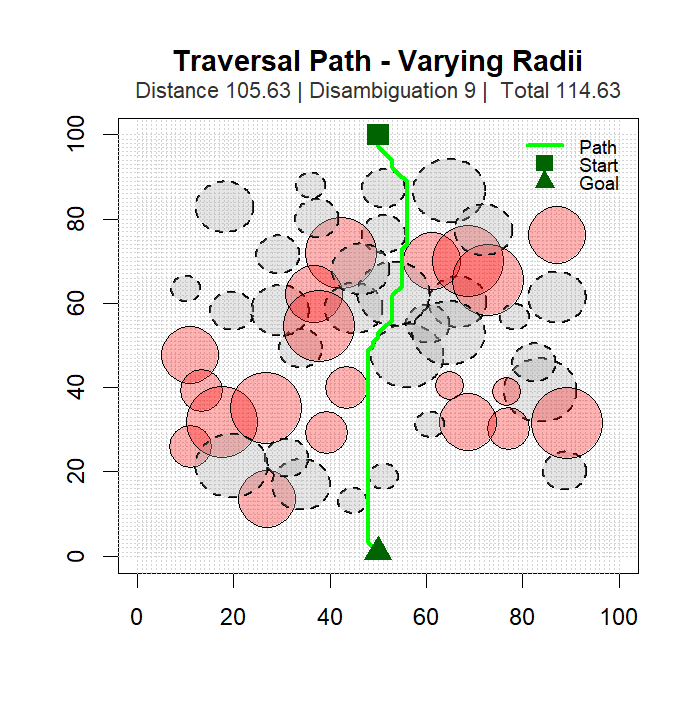}\\
    (b)
\end{minipage}
\caption{Comparison of obstacle size heterogeneity effects using Strauss process ($ \gamma = 0.1, d = 7, n = 50$). 
(a) Fixed obstacle size: total cost 109.77. (b) Heterogeneous obstacle sizes: total cost 114.63.
Distance: Euclidean distance, Disambiguation: Cost of disambiguation, and Total: total cost.}
\label{fig:radius_comparison}
\end{figure}

\begin{remark}
We did not include variable obstacle size as a systematic factor in the main set of experiments
or in the regression analysis. In operational obstacle placement applications,
inserted obstacles are usually of uniform size, so focusing on equal-radius settings
better reflects the intended deployment scenario. Moreover, keeping obstacle size fixed
avoids confounding interactions with other variables of interest such as clustering,
regularity, and the ratio of true to false obstacles, thereby isolating the effects
of these primary drivers on traversal cost.
\end{remark}

\subsection{Regression Models and Predictor Importance}
\label{sec:regression-results}

To quantify how spatial regularity parameters $(d, \gamma)$ and obstacle count $n_F$  
influence traversal cost,  
we fit robust linear regression models to log-transformed cost data  
using $M$-estimation with Huber weights \citep{Huber:1981},  
implemented via \texttt{rlm} in the \texttt{MASS} package in \texttt{R}.  
We choose this modeling approach due to the right-skewness and existence of outliers in the traversal costs.

\subsubsection{Robust Regression Model}
\label{sec:robust-reg}

We model the log-transformed traversal cost ($C$)  
as a second-order polynomial with interaction terms:
\begin{equation}
\widehat{C} = \beta_0 + \beta_i \text{ linear terms} + \beta_{j} \text{ quadratic terms}  
+ \beta_{k} \text{ two-way interactions}.
\end{equation}

In the \textbf{false-obstacle-only case with Strauss regularity},  
the final model (after dropping the non-significant $\gamma n_F$ interaction) is:
\begin{equation}
\label{eqn:rlm-false-regular}
\widehat C = 99.270 -4.21 \gamma + 0.33 d + 0.1897 n_F + 1.68 \gamma^2  
- 0.025 d^2 + 0.0006 n_F^2 + 0.365 \gamma d - 0.0064 d n_F.
\end{equation}

The coefficients indicate the following:
(i) \textbf{$\gamma$ effect:} Decreases cost for small $d$, increases it for larger $d$;
a stronger interaction parameter lowers traversal cost when obstacles are moderately spaced but raises it when they are farther apart.
(ii) \textbf{$d$ effect:} Concave-down (unimodal);
the influence of obstacle spacing on traversal cost follows a concave-down pattern, with costs peaking at an intermediate distance. 
(iii) \textbf{$n_F$ effect:} Concave-up growth;
increasing the number of false obstacles raises traversal cost in a concave-up manner, with the effect becoming stronger as more obstacles are added.
The residual standard error improves from 4.48 to 3.78 with robust regression,  
confirming its effectiveness.
Here are one-sentence interpretations for each case, written in your style:

Similar models were fitted for the true-only and mixed-obstacle cases  
under both Strauss and Matérn patterns.  
See the Appendix for detailed coefficients.  
A summary of selected models appears in Table~\ref{tab:rlm-summary}.

\begin{table}[htb]
\centering
\caption{Summary of Robust Linear Models for Traversal Cost (only significant terms are retained).}
\label{tab:rlm-summary}
\begin{tabular}{lccc}
\toprule
\textbf{Setting} & \textbf{Main Effects} & \textbf{Quadratic Terms} & \textbf{Interactions} \\
\midrule
False Obstacles (Regular) & $\gamma, d, n_F$ & $\gamma^2, d^2, n_F^2$ & $\gamma d, d n_F$ \\
True Obstacles (Regular) & $\gamma, d, n_T$ & $\gamma^2, d^2, n_T^2$ & $\gamma d, d n_T$ \\
Mixed Obstacles (Regular) & $\gamma, d, n_F, n_o$ & $\gamma^2, d^2, n_F^2, n_o^2$  
& $\gamma d, d n_o, \gamma n_F$ \\
\hline
False Obstacles (Clustered) & $\kappa, r_0$ & $r_0^2$ & $\kappa r_0$ \\
True Obstacles (Clustered) & $\kappa, r_0$ & $r_0^2$ & $\kappa r_0$ \\
Mixed Obstacles (Clustered) & $\kappa, r_0, n_F$  
& $\kappa^2, r_0^2, n_F^2$  
& $\kappa r_0, \kappa n_F, r_0 n_F$ \\
\bottomrule
\end{tabular}
\end{table}

\subsubsection{Random Forest Regression}
\label{sec:rf-reg}
To assess the relative importance of predictors influencing traversal cost,  
we apply Random Forest (RF) regression \citep{breiman2001} with 100 trees,  
implemented using the \texttt{randomForest} package in \texttt{R}.  
This nonparametric method complements the robust regression models  
by identifying nonlinear interactions and ranking predictors  
by their contribution to variance reduction.

In the false-obstacle-only case with Matérn clustering,  
the most influential variables are:  
$n_F$ (highest impact), $r_0$, and $\kappa$ (lowest impact).  
While the RF model explains 66.86\% of the variance,  
its mean squared residual is 19.04,  
indicating limited predictive accuracy for exact cost values.  
However, it provides valuable insight into predictor influence.  
Figure~\ref{fig:RF-varimp-plot-false-reg}(a)  
visualizes the variable importance rankings.

\begin{figure}[htb]
\centering
\begin{minipage}{0.45\textwidth}
    \centering
    \includegraphics[width=\linewidth]{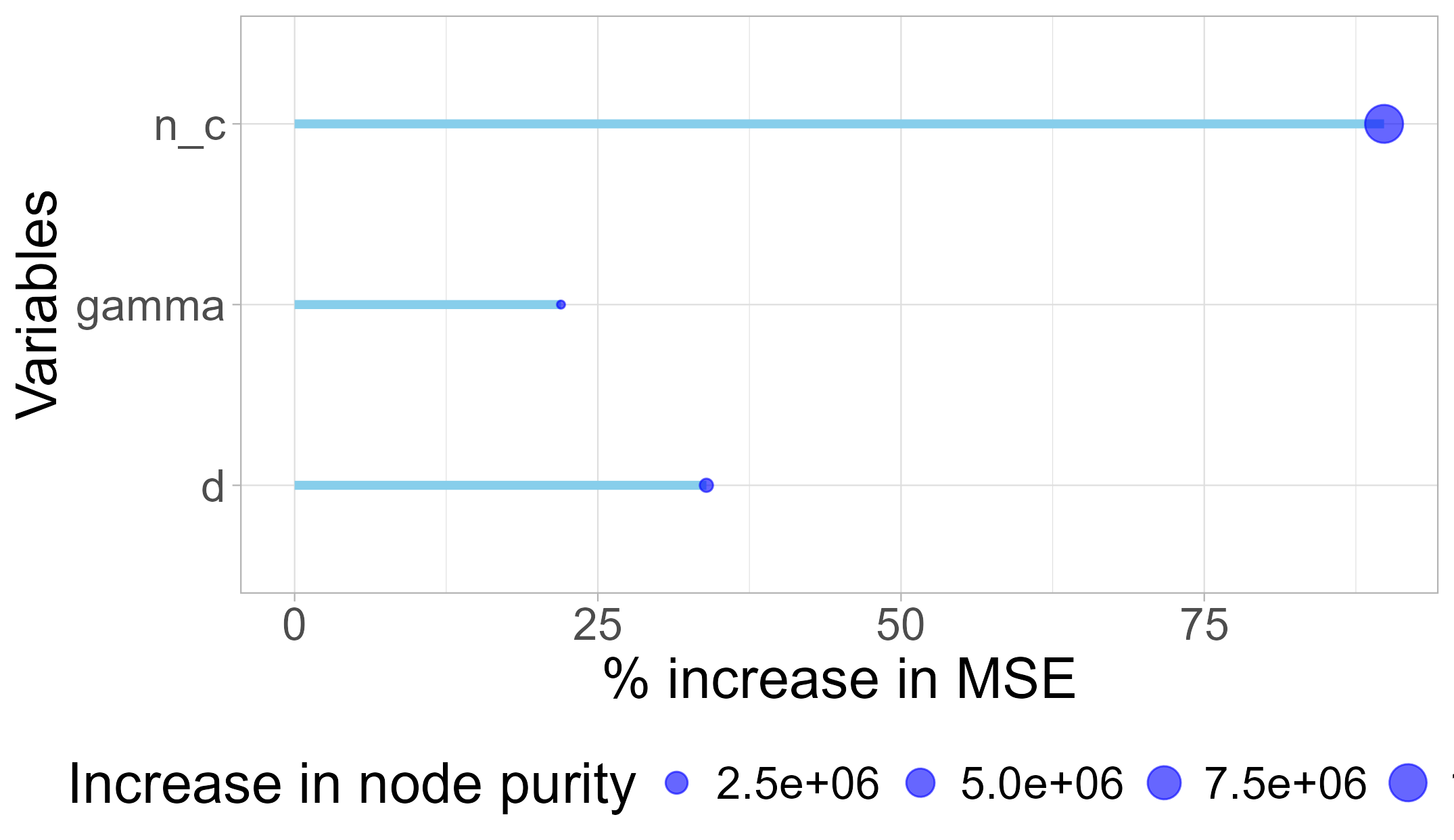}\\
    (a)
\end{minipage}\hfill
\begin{minipage}{0.45\textwidth}
    \centering
    \includegraphics[width=\linewidth]{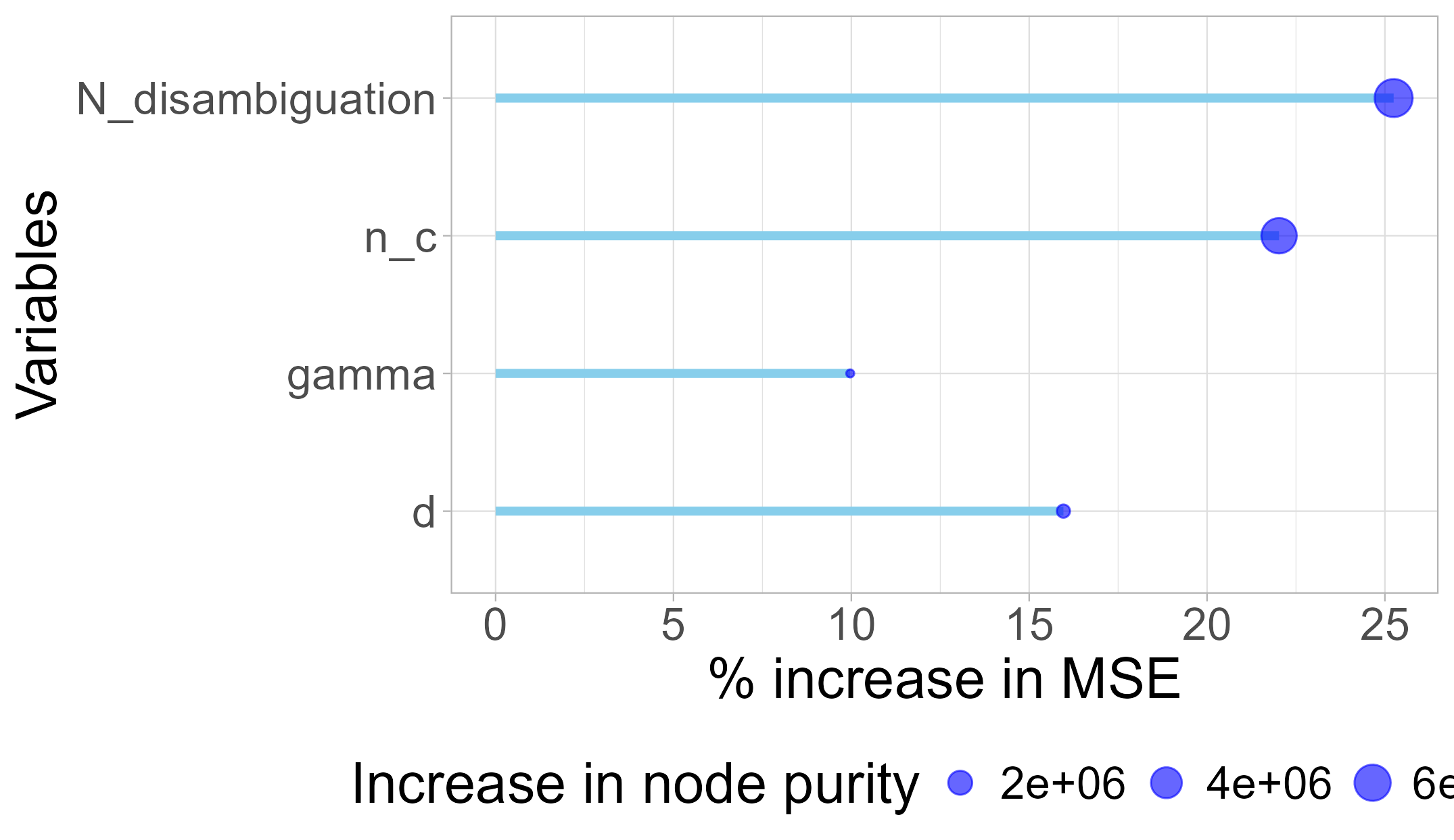}\\
    (b)
\end{minipage}
\caption{Variable importance (decrease in mean square error) and node purity (residual sum of squares)  
for RF regression models with $C$ as the response.  
(a) Using $n_F$, $d$, $\gamma$ as predictors.  
(b) Using $n_F$, $d$, $\gamma$, and $N_{dis}$ (number of disambiguations) as predictors.}
\label{fig:RF-varimp-plot-false-reg}
\end{figure}

RF-based variable importance results for other obstacle compositions  
(true-only and mixed) under both Strauss and Matérn settings  
are qualitatively consistent.  
See the Appendix for detailed plots and model diagnostics.  
Table~\ref{tab:rf-summary} summarizes the top-ranked predictors for each setting.

\begin{table}[htb]
\centering
\caption{Most significant predictors ranked by RF variable importance 
(with decreasing importance from left to right).}
\label{tab:rf-summary}
\begin{tabular}{ll}
\toprule
\textbf{Scenario} & \textbf{Key Predictors} \\
\midrule
False Obstacles (Regular) & $n_F, d, \gamma$ \\
True Obstacles (Regular) & $n_T, d, \gamma$ \\
Mixed Obstacles (Regular) & $n_o, n_F, d$ \\
\hline
False Obstacles (Clustered) & $r_0, \kappa$ \\
True Obstacles (Clustered) & $r_0, \kappa$ \\
Mixed Obstacles (Clustered) & $n_F, r_0, \kappa$ \\
\bottomrule
\end{tabular}
\end{table}

\subsubsection{Modeling Number of Disambiguations}
\label{sec:zinb-reg}

Disambiguations ($N_{dis}$) play a central role in traversal cost,  
particularly in settings with true or mixed obstacles.  
To better understand the factors influencing $N_{dis}$,  
we fit Zero-Inflated Negative Binomial (ZINB) models \citep{Zeileis2008}  
using the \texttt{zeroinfl} function in \texttt{R}. 
Here,
the ZINB specification is appropriate because disambiguation counts are discrete, 
highly overdispersed relative to a Poisson model, 
and include an excess of zeros corresponding to obstacle-free traversals. 

In the false-obstacle-only scenario under Strauss regularity,  
we include $\gamma$ and $d$ as predictors in the count model  
and $n_F$ in the zero-inflation model.  
The fitted model yields
(i) $\gamma$ coefficient: $-0.062$ \quad (higher regularity increases disambiguations),
(ii) $d$ coefficient: $-0.054$ \quad (greater spacing reduces disambiguations), and
(iii) $n_F$ in the logit model: negative effect on probability of zero disambiguations.

This indicates that regular spacing of obstacles increases disambiguation frequency,  
while increasing obstacle count makes disambiguation events more likely.

Similar patterns hold across other obstacle types:
\begin{itemize}
    \item In both true-only and mixed scenarios, $N_{dis}$ decreases with $\gamma$ and $d$,  
    and increases with $n_T$ or $n_F$.
    
    \item In Matérn-clustered layouts, $r_0$ significantly affects disambiguations  
    (larger $r_0$ reduces $N_{dis}$),  
    while $\kappa$ shows limited influence.
\end{itemize}

See the Appendix for full model summaries.

\subsection{Summary of Recommendations from Monte Carlo Analysis}

\subsubsection{Recommendations for OPA (Adversarial Obstacle Placement Perspective)}
This work is written from the angle of OPA to maximize traversal cost for NAVA.
Based on our finding, we recommend OPA to do the following.
\begin{itemize}
    \item \textbf{Strauss Regularity:} Place obstacles with high regularity (low $\gamma$)  
    and moderate spacing ($d \in (6,8)$) to create evenly distributed barriers that blanket the corridor.
    \item \textbf{Matérn Clustering:} Favor dispersed configurations with large $r_0$ and high $\kappa$  
    to cover more area and reduce the chance of long, unobstructed corridors.
    \item \textbf{Mixed Obstacles:} Prioritize true obstacles over false ($n_T > n_F$)  
    since they force resets and disambiguations, amplifying traversal cost.
\end{itemize}

\subsubsection{Recommendations for NAVA (Defensive/Traversal Perspective)}
Our suggestions would be different for NAVA (to minimize her traversal cost), not necessarily the opposite of recommendations to OPA.
The main reason is that their capabilities are different, with NAVA having no obstacle insertion capability,
while OPA lacking a sensor to guide his insertion schemes (or to predict NAVA's traversal better).
Based on our finding, we recommend NAVA to do the following.
\begin{itemize}
    \item \textbf{Strauss Regularity:} When obstacle spacing is conspicuously uniform,  
    interpret it as adversarial placement. Counter by routing toward the periphery of the insertion window,  
    probing selectively at high–leverage obstacles, and avoiding blanket central zones.
    \item \textbf{Matérn Clustering:} Clusters often leave navigable inter–cluster corridors.  
    Counter by scouting for these low–density seams, delaying disambiguations until bottlenecks are reached,  
    and aligning paths along the sparse axis of elongated clusters.
    \item \textbf{Mixed Obstacles:} Early probes help infer the $n_T/n_F$ ratio.  
    If false obstacles dominate, thread the corridor with minimal probing;  
    if true obstacles dominate, pre–commit to wider detours with fewer but strategically placed disambiguations.
\end{itemize}

From the OPA’s side, deliberate regularity or dense clustering maximizes cost.  
From the NAVA’s side, diagnosing these patterns is critical:  
regularity $\Rightarrow$ treat as adversarial and route wide;  
clustering $\Rightarrow$ exploit inter–cluster corridors for lower–cost traversal.  
In practice, adversarial recognition (regular layouts) signals the need for caution and probing economy,  
while clustered layouts indicate exploitable corridors consistent with natural or incidental blockage.

\begin{remark}
Although linear and count models are primarily used for analyzing covariate influence  
on traversal cost and disambiguations,  
they also support prediction when obstacle configurations  
and spatial parameters are known or estimated.  
For instance, the \texttt{spatstat.model} package in \texttt{R} \citep{baddeley2010}  
provides \texttt{ppm} for Strauss and \texttt{clusterfit} for Matérn processes.  
These tools allow practitioners to calibrate models from real spatial data  
and forecast traversal cost under plausible obstacle arrangements.
These aspects are deferred for future work.
\end{remark}

\section{Illustrative Geospatial Case Study}
\label{sec:llust-geospat}

To illustrate the real-world relevance of our framework, 
we constructed a street network for downtown Auburn, Alabama, 
centered at Toomer’s Corner—the symbolic intersection of College Street and Magnolia Avenue 
and one of the busiest pedestrian and vehicle corridors in the city. 
This setting provides a natural testbed: disruptions along the main streets of Auburn 
can immediately affect both everyday traffic and emergency evacuation.

Obstacles were modeled as disk-shaped disruption zones positioned directly on the street network. 
These zones represent realistic short-term blockages such as construction sites, 
accident scenes, or barricades during football game weekends. 
Two spatial patterns were imposed:  
(i) a \emph{regular pattern} approximating a Strauss process, mimicking evenly spaced work sites or coordinated closures, and  
(ii) a \emph{clustered pattern} generated by a Matérn process, reflecting incidents concentrated in a few downtown blocks. 
Each disruption was designated either as a true obstacle (completely blocking traffic)
or a false obstacle (appearing disruptive but passable), 
with sensor readings providing uncertain information on their status.

Figure~\ref{fig:auburn_traversal} compares RD-based traversal paths 
from the northwestern corner of downtown Auburn 
to the southeastern side along the main street network.  
In the baseline case (no disruptions), the optimal route measures 2067 meters.  
When clustered obstacles are introduced, traversal length increases to 2236 meters, 
as detours occur only near localized blockages but unobstructed corridors remain open.  
By contrast, regular spacing of obstacles forces repeated detours and resets, 
raising the total cost to 2480 meters.  

This case demonstrates the practical implications of obstacle spatial patterns:  
for Auburn drivers or emergency responders, recognizing whether blockages are scattered in a regular sequence 
(as might occur with coordinated construction projects) or concentrated in clusters 
(as in storm debris or localized accidents) fundamentally changes the best routing strategy.  
From a planning perspective, this insight highlights how adversarial or coordinated disruptions can be more damaging 
than naturally clustered ones, even when the total number of obstacles is the same.

\begin{figure}[!ht]
\centering
\begin{minipage}{0.5\textwidth}
    \centering
    \includegraphics[width=\textwidth]{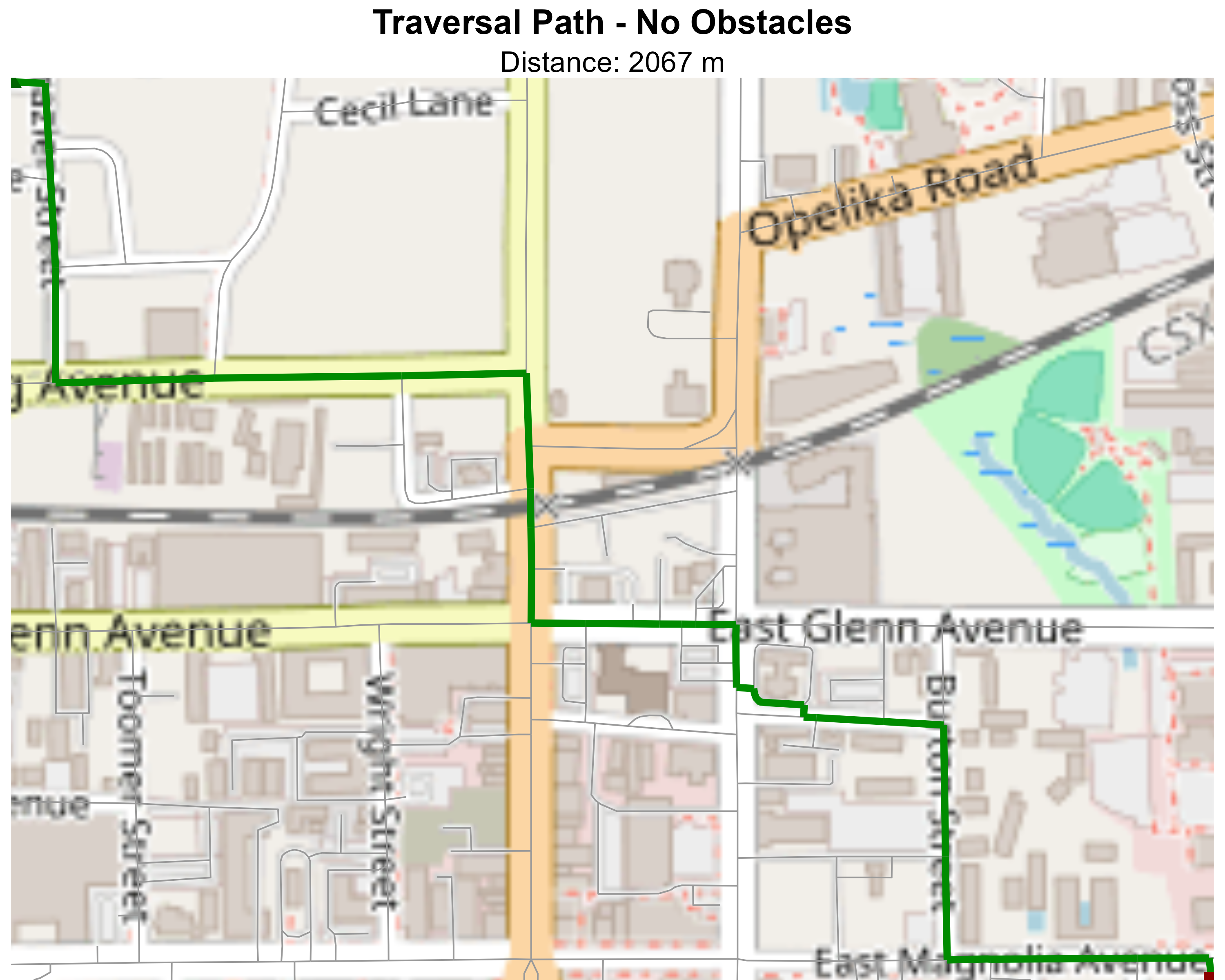}\\
    (a)
\end{minipage}\hfill
\begin{minipage}{0.5\textwidth}
    \centering
    \includegraphics[width=\textwidth]{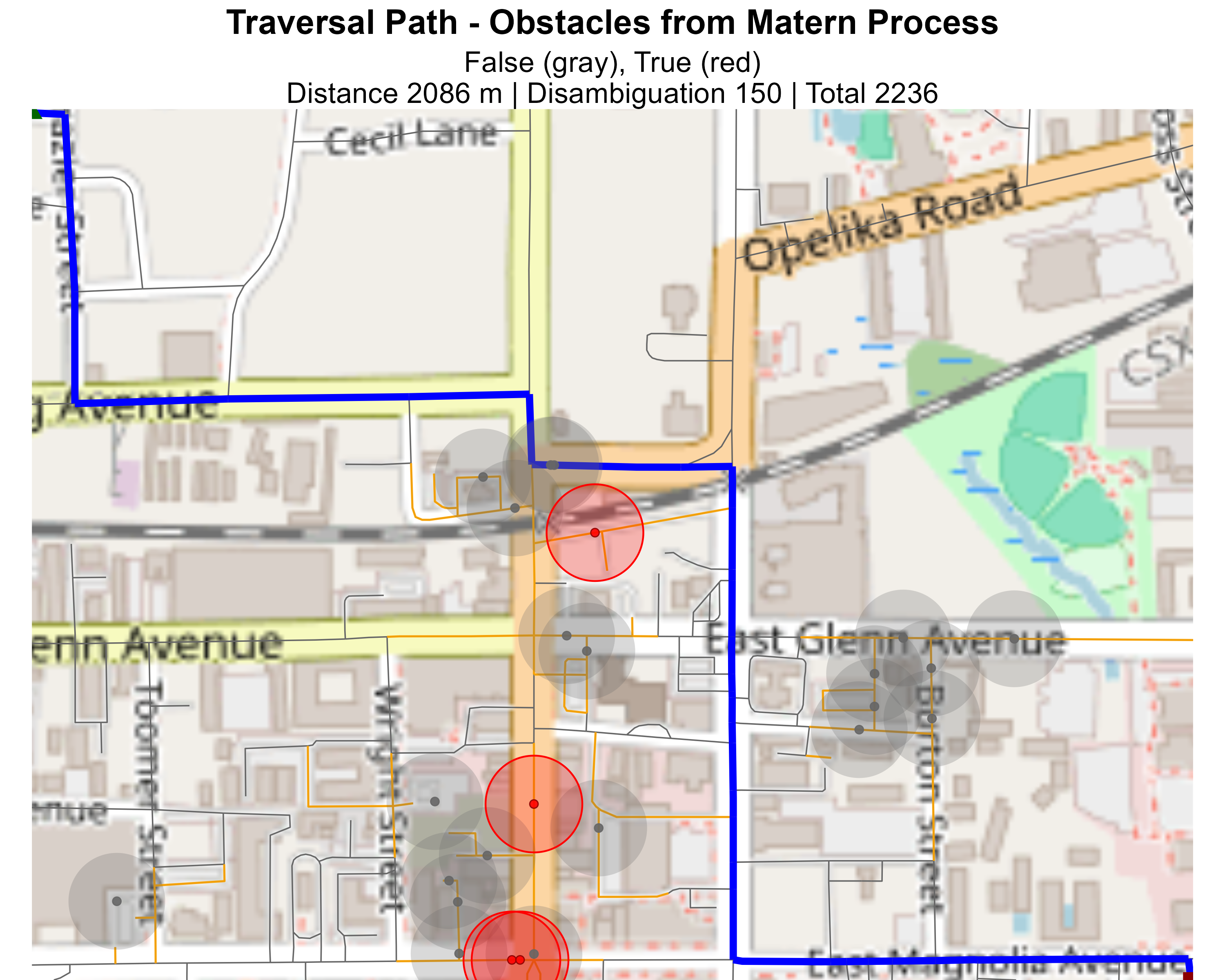}\\
    (b)
\end{minipage}\hfill
\begin{minipage}{0.5\textwidth}
    \centering
    \includegraphics[width=\textwidth]{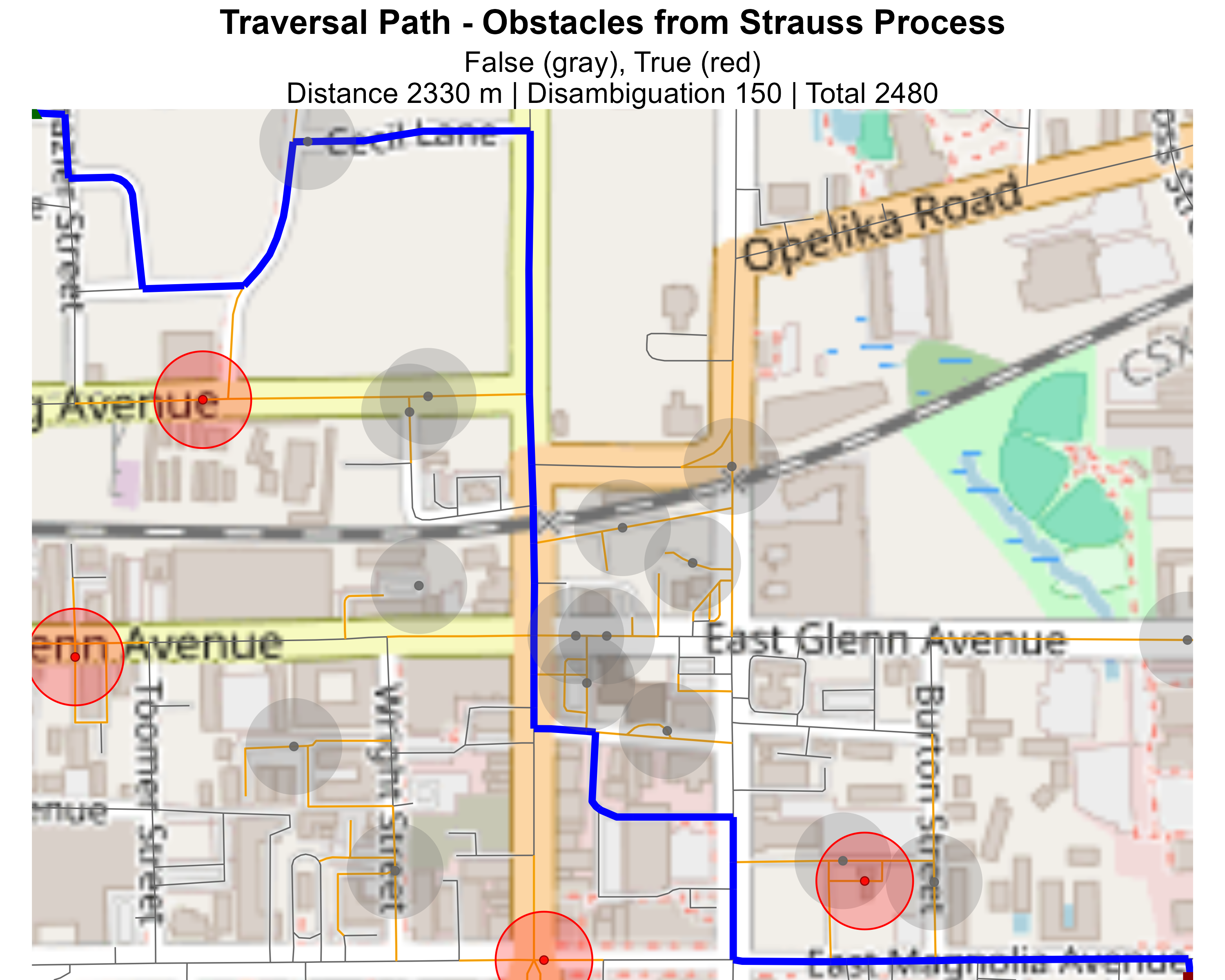}\\
    (c)
\end{minipage}
\caption{Street network traversal from northwest to southeast in Auburn, AL. 
(a) Baseline case: 2067m. 
(b) Clustered obstacles (Matérn): 2236m (2086m path + 150m disambiguation). 
(c) Regular obstacles (Strauss): 2480m (2330m path + 150m disambiguation).}
\label{fig:auburn_traversal}
\end{figure}

\section{Stochastic Ordering of Path Weights and Traversal Costs}
\label{sec:stoch-order}
Beyond Monte Carlo averages, decision makers often require comparative
assessments of obstacle patterns: which configurations are systematically more
disruptive? Stochastic ordering provides a rigorous tool for such comparisons.
For instance, in maritime defense, a regular minefield may stochastically dominate
a uniform placement in traversal cost, highlighting its greater blocking power.
Conversely, in flood evacuation scenarios, clustered debris may result in
stochastically lower traversal costs by leaving wider unobstructed corridors.

We provide stochastic ordering results for the total path weights 
and conjecture a stochastic ordering for the traversal costs of NAVA,
some of which being inspired by our simulation results,
under various obstacle pattern settings.
We first provide the definition of stochastic ordering for completeness.
If $X$ and $Y$ are random variables defined on the same sample space $\Omega$,
then $\X$ is \emph{stochastically smaller} than $Y$ (denoted as $X \leq_{st} Y$)
if $F_X(\omega) \geq F_Y(\omega)$ for all $\omega \in \Omega$.
Trivially, it follows that $X \leq_{st} Y$ iff $a\,X \leq_{st} a\,Y$ for any constant $a>0$
and $X \leq_{st} Y$  iff $b+X \leq_{st} b+Y$ for any constant $b$.
Furthermore,
$X \leq_{st} Y$ iff $g(X) \leq_{st} g(Y)$ provided that $g(x)$ is monotone increasing,
and 
$X \leq_{st} Y$ iff $g(Y) \leq_{st} g(X)$ provided that $g(x)$ is monotone decreasing.

In this section, we present formal results on the \textbf{stochastic ordering}  
of path weights and, less formally, traversal costs incurred by the NAVA  
under various spatial obstacle configurations.  
These theoretical insights complement the empirical trends observed  
in Section~\ref{sec:MCexp-results-and-analysis},  
offering a rigorous perspective on how obstacle composition  
and spatial structure shape expected navigation outcomes.
In this section,
we only state the main results and defer the formal proofs to the Appendix.

We first define \emph{stochastic ordering},  
a concept widely used in probability and decision theory.  
Let $X$ and $Y$ be random variables on the same probability space.  
Then $X$ is said to be \emph{stochastically smaller} than $Y$ (denoted $X \le_{st} Y$)  
if $F_X(x) \ge F_Y(x)$ for all $x \in \mathbb{R}$,  
where $F_X$ and $F_Y$ are the respective cumulative distribution functions (cdfs).  
That is, $X$ tends to take on lower values than $Y$ with higher probability.

This ordering is particularly relevant for spatial navigation problems like ours:  
if the weight of a path $\pi$ under obstacle pattern A is stochastically smaller than under pattern B,  
then we can expect NAVA to incur lower traversal costs under pattern A with high probability.


Recall that NAVA aims to traverse from the start point $s$ to the target point $t$
on an $(s,t)$ path on the discretized grid (i.e., an $(s,t)$ walk on the graph $G$),
with each edge weighted by its Euclidean length and
the disambiguation cost of each obstacle disk intersecting the edge (see Section \ref{sec:dicrete-OOP-problem}).
Suppose that one vertex on the middle of one boundary edge of $\Omega$ is the start point $s$,
and the middle vertex on the opposite boundary edge of $\Omega$ is the target $t$.
NAVA traverses on $(s,t)$ paths (walks to be more precise) on $G$
(i.e., on paths connecting $s$ to $t$ on edges of $G$).
Suppose also that the (Euclidean) lengths of all $(s,t)$-paths consist of $k$ many distinct values (avoiding loops)
with $i^{th}$ length occurring $k_i$ times.
Denote such an $(s,t)$ path as $\pi_{ij}=\pi_{ij}(s,t)$ for $i=1,2,\ldots,k$ and $j=1,2,\ldots,k_i$
so that there are in total $K = \sum_{i=1}^{k}k_i$ many $(s,t)$ paths on the usual 8-adjacency integer grid.
Let $L_{ij}=L(\pi_{ij})$ be the corresponding Euclidean length for path $\pi_{ij}$
and without loss of generality we can assume that $L_{1j}=m < L_{2j} < \ldots < L_{kj}$.
Notice that for any fixed $i$,
there are $k_i$ paths of equal length, i.e., $L(\pi_{ij}) = L(\pi_{ij'})$ for $j,j'=1,2,\ldots,k_i$.

We further assume that $r > 1/2$ so that an obstacle intersects grid edges with probability 1
and the disambiguation cost $c(x)=c > 0$ is fixed for all obstacles,
$p(x)$ is from a continuous distribution,
and NAVA uses the RD algorithm \citep{aksakalli2011}.

Let $W^{\mF}_{ij}$ denote the path weight when all obstacles are false (false-only),  
$W^{\mT}_{ij}$ when all are true (hard obstructions),  
and $W^{\mM}_{ij}$ for a mixed setting with both true and false obstacles.  
Our goal is to compare these random variables via stochastic ordering  
to better understand how obstacle composition impacts perceived traversal cost.
Let $W_{ij}=W(\pi_{ij},\X)$ be the cost assigned to path $\pi_{ij}$ by NAVA
utilizing its (imperfect) sensor (prior to its traversal),
and $w(e)$ be the weight for any edge $e$ in $\pi_{ij}$ assigned by the RD algorithm,
which incorporates the probability $p(x)$ if $e$ intersects the obstacle(s)
(see Equations \eqref{eqn:weight-of-edges} and \eqref{eqn:Lpi}).
Let $\pi^*:=\arg \min_{\pi_{ij}} W_{ij}$ be the path of minimum weight for NAVA 
and $W^*$ be the corresponding path weight.

We will present a stochastic ordering between total weights of path $\pi_{ij}$
under false-only, true obstacle only, and mixed obstacle cases.
Let $W^{\mF}_{ij}$ be the path weight of path $\pi_{ij}$ under the false-only case,
$W^{\mT}_{ij}$ be the path weight under the true obstacle only case,
and $W^{\mM}_{ij}$ be the path weight under the mixed obstacle case.
Also,
assume that NAVA's sensor is using Beta$(a,b)$ distribution for assigning probabilities $p_F$ and $p_T$
to false and true obstacle disks, respectively.

\subsection{Theoretical Tools for Stochastic Ordering and Implications for NAVA Path Weights}

To formalize comparisons between traversal costs under different obstacle compositions,  
we use the following standard result.

\begin{lemma}
\label{lem:stoch-order-sums-rvs}
Let $X_i$ and $Y_i$, for $i=1,2,\ldots,n$, be independent continuous random variables.  
If $X_i \le_{st} Y_i$ for all $i$, then:
\[
\sum_{i=1}^{n} X_i \le_{st} \sum_{i=1}^{n} Y_i.
\]
\end{lemma}

\begin{proof}
  Here, we will prove it for $n=2$, the general result follows by induction on $n$. 
By the independence assumption together with the convolution formula, we get
  \begin{align*}
     f_{X_1+X_2}(x) = \left(f_{X_1} * f_{X_2}\right)(x)
                    = \int_{-\infty}^{\infty} f_{X_1}(x-y) f_{X_2}(y) \ud y.
  \end{align*}
  Hence, for all $x\in \mathbb{R}$, we have that
  \begin{align*}
    F_{X_1+X_2}(x)
    =   & \int_{-\infty}^x f_{X_1+X_2}(z) \ud z                                             
    =   \int_{-\infty}^x \ud z  \int_{-\infty}^{\infty}\ud z\: f_{X_1}(y-z) f_{X_2}(z)    \\
    =   & \int_{-\infty}^{\infty} \ud z\: f_{X_2}(z) \int_{-\infty}^{x}\ud y\: f_{X_1}(y-z) 
    =   \int_{-\infty}^{\infty} \ud z\: f_{X_2}(z) F_{X_1}(x-z)                           \\
    \ge & \int_{-\infty}^{\infty} \ud z\: f_{X_2}(z) F_{Y_1}(x-z)   \quad \text{(since $X_1 \le_{st} Y_1$)} \\
    =   & \int_{-\infty}^{\infty} \ud z\: f_{X_2}(z) \int_{-\infty}^{x}\ud y\: f_{Y_1}(y-z) 
    =   \int_{-\infty}^{x} \ud y \int_{-\infty}^{\infty}\ud z\: f_{X_2}(z) f_{Y_1}(y-z)   \\
    =   & \int_{-\infty}^{x} \ud y \int_{-\infty}^{\infty}\ud z\: f_{X_2}(y-z) f_{Y_1}(z) \quad 
    \text{(by the change of variables $z \leftrightarrow y-z$)}   \\
    =   & \int_{-\infty}^{\infty} \ud z\: f_{Y_1}(z) \int_{-\infty}^{x}\ud y\: f_{X_2}(y-z) 
    =   \int_{-\infty}^{\infty} \ud z\: f_{Y_1}(z) F_{X_2}(x-z)                           \\
    \ge & \int_{-\infty}^{\infty} \ud z\: f_{Y_1}(z) F_{Y_2}(x-z)    \quad \text{(since $X_2 \le_{st} Y_2$)}  \\
    =   & \int_{-\infty}^{\infty} \ud z\: f_{Y_2}(z) \int_{-\infty}^{x}\ud y\: f_{Y_1}(y-z) 
    =   \int_{-\infty}^x \ud y  \int_{-\infty}^{\infty}\ud z\: f_{Y_1}(y-z) f_{Y_2}(z)    \\
    =   & F_{Y_1+Y_2}(x).
  \end{align*}
  This proves the lemma.
\end{proof}

\begin{remark}
The result in Lemma \ref{lem:stoch-order-sums-rvs} also holds for discrete random variables with integrals replaced with sums above.
\end{remark}

Since path weights $W_{ij}$ are computed as sums over edge weights—  
each influenced by obstacle proximity and sensor-derived probabilities—  
Lemma~\ref{lem:stoch-order-sums-rvs} provides a natural foundation  
for comparing path distributions under different obstacle types.  

Recall that sensor outputs follow Beta distributions,  
with true obstacles typically assigned $p_T \sim \Beta(b,a)$  
and false ones $p_F \sim \Beta(a,b)$ for $a < b$,  
reflecting stronger sensor discrimination.  
This difference implies a stochastic ordering between the respective probability marks,  
which directly impacts the edge weights and, hence, the total perceived path weights.

\begin{lemma}
\label{lem:stoch-order-beta(a,b)-rvs}
Let $X \sim \Beta(a,b)$ and $Y \sim \Beta(b,a)$ with $a < b$.  
Then $X$ is stochastically smaller than $Y$, that is,
\[
X \le_{st} Y.
\]
\end{lemma}

\begin{proof}
We need to show that $F_X(x) \ge F_Y(x)$ for all $x \in (0,1)$.
That is, 
$$\int_{0}^{x} \frac{\Gamma(a+b)}{\Gamma(a)\Gamma(b)} x^{a-1}(1-x)^{b-1} \ud x \ge 
\int_{0}^{x} \frac{\Gamma(b+a)}{\Gamma(a)\Gamma(b)} y^{b-1}(1-y)^{a-1} \ud y.$$
which holds iff
$$\int_{0}^{x} x^{a-1}(1-x)^{b-1} \ud x \ge \int_{0}^{x} y^{b-1}(1-y)^{a-1} \ud y.$$
In the right hand side, let $u=1-y$, then the integral becomes $\int_{1-x}^{1} u^{a-1}(1-u)^{b-1} \ud u$
where the integrand is the kernel of $Beta(a,b)$ distribution.
So, we need to show
$$\int_{0}^{x} x^{a-1}(1-x)^{b-1} \ud x \ge \int_{1-x}^{1} u^{a-1}(1-u)^{b-1} \ud u.$$
If $x<1-x$ (i.e. $x<1/2$), this inequality holds as is for $a<b$,
because for $t<1/2$,
$t^{a-1}(1-t)^{b-1} \ge (1-t)^{a-1}t^{b-1}$ as $(t/(1-t))^{a-1} \ge (t/(1-t))^{b-1}$ for $a<b$ since $t/(1-t)<1$. 
If $x>1-x$ (i.e. $x>1/2$), this inequality simplifies to
$$\int_{0}^{1-x} x^{a-1}(1-x)^{b-1} \ud x \ge \int_{x}^{1} u^{a-1}(1-u)^{b-1} \ud u.$$
which holds for $a<b$. 

If $x=1/2$, then the inequality becomes $\displaystyle \int_{0}^{1/2} x^{a-1}(1-x)^{b-1} \ud x \ge \int_{1/2}^{1} u^{a-1}(1-u)^{b-1} \ud u$
which holds for $a<b$.
This proves the lemma.
\end{proof}

In the context of NAVA’s decision-making,  
this lemma implies that using a sensor with $\Beta(a,b)$ distribution (with $a < b$)  
results in systematically lower probability estimates for uncertain obstacles  
than a sensor with $\Beta(b,a)$,  
confirming the alignment of sensor marking higher probabilities for true obstacles.

\subsection{Stochastic Ordering of Path Weights Across Obstacle Types}
\label{subsec:stochastic-ordering-CMF}

Building on the stochastic ordering principles  
and Lemma~\ref{lem:stoch-order-sums-rvs},  
we now formalize a key theoretical insight supported by simulation evidence:  
the traversal path weights incurred by the NAVA exhibit a stochastic ordering  
depending on obstacle composition—  
whether the field is composed entirely of false obstacles,  
true obstacles, or a mix.

Let $p_F \sim \Beta(a,b)$ and $p_T \sim \Beta(b,a)$ with $a < b$  
be the sensor-assigned probabilities for false and true obstacles, respectively,  
and assume the spatial distribution of obstacles is identical across settings.

\begin{proposition}
\label{prop:stoch-order-C<mix<L}
Under the assumptions above,  
the following stochastic ordering holds for any fixed $(s,t)$-path $\pi_{ij}$:
\[
W^{\mF}_{ij} \leq_{st} W^{\mM}_{ij} \leq_{st} W^{\mT}_{ij}.
\]
\end{proposition}

\begin{proof} 
By Lemma \ref{lem:stoch-order-beta(a,b)-rvs},
we obtain that $p_F \leq_{st}p_T$.
From Equation \eqref{eqn:weight-of-edges},
it follows that $F(e,p_F,\X) \leq_{st} F(e,p_T,\X)$,
which, by the properties of stochastic ordering,
implies $w(e,p_F,\X) \leq_{st} w(e,p_T,\X)$ for each edge $e$ in $\pi_{ij}$
since $1/(1-x)$ is increasing in $x \in (0,1)$.
Let $\X_{\mF,n_o}$ be the set of $n_o$ false obstacle disks in the false-only case,
$\X_{\mT,n_o}$ be the set of $n_o$ (true) obstacles in the true obstacle only case,
and
$\X_{\mF,n_F}$ and $\X_{\mT,n_T}$ be the false and true obstacle disks, respectively, in the mixed obstacle case.
Partition $\X_{\mF,n_o}$ into $\X'_{\mF,n_F}$ and $\X'_{\mF,n_T}$.
Then
$$W_{ij}^\mF=\sum_{e \in \pi_{ij}}w(e,p,\X_{\mF,n_o})=
\sum_{e \in \pi_{ij}}w(e,p,\X'_{\mF,n_F})+\sum_{e \in \pi_{ij}}w(e,p,\X'_{\mF,n_T})$$
and
$$W^{\mM}_{ij}=\sum_{e \in \pi_{ij}}w(e,p,\X_{\mF,n_F})+\sum_{e \in \pi_{ij}}w(e,p,\X_{\mT,n_T}).$$
Since $F(e,p_F) \stackrel{d}{=} F(e',p_F)$ for all $e,e' \in \pi_{ij}$,
it follows that
$$\sum_{e \in \pi_{ij}}w(e,p,\X_{\mF,n_F}) \stackrel{d}{=} \sum_{e \in \pi_{ij}}w(e,p,\X'_{\mF,n_F}).$$
Also,
$$\sum_{e \in \pi_{ij}}w(e,p,\X'_{\mF,n_T}) \leq_{st} \sum_{e \in \pi_{ij}}w(e,p,\X_{\mT,n_T}),$$
since it follows that for $X_i \stackrel{iid}{\sim}\Beta(a,b)$
and $Y_j \stackrel{iid}{\sim}\Beta(b,a)$ with $a<b$,
$1/(1-X_i) \leq_{st} 1/(1-Y_j)$
and by Lemma \ref{lem:stoch-order-sums-rvs}, we have $\sum_{i=1}^{n} 1/(1-X_i) \leq_{st} \sum_{j=1}^{n} 1/(1-Y_j)$.
Thus,
$$W^{\mF}_{ij} \leq_{st} W^{\mM}_{ij}.$$
By a similar argument,
since all of the summands in $W^{\mT}_{ij}$ include $F(e,p_T)$ and
$F(e,p_T) \stackrel{d}{=} F(e',p_T)$ for all $e,e' \in \pi_{ij}$,
we can also show that
$$W^{\mM}_{ij} \leq_{st} W^{\mT}_{ij}.$$
Hence, the desired result follows.
\end{proof} 

\noindent
This result confirms that increasing the proportion of true obstacles—  
while holding the total number and spatial layout constant—  
leads to stochastically greater traversal costs.  
The ordering applies to the full distributions, not just the expectations,  
reinforcing the robustness of this pattern across scenarios.

We next generalize to the case where the ratio of true to false obstacles varies.
Let $W^{\rho}_{ij} := W^{\mM,\rho}_{ij}$ be the path weight of path $\pi_{ij}$ under mixed obstacle case
where $\rho=n_T/n_F$ is the ratio of true obstacles to false obstacles for a given $n_o=n_T+n_F$. 
Then we also have the following result in Corollary \ref{cor:stoch-order-obs-ratio} 
which follows from Proposition \ref{prop:stoch-order-C<mix<L}:

\begin{corollary}
\label{cor:stoch-order-obs-ratio}
Let $\rho = n_T / n_F$ and $\rho' = n'_T / n'_F$  
be two true-to-false obstacle ratios such that $\rho < \rho'$,  
with total obstacle count $n_o = n_T + n_F = n'_T + n'_F$ fixed across both scenarios.  
Then the corresponding path weights satisfy:
\[
W^{\rho}_{ij} \leq_{st} W^{\rho'}_{ij},
\]
where $W^\rho_{ij}$ denotes the total traversal cost for $\pi_{ij}$ under obstacle ratio $\rho$.
\end{corollary}

\begin{proof}
Let $\X^{\rho}_{\mF,n_F}$ and $\X^{\rho}_{\mT,n_T}$ be the false and true obstacle disks, respectively, 
in the mixed obstacle case corresponding to $\rho=n_T/n_F$
and  $\X^{\rho'}_{\mF,n_F}$ and $\X^{\rho'}_{\mT,n_T}$ be the false and true obstacle disks, respectively, 
in the mixed obstacle case corresponding to $\rho'=n'_T/n'_F$.
Since $\rho < \rho'$, 
we have $n'_F \le n_F$ and $n'_T \ge n_T$.
In either case, without loss of generality,
we take $n_T = \lfloor \rho n_F \rfloor$ and $n'_T = \lfloor \rho' n'_F \rfloor$, 
since the results will also hold if we take ceilings instead.
We also let $n_r = n_o-(n'_F+n_T)$ (note that $n_r \ge 0$).

Partition $\X^{\rho}_{\mF,n_F}$ into $\widetilde \X^{\rho}_{\mF,n'_F}$ and $\widetilde \X^{\rho}_{\mF,n_r}$
and
partition $\X^{\rho'}_{\mT,n'_T}$ into $\widetilde \X^{\rho'}_{\mT,n_T}$ and $\widetilde \X^{\rho'}_{\mT,n_r}$.
Note that such partitions are possible, since $n_F=n'_F+n_r$ and $n'_T=n_T+n_r$.

Then
\begin{align*}
W^{\rho}_{ij} & = \sum_{e \in \pi_{ij}}w(e,p,\X^{\rho}_{\mF,n_F})+\sum_{e \in \pi_{ij}}w(e,p,\X^{\rho}_{\mT,n_T}) \\
   & =\sum_{e \in \pi_{ij}}w(e,p,\widetilde \X^{\rho}_{\mF,n'_F})+\sum_{e \in \pi_{ij}}w(e,p,\widetilde \X^{\rho}_{\mF,n_r})
   +\sum_{e \in \pi_{ij}}w(e,p,\X^{\rho}_{\mT,n_T})
\end{align*}
and
\begin{align*}
W^{\rho'}_{ij} & = \sum_{e \in \pi_{ij}}w(e,p,\X^{\rho'}_{\mF,n'_F})+\sum_{e \in \pi_{ij}}w(e,p,\X^{\rho'}_{\mT,n'_T}) \\
   & =\sum_{e \in \pi_{ij}}w(e,p,\X^{\rho'}_{\mF,n'_F})+\sum_{e \in \pi_{ij}}w(e,p,\widetilde \X^{\rho'}_{\mT,n_T})
   +\sum_{e \in \pi_{ij}}w(e,p,\widetilde \X^{\rho'}_{\mT,n_r}).
\end{align*}

As in the proof of Lemma \ref{lem:stoch-order-sums-rvs},
it follows that
$$\sum_{e \in \pi_{ij}}w(e,p,\widetilde \X^{\rho}_{\mF,n'_F}) \stackrel{d}{=} \sum_{e \in \pi_{ij}}w(e,p,\X^{\rho'}_{\mF,n'_F}),$$
$$\sum_{e \in \pi_{ij}}w(e,p,\X^{\rho}_{\mT,n_T}) \stackrel{d}{=} \sum_{e \in \pi_{ij}}w(e,p,\widetilde \X^{\rho'}_{\mT,n_T}),$$
and
$$\sum_{e \in \pi_{ij}}w(e,p,\widetilde \X^{\rho}_{\mF,n_r}) \leq_{st} \sum_{e \in \pi_{ij}}w(e,p,\widetilde \X^{\rho'}_{\mT,n_r}).$$
Thus,
$$W^{\rho}_{ij} \leq_{st} W^{\rho'}_{ij} ~~ \text{ for } ~ \rho < \rho'$$
which is the desired result.
\end{proof}

\noindent
This corollary formalizes the observed monotonicity in traversal cost:  
as the environment becomes more hazardous (higher $\rho$),  
the distribution of possible traversal outcomes shifts toward higher cost.  
This insight supports risk-aware route planning  
under varying levels of environmental threat.

\subsection{Sensor Quality and Its Impact on Traversal Cost}
\label{subsec:sensor-accuracy-stoch-order}

In real-world navigation systems,  
the fidelity of sensor-derived probability marks is a key factor in assessing traversal risk.  
We now analyze how the \emph{quality of probabilistic sensor readings},  
modeled via Beta distributions, affects the traversal cost for the NAVA.

Let $W^{\mF,a,b}_{ij}$ denote the total traversal cost of path $\pi_{ij}$  
when all obstacles are false and sensor readings follow $p_F \sim \Beta(a,b)$.  
Similarly, let $W^{\mT,a,b}_{ij}$ denote the cost when all obstacles are true  
and $p_T \sim \Beta(b,a)$—  
a setting where higher $b/a$ ratios correspond to sharper sensor discrimination.

Now consider two sensor regimes:
\begin{itemize}
  \item A \textbf{high-fidelity sensor}, using $\Beta(a,b)$ for false obstacles  
  and $\Beta(b,a)$ for true obstacles, with $a < b$.

  \item A \textbf{lower-fidelity sensor}, using $\Beta(a',b')$ and $\Beta(b',a')$  
  with $a < a'$ and $b > b'$ (i.e., flatter distributions, less concentrated near 0 or 1).
\end{itemize}

The intuition is that better sensors produce probability marks  
that more accurately reflect obstacle type—  
e.g., clustering near 0 for false and near 1 for true obstacles—  
thus enabling more informed disambiguation decisions and reducing traversal cost.

\begin{proposition}
\label{prop:stoch-order-beta-a,b}
Assume the spatial configuration of obstacle locations is fixed and identical across scenarios,  
and the total number of obstacles $n_o$ remains constant. Then:
\[
W^{\mF,a,b}_{ij} \leq_{st} W^{\mF,a',b'}_{ij}
\quad \text{and} \quad
W^{\mT,a,b}_{ij} \leq_{st} W^{\mT,a',b'}_{ij},
\]
where $\leq_{st}$ denotes stochastic dominance (i.e. ordering).
\end{proposition}

\begin{proof}
Using the cdf's of given Beta distributions,
we see that $p_F \leq_{st} p'_F$
and $p'_T \leq_{st}p_T$.
Then the result follows similar to the Proof of Proposition \ref{prop:stoch-order-C<mix<L}
(hence details are not presented).
\end{proof}

This result formalizes the idea that better sensor quality—  
reflected in more peaked Beta distributions—  
yields stochastically lower traversal costs.  
It provides a theoretical basis for evaluating sensor designs  
within Beta-distributed uncertainty models,  
commonly used in probabilistic risk-aware navigation frameworks.

\medskip

\textbf{Interpretation of Stochastic Ordering via Mean and Median:}  
The stochastic ordering results established in Propositions~\ref{prop:stoch-order-C<mix<L}  
and \ref{prop:stoch-order-beta-a,b}  
have direct implications for commonly used summary statistics—  
namely, the mean and the median.

\begin{itemize}
  \item \textbf{Ordering in Expected Traversal Cost:} 
  Let $X$ and $Y$ be two non-negative continuous random variables with $X \leq_{st} Y$  
  (i.e., $F_X(x) \geq F_Y(x)$ for all $x$). Then:
  \[
  \E[X] = \int_0^\infty \big(1 - F_X(x)\big) \, \ud x  
  \leq \int_0^\infty \big(1 - F_Y(x)\big) \, \ud x = \E[Y].
  \]
  Applying this to our setting yields:
  \[
  \E[W^{\mF}_{ij}] \leq \E[W^{\mM}_{ij}] \leq \E[W^{\mT}_{ij}],
  \]
  which affirms that expected traversal cost increases  
  with the presence of more true obstacles or poorer sensor discrimination.

  \item \textbf{Ordering in Median (and Other Quantiles):} 
  If $X \leq_{st} Y$, then $\text{Median}(X) \leq \text{Median}(Y)$.  
  In our case:
  \[
  \text{Median}(W^{\mF}_{ij}) \leq \text{Median}(W^{\mM}_{ij}) \leq \text{Median}(W^{\mT}_{ij}).
  \]
  This is particularly relevant when robust path planning  
  or percentile-based risk analysis is preferred over mean-based metrics.
\end{itemize}

\noindent
In summary, stochastic dominance among traversal costs  
ensures corresponding orderings in both expectation and central tendency.  
These interpretations enhance the practical value of our theoretical results  
for risk-aware navigation under uncertain obstacle environments.

\subsection{Impact of Spatial Obstacle Patterns on Path Cost Distribution}
\label{sec:spatial-pattern-range-cost}

We now shift focus from expectations to variability in path weights  
under different obstacle placement patterns.  
Let $W^{\text{Str}}_{ij}$, $W^{\text{Unif}}_{ij}$, and $W^{\text{Mat}}_{ij}$  
denote the total weights of a fixed path $\pi_{ij}$  
when obstacle centers follow a spatially regular (Strauss), uniform,  
or clustered (Matérn) distribution, respectively.  
These patterns are generated via Strauss$(n, d, \gamma)$  
and Matérn$(\kappa, r_0, \mu)$ processes,  
commonly used in spatial point pattern modeling.

\begin{itemize}
\item \textbf{Regular Pattern:} \\
Strauss processes with low $\gamma$ and interaction radius $d \le 2r$  
produce near-grid-like obstacle layouts,  
leading to consistent obstruction across realizations.  
Consequently, path weights exhibit low variance.

\item \textbf{Uniform Pattern:} \\
At $\gamma = 1$, the Strauss process approximates complete spatial randomness (CSR).  
Variability increases as paths encounter differing obstacle densities across trials.

\item \textbf{Clustered Pattern:} \\
Matérn processes with small $r_0$ generate tight clusters.  
Traversal cost becomes highly variable:  
some paths intersect dense clusters (yielding high cost),  
while others pass through cluster-free regions (yielding low cost).
\end{itemize}

\textbf{Empirical Ordering of Variability:}
Simulation results suggest:
\[
\Range(W^{\text{Mat}}_{ij}) \geq \Range(W^{\text{Unif}}_{ij}) \geq \Range(W^{\text{Str}}_{ij}) 
\quad \text{for } d \le 2r,
\]
with similar trends observed for variances.  
These findings complement the stochastic ordering results in Section~\ref{sec:stoch-order}  
by highlighting variability, not just central tendency.

\textbf{Remarks on Probabilistic Ordering:}
While mean and median traversal costs obey stochastic dominance  
(Proposition~\ref{prop:stoch-order-C<mix<L}),  
full stochastic ordering may not hold between spatial configurations.  
Instead, we observe probabilistic dominance with positive probability:
\begin{enumerate}[label=(\roman*)]
  \item $P(W^{\text{Mat}}_{ij} \leq W^{\text{Unif}}_{ij} \leq W^{\text{Str}}_{ij}) > 0$  
  when $d \le 2r$.
  
  \item For Strauss processes with $\gamma < \gamma'$:
    \begin{itemize}
      \item[(a)] $P(W^{\gamma'}_{ij} \leq W^{\gamma}_{ij}) > 0$  
      and vice versa for $d \le 2r$,
      
      \item[(b)] $P(W^{\gamma}_{ij} \leq W^{\gamma'}_{ij}) > 0$  
      and vice versa for $d > 2r$.
    \end{itemize}
  
  \item For Matérn processes with $r_0 < r'_0$:
  \[
  P(W^{r_0}_{ij} \leq W^{r'_0}_{ij}) > 0 
  \quad \text{and} \quad 
  P(W^{r_0}_{ij} \geq W^{r'_0}_{ij}) > 0.
  \]
\end{enumerate}

\textbf{Conjectured Probabilistic Orderings:}
For fixed path $\pi_{ij}$:
\begin{itemize}
  \item[] (ii)-(a) $P(W^{\gamma'}_{ij} \leq W^{\gamma}_{ij}) > 0.5$ if $d \le 2r$,
  \item[] (ii)-(b) $P(W^{\gamma}_{ij} \leq W^{\gamma'}_{ij}) > 0.5$ if $d > 2r$,
  \item[] (iii) $P(W^{r_0}_{ij} \leq W^{r'_0}_{ij}) > 0.5$ for $r_0 < r'_0$.
\end{itemize}

These probabilistic comparisons highlight the interplay  
between spatial structure and traversal risk  
that goes beyond expectation-based analysis.

\textbf{Path Cost under NAVA’s Sensor Model:}
Let $C$ denote the actual traversal cost of the path  
selected by the NAVA using the RD algorithm.  
Depending on the obstacle composition (false, true, or mixed),  
the traversal cost behaves differently:

\begin{itemize}
\item \textit{Clutter-Only Case:} \\
When all obstacles are false,  
the cost can be approximated as $C^\mF = L^*_{ij} + w^*_{ij} \cdot c$,  
where $L^*_{ij}$ is the length of the chosen path  
and $w^*_{ij}$ counts the number of false obstacle disks intersecting it.  
Since traversal probabilities are all less than 1,  
disambiguation is sometimes avoided,  
and we typically have $C^\mF \le W^*$ with high probability.

\item \textit{Mixed Obstacle Case:} \\
Some paths may be blocked due to true obstacles.  
The actual cost becomes $C^\mM = \sum w(e) + w^{RD}_{ij} \cdot c$,  
where the second term reflects the number of disambiguation events  
caused by true obstacles encountered and resolved by RD during traversal.

\item \textit{True-Only Case:} \\
All obstacles are true.  
The NAVA avoids impassable regions,  
and the cost is given by $C^\mT = \sum \ell(e) + w^{RD}_{ij} \cdot c$  
over the selected traversable path.
\end{itemize}

Let $\mathscr{P}$ denote the set of all $s$--$t$ paths in $G$,  
and define $\mathscr{P}_\mF$, $\mathscr{P}_\mM$, and $\mathscr{P}_\mT$  
as the subsets of traversable paths under false-only, mixed,  
and true-only obstacle settings, respectively.  
Since true obstacles reduce path feasibility:
\[
\mathscr{P}_\mT \subseteq \mathscr{P}_\mM \subseteq \mathscr{P}_\mF,
\]
we obtain the following ordering for the minimum attainable path weights:
\[
\min_{\pi_{ij} \in \mathscr{P}_\mF} W_{ij} 
\le \min_{\pi_{ij} \in \mathscr{P}_\mM} W_{ij} 
\le \min_{\pi_{ij} \in \mathscr{P}_\mT} W_{ij}.
\]

These relationships suggest the conjectured stochastic ordering in realized traversal cost:
\[
C^\mF \le_{st} C^\mM \le_{st} C^\mT,
\]
although this remains analytically unproven due to the heuristic nature of RD.

\textbf{Final Observations:}
For a given path $\pi_{ij}$:
\begin{itemize}
  \item Traversal costs tend to be lower under clustered obstacle patterns,  
  as such configurations create wider obstacle-free corridors.
  
  \item Between regularity and uniformity,  
  expected traversal cost is higher under regularity when obstacle spacing is moderate ($d \le 2r$),  
  and higher under uniformity when spacing becomes too sparse ($d > 2r$).
  
  \item When obstacles are arranged regularly with $d \approx 1.5r$,  
  mean traversal cost peaks due to frequent and systematic edge--disk intersections.
\end{itemize}

These conclusions reinforce the earlier theoretical  
and simulation-based findings on how spatial structure  
and obstacle composition shape navigation outcomes in adversarial environments. 

\subsection{Summary of Stochastic Ordering Results}
\label{sec:stoch-order-summary}
Beyond average outcomes, it is often important to compare entire
distributions of traversal cost across different obstacle settings.
Stochastic ordering provides a rigorous way to formalize such comparisons.

The main findings can be summarized as follows (see the Appendix
for the technical details and proofs of these results):
\begin{itemize}
\item \textbf{Effect of Obstacle Composition:}  
Traversal costs are stochastically smallest when only false obstacles are present,
largest when only true obstacles are present, and intermediate in mixed settings.
This ordering holds not just in expectation, but also for medians and other quantiles.
This means that routes are systematically easier to find in environments filled only with false obstacles, 
while environments dominated by true obstacles are the most restrictive. 
Mixed settings fall in between, reflecting the balance between passable clutter and impassable barriers.

\item \textbf{Effect of Obstacle Ratio:}  
As the proportion of true obstacles increases while holding the total number of obstacles fixed,
the distribution of traversal costs systematically shifts upward.
As the share of true obstacles increases, the likelihood of longer or more costly routes rises. 
In practice, this highlights how even small increases in hazardous objects can shift the overall navigation risk profile.

\item \textbf{Effect of Sensor Quality:}  
Higher-fidelity sensors, which more reliably distinguish true from false obstacles,
lead to stochastically smaller traversal costs than lower-fidelity sensors.
Better sensors that more reliably separate false from true obstacles reduce the burden of unnecessary detours. 
In applied navigation systems, this underscores the value of investing in high-fidelity sensing technologies.

\item \textbf{Effect of Spatial Structure:}  
Obstacle arrangement strongly shapes both the magnitude and variability of costs.  
Clustered patterns often allow wider unobstructed corridors and hence lower average cost but higher variability;  
regular patterns create consistent blocking and yield higher mean costs with lower variability;  
uniform patterns fall in between.
The way obstacles are arranged influences both the average cost and the variability of travel. 
Clustering can sometimes leave wide corridors open, regular spacing tends to block paths more uniformly, 
and random arrangements fall in between these extremes.

\item \textbf{Implications for NAVA:}  
Although exact dominance results for realized traversal costs remain conjectural under the RD algorithm,
both theoretical and empirical evidence consistently support the ordering:
false-only configurations yield the lowest traversal costs, followed by mixed, and true-only with the highest.
This clearly indicates that navigation is easiest in false-only fields, 
harder in mixed settings, and hardest in true-obstacle environments. 
This provides a clear risk ranking of environments for practical planning.
\end{itemize}

Overall, these results provide a rigorous foundation for understanding how obstacle type,
composition, sensor fidelity, and spatial arrangement jointly determine the distribution of
navigation outcomes. 
 
\section{Discussion and Conclusions}
\label{sec:disc-conc}
This study presents a unified geospatial framework for analyzing  
how spatial obstacle patterns influence the traversal efficiency  
of a NAVA operating in uncertain environments.  
By simulating adversarially placed obstacles using spatial point processes  
and evaluating pathfinding outcomes under varying compositions and configurations,  
we systematically quantify the relationship between obstacle geometry  
and expected traversal cost.

We conduct extensive Monte Carlo experiments under three obstacle compositions: 
(i) false-only (clutter), (ii) true-only (impassable), and (iii) mixed obstacles. 
These are evaluated across spatial layouts ranging from uniform randomness 
to regularity (Strauss processes) and clustering (Matérn processes).
Our results show that obstacle pattern structure significantly affects navigability.  
Specifically, spatial regularity tends to increase traversal cost  
by creating evenly dispersed obstacles,  
while clustering facilitates passage by forming obstacle-free corridors.  
These effects are modulated by key spatial parameters  
such as the Strauss inhibition parameter ($\gamma$)  
and interaction radius ($d$),  
as well as the Matérn cluster radius ($r_0$) and parent intensity ($\kappa$).

To rigorously assess the impact of these factors,  
we employ robust linear regression and random forest models.  
In Strauss-based settings,  
the number of obstacles ($n$), interaction radius ($d$),  
and repulsion parameter ($\gamma$) emerge as dominant predictors.  
For clustered configurations,  
the most influential parameters are the cluster radius and number of clusters,  
while the obstacle count plays a secondary role.  
We further analyze the number of disambiguations ($N_{dis}$)—  
sensor-driven clarifications of obstacle status—  
and find it strongly correlated with traversal cost across all scenarios.  
This reinforces its utility as a proxy measure in operational settings  
where full cost evaluation may be infeasible.

We also establish stochastic ordering results  
that provide theoretical backing for the empirical trends observed.  
We show that traversal costs satisfy the dominance relationship:  
false-only $\le_{st}$ mixed obstacles $\le_{st}$ true-only configurations,  
and that patterns with clustering tend to yield stochastically smaller costs  
than those with uniform or regular layouts.  
The ordering between uniform and regular configurations depends on obstacle spacing:  
moderate interaction distances ($d \approx 1.5r$) in Strauss processes  
yield peak traversal burden,  
while overly sparse regular patterns become less effective.

From a spatial decision-making perspective,  
our findings offer actionable insights.  
If an OPA  
wishes to maximize NAVA's traversal cost under stochastic placement constraints,  
the optimal strategy involves placing true obstacles in a regular pattern  
using a Strauss process with low $\gamma$ and moderate $d$.  
Additionally, increasing the true-to-false obstacle ratio ($\rho$)  
enhances traversal difficulty without revealing the adversarial intent,  
particularly when the placement retains stochastic variability.  
Since NAVA operates without prior knowledge of the underlying spatial pattern,  
it may still exploit spatial diagnostics—  
such as fitting Strauss or Matérn models  
using \texttt{ppm} or \texttt{clusterfit}  
in \texttt{R}'s \texttt{spatstat.model} package \citep{baddeley2010}—  
to inform path adaptation heuristics.

Beyond methodological contributions, our study carries practical implications.
The illustrative case study in Section~5 demonstrates how the framework can be
applied to geospatial decision-making. For instance, in urban evacuation under
flood-induced blockages, clustered obstacles may paradoxically reduce traversal
costs by creating navigable corridors, while dispersed regular placements maximize
obstruction. Similar insights extend to maritime defense (minefields) and off-road
mobility planning (landslides or debris fields). These examples illustrate that
our framework not only advances theoretical understanding of OOP but also provides
a systematic way to evaluate obstacle impacts in real-world navigation settings.

Future research could explore  
deterministic or semi-deterministic obstacle placement strategies  
to identify worst-case spatial configurations.  
Incorporating learning agents that adapt over time to obstacle distributions,  
or extending the framework to continuous or multi-resolution graph models,  
would further enhance its applicability.  
Finally, extending this framework to 3D environments,  
such as underwater minefield navigation or aerial drone routing,  
would broaden its impact in geospatial intelligence,  
autonomous mobility, and risk-aware environmental planning.

\section*{Acknowledgements}
We would like to thank Dr Le Chen for his help  
in proving the result in Lemma \ref{lem:stoch-order-sums-rvs}.  
Most of the Monte Carlo simulations presented in this article  
were executed at Easley HPC Laboratory of Auburn University.  
LZ and EC were supported by Office of Naval Research Grant N00014-22-1-2572  
and EC was supported by NSF Award \# 2319157.
%
%

%

\newpage

\section*{APPENDIX}
\appendix
\setcounter{section}{0}
\setcounter{equation}{0}
\setcounter{figure}{0}
\setcounter{table}{0}

\renewcommand\thesection{A\arabic{section}}
\renewcommand\thesubsection{\thesection.\arabic{subsection}}
\renewcommand\thesubsubsection{\thesubsection.\arabic{subsubsection}}

\renewcommand\theequation{A\arabic{equation}}
\renewcommand\thefigure{A\arabic{figure}}
\renewcommand\thetable{A\arabic{table}}

This appendix provides additional simulation results,  
extended figures, and full regression outputs supporting the main findings  
of the manuscript titled  
``Adversarial Obstacle Placement with Spatial Point Processes for Optimal Path Disruption.''  
Specifically, it includes detailed traversal cost patterns under varying parameters  
for Strauss and Matérn point processes,  
model diagnostics and summaries for all robust regression and random forest models,  
and supplemental contour plots stratified by obstacle number.  
These materials offer further insight into the dependence of traversal cost  
on spatial configuration parameters  
and confirm the robustness of the trends discussed in the main paper.

\section{Monte Carlo Experiments, Analysis, and Results} 
\label{sec:MCexp-results-and-analysis}

To study how the traversal cost of a NAVA using the RDP navigation protocol  
depends on obstacle pattern parameters,  
we use the uniform pattern as a benchmark,  
as is common in spatial point pattern analysis \citep{baddeley2010}. 

For each deviation from uniformity,  
we consider three obstacle composition cases:  
(i) false obstacles only,  
(ii) true obstacles only, and  
(iii) a mix of false and true obstacles,  
referred to as the \emph{mixed obstacles case} for brevity.

\subsection{OOP with Uniform to Regular Obstacle Patterns}
\label{sec:OOP-unif2reg}

\subsubsection{False Obstacles from Uniform to Regular Patterns}
\label{sec:OOP-clut-unif2reg}

We consider the \emph{false obstacles only} case,  
where the obstacle pattern transitions from uniformity to spatial regularity  
using the Strauss$(n,d,\gamma)$ process.  
Given 30 values of $\gamma$, 22 values of $d$,  
and 100 Monte Carlo (MC) replications for each false obstacle level $n_F$,  
we obtain $11 \times 30 \times 100 = 33{,}000$ measurements per $n_F$.

For each replication, the traversal cost $C$ is computed using the RD algorithm \citep{aksakalli2011},  
and the number of disambiguations is recorded.  
A sample realization is shown in Figure~\ref{fig:sample-Strauss-false-mix}(a).

\begin{figure}[htb]
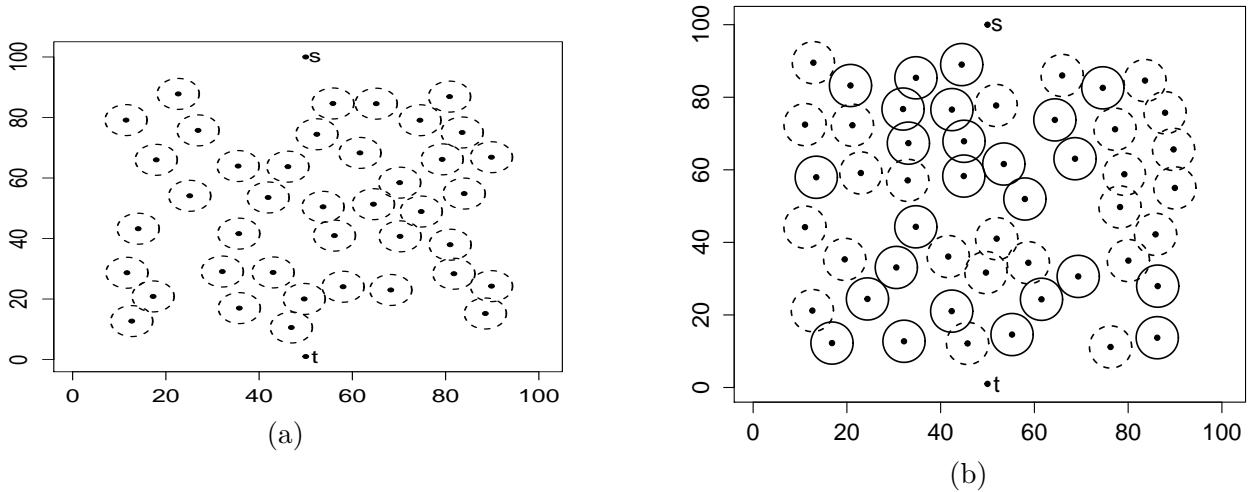

\centering
\begin{minipage}{0.45\textwidth}
    \centering
    \includegraphics[trim=10mm 15mm 10mm 15mm, clip,height=5.1cm,width=\linewidth]{clutter.pdf}\\
    (a)
\end{minipage}\hfill
\begin{minipage}{0.45\textwidth}
    \centering
    \includegraphics[trim=10mm 15mm 10mm 15mm, clip,width=\linewidth]{mix.pdf}\\
    (b)
\end{minipage}
\caption{
(a) Realization of a regular pattern in the false-only case,  
Strauss$(n_F=40, d=9, \gamma=0)$.  
(b) Realization of the mixed obstacle case, Strauss$(n=50, d=9, \gamma=0)$,  
with 25 false (dashed circles) and 25 true (solid circles) obstacles.
}
\label{fig:sample-Strauss-false-mix}
\end{figure}

We examine trends in the mean traversal cost $\bar{C}$  
as a function of obstacle regularity,  
governed by $d$ and $\gamma$ in the Strauss process, and false obstacle number $n_F$.  
Figure~\ref{fig:Lbar-vs-Strauss}(a) plots $\bar{C}$ versus $\gamma$ for each $d$,  
averaged over obstacle number levels.  
Figure~\ref{fig:Lbar-vs-Strauss}(b) shows the correlation $\Corr(\bar{C}, \gamma)$  
versus $d$ at different values of $n_F$.

\begin{figure} [!ht]
\centering
\begin{tabular}{cc}
\begin{subfigure}[b]{0.45\textwidth}
 \centering
\includegraphics[width=\textwidth]{Strauss_Lbarvsgamma_d_clut.png}
\includegraphics[width=\textwidth]{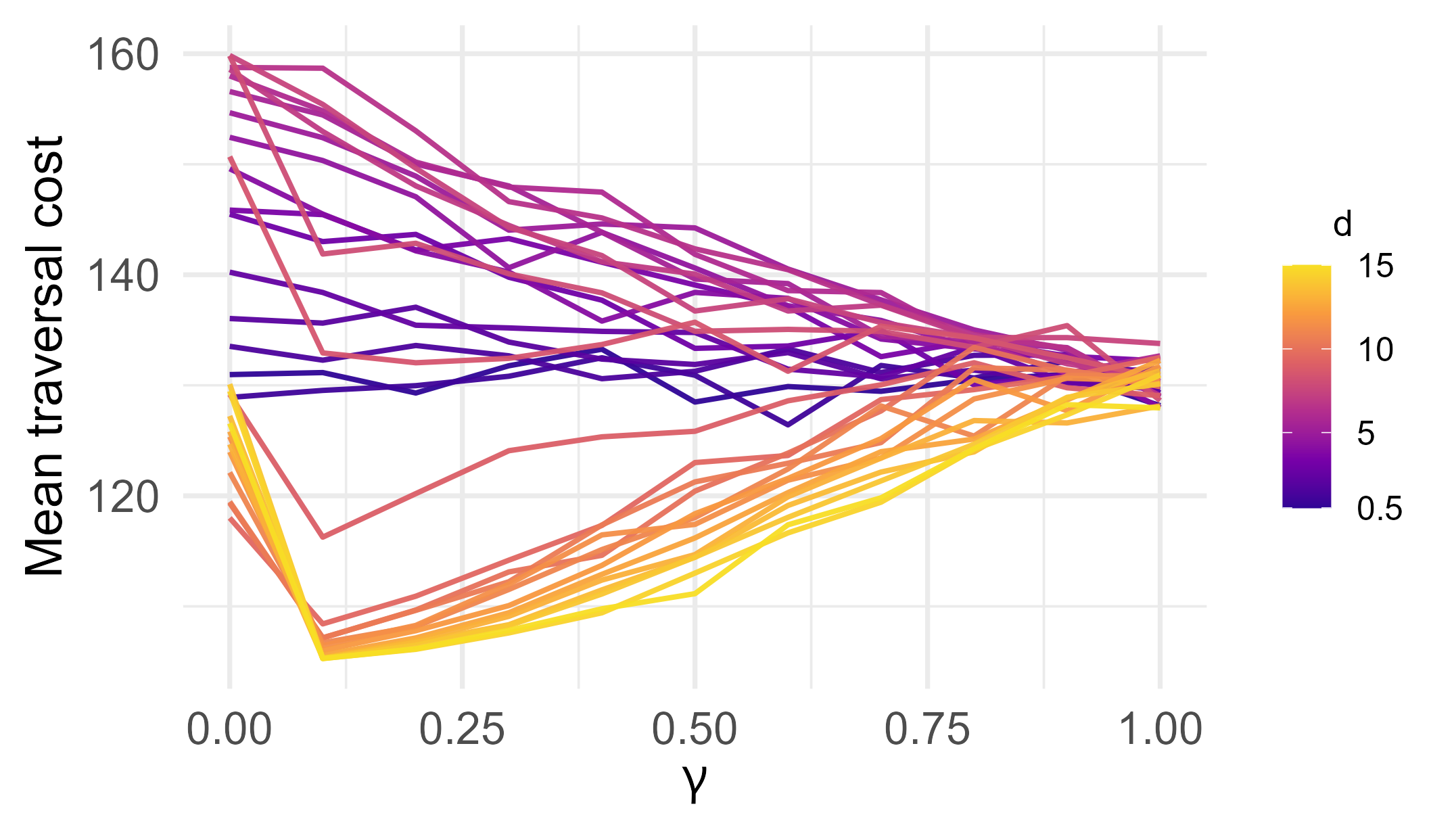}
\includegraphics[width=\textwidth]{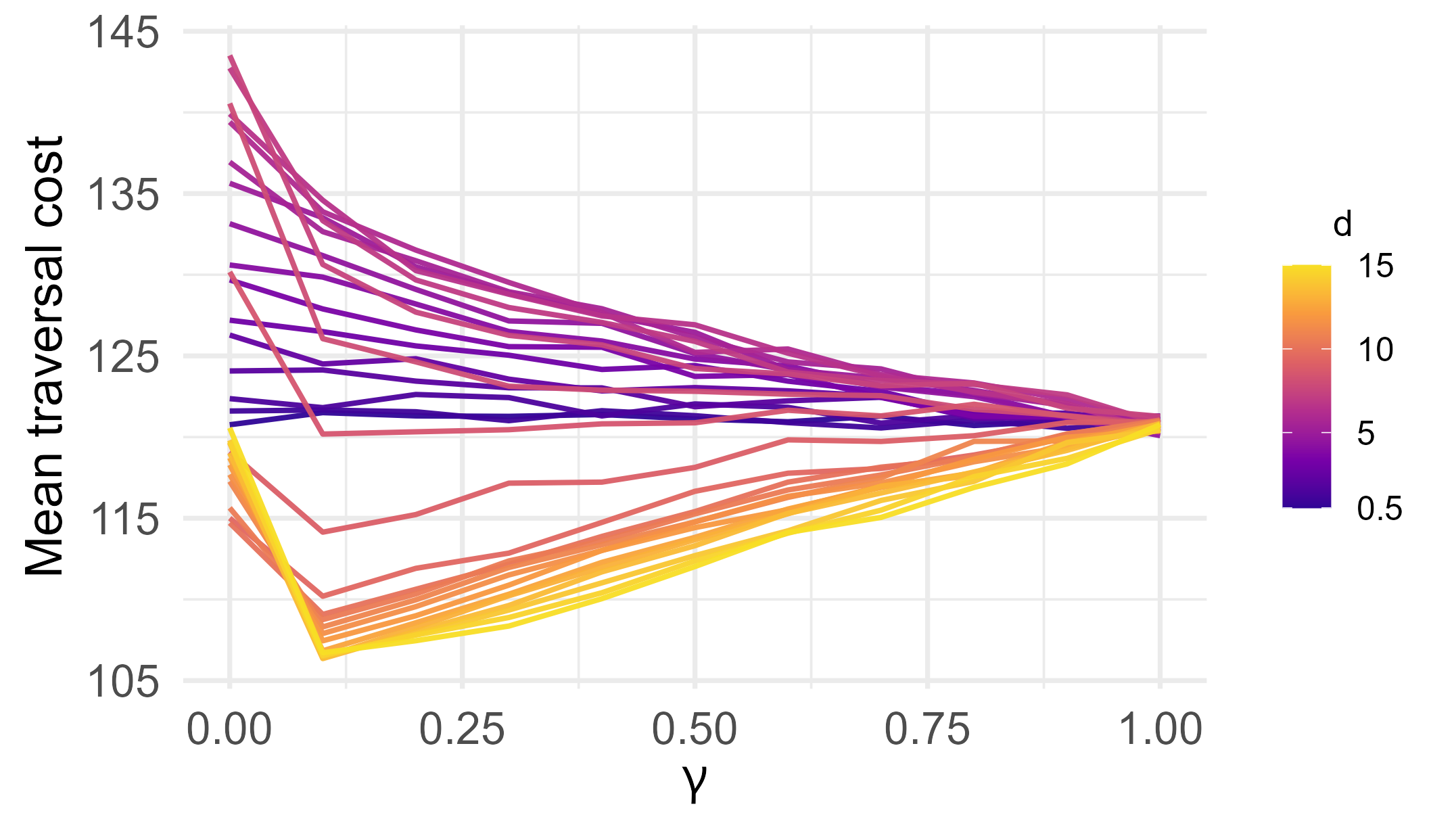}
 \caption{}
\end{subfigure}
\hfill
\begin{subfigure}[b]{0.45\textwidth}
 \centering
\includegraphics[width=\textwidth]{Strauss_Corr_Lbar_gamma_d_clut.png}
\includegraphics[width=\textwidth]{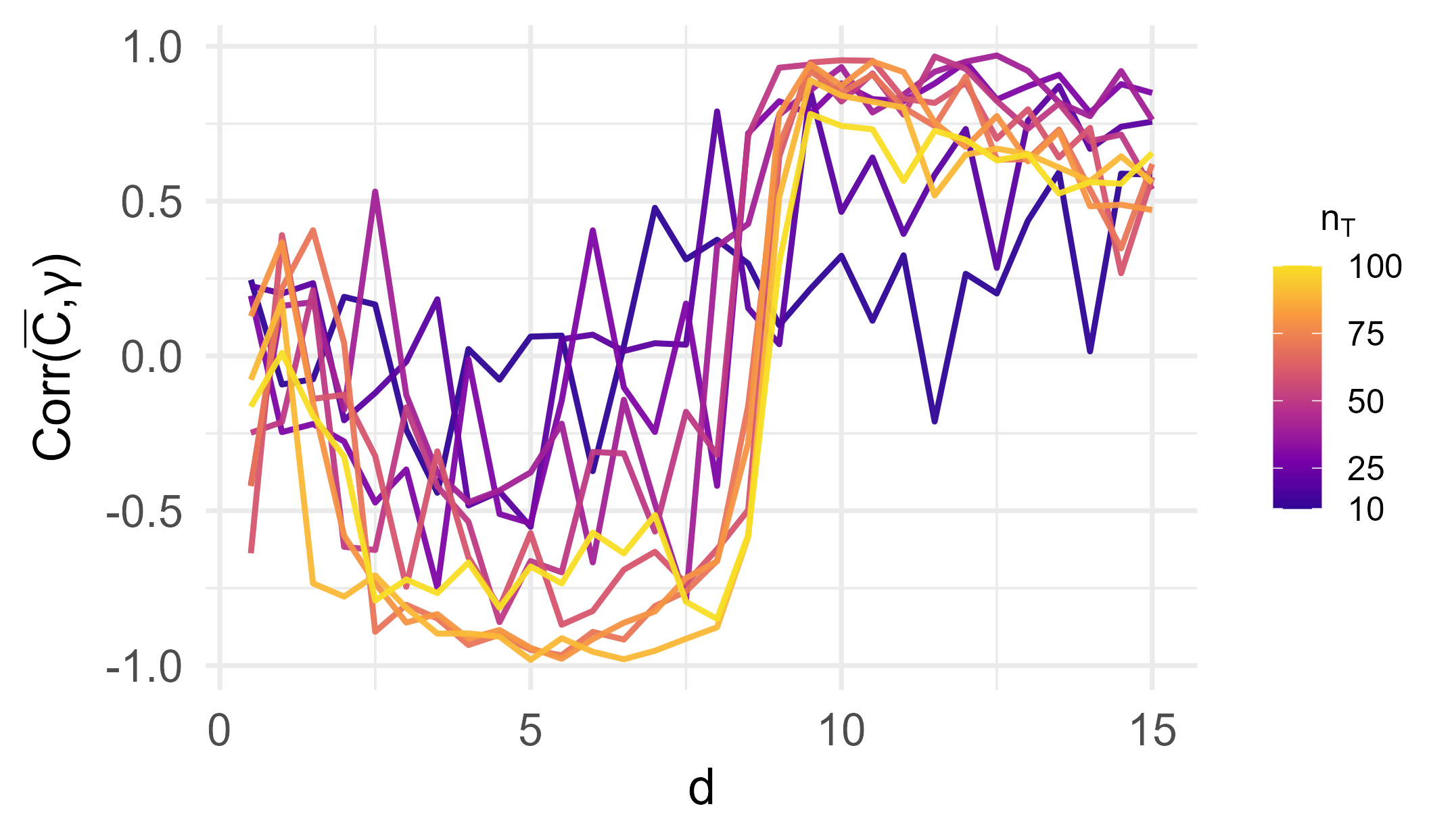}
\includegraphics[width=\textwidth]{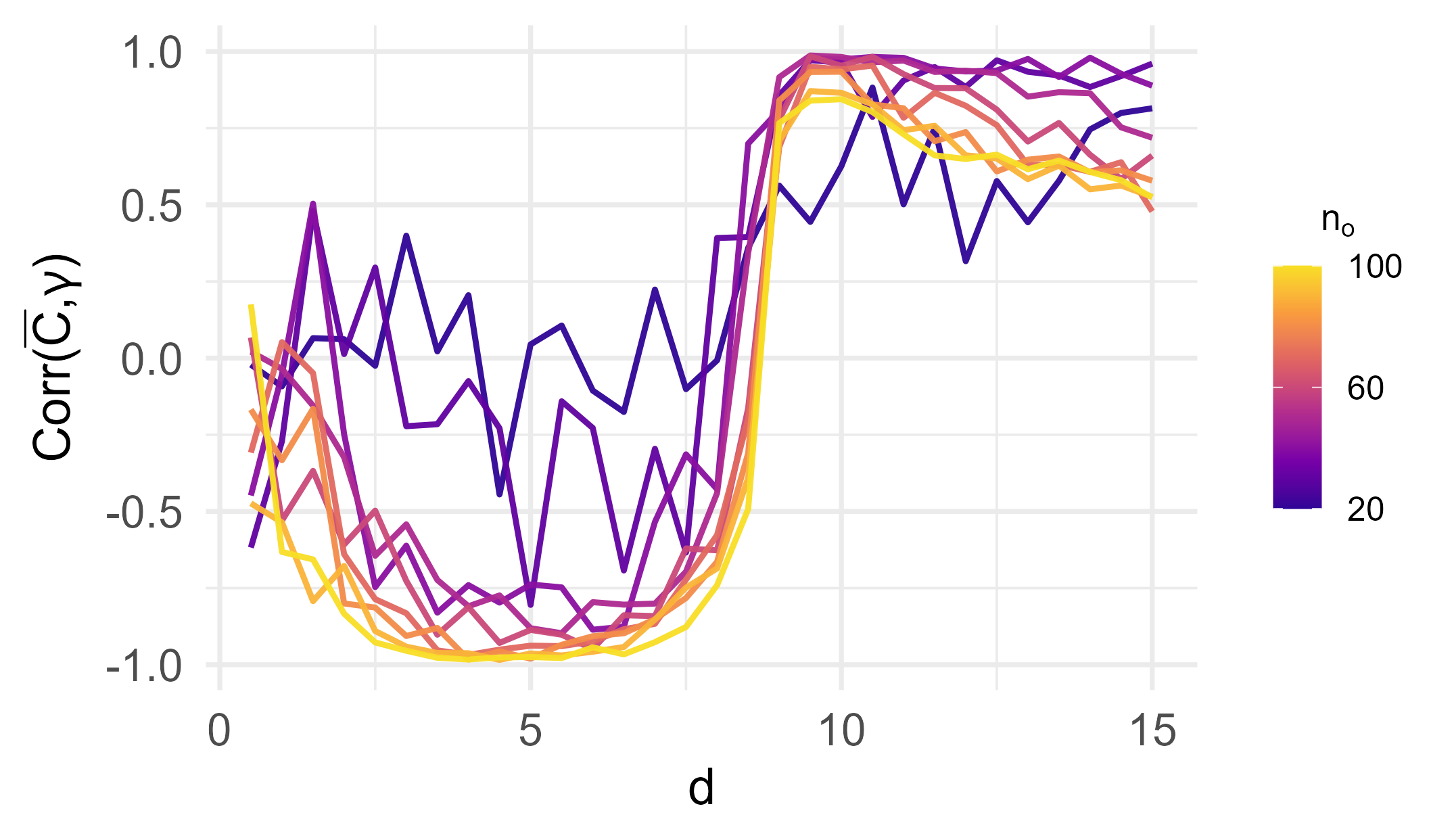}
 \caption{}
\end{subfigure}
\end{tabular}
\caption{
(a) Interaction Plots with the mean traversal cost $\bar{C}$ (averaged over obstacle number levels)
versus $\gamma$ values plotted for varying $d$ values under Strauss$(n, d, \gamma)$ regularity patterns,  
for the false obstacle only (top),  
true obstacle only (middle), and mixed obstacle (bottom) cases.  
(b) Correlation $\Corr(\bar{C}, \gamma)$ versus $d$ values for the corresponding three cases in the left column.
}
\label{fig:Lbar-vs-Strauss}
\end{figure}

In Figure~\ref{fig:Lbar-vs-Strauss}(a) (top),  
we observe that as the point pattern transitions from uniformity to regularity  
(i.e., as $\gamma$ decreases from 1 to 0),  
the mean traversal cost $\bar{C}$ remains largely unchanged for smaller $d$ values.  
This is because the Strauss$(n_F, d, \gamma)$ process, with small $d$,  
does not enforce sufficient separation between obstacle disks,  
resulting in substantial overlap and leaving ample false-free space for traversal.  
However, as $d$ increases, regularity becomes more pronounced.  
For moderate $d$ values (e.g., $d \approx 3r/2 \approx 7$),  
increasing regularity (i.e., decreasing $\gamma$) raises $\bar{C}$.  
In this regime, disk centers are spaced far enough  
to collectively occupy more of the traversal region, hindering navigation.  
For large $d$ values ($d \gtrsim 9 = 2r$), this trend reverses:  
disk centers are placed so far apart that regular patterns  
create more clutter-free space than uniform ones,  
resulting in decreased $\bar{C}$ as $\gamma$ decreases.

Figure~\ref{fig:Lbar-vs-Strauss}(b) (top) reflects these effects through correlation trends.  
The correlation $\Corr(\bar{C}, \gamma)$ is negative for small to moderate $d$,  
indicating increasing regularity raises traversal cost.  
In contrast, for large $d$, the correlation turns positive  
due to the increase in spacing reducing obstruction.  
This interaction also depends on the number of false obstacles $n_F$.  
For small $n_F$ (e.g., $n_F \leq 30$),  
correlation is weak across all $d \leq 9$ due to insufficient obstacle density.  
Sparse disks—and possible overlaps—fail to constrain movement significantly.  
In contrast, for $n_F \geq 50$ and moderate $d$ (between 4 and 9),  
the negative correlation strengthens.  
In these settings, increasing $\gamma$ (i.e., reducing regularity)  
increases overlap and clutter-free space,  
thereby lowering $\bar{C}$.  
Nonetheless, even though correlation trends are clear,  
the absolute value of $\bar{C}$ is generally highest for moderate $d$ values across the $\gamma$ range,  
as seen in Figure~\ref{fig:Lbar-vs-Strauss}(a) (top).

As $d$ appears to substantially influence how traversal cost varies with $\gamma$,  
we further explore this relationship in Figure~\ref{fig:Lbarvsd-at-gamma}(a),  
which shows mean traversal cost versus $d$ for fixed $\gamma$ values,  
averaged over obstacle number levels.  
A clear concave-down pattern emerges,  
with traversal cost peaking near $d \approx 1.5\,r$.  
This suggests that intermediate spacing between obstacles  
leads to the most obstructive configurations.  
The curvature is more pronounced at lower $\gamma$ values—i.e., under stronger regularity—  
since $\gamma$ controls the number of disk pairs closer than $d$.  
At moderate $d$ values (around 7), disks are typically neither overlapping nor too widely spaced,  
which maximizes obstruction.  
For $d > 9 = 2r$, disks are too dispersed to significantly hinder traversal,  
and mean costs drop—especially for low $\gamma$ values,  
where regularity amplifies the spacing effect.  
At the other extreme, when $\gamma=1$,  
the Strauss process approaches a uniform distribution,  
so the influence of $d$ largely vanishes,  
resulting in flat traversal cost trends across $d$.

\begin{figure}[htb]
\centering
\begin{minipage}{0.45\textwidth}
    \centering
    \includegraphics[width=\linewidth]{Strauss_Lbarvsd_gamma_clut.png}\\
    (a)
\end{minipage}\hfill
\begin{minipage}{0.45\textwidth}
    \centering
    \includegraphics[width=\linewidth]{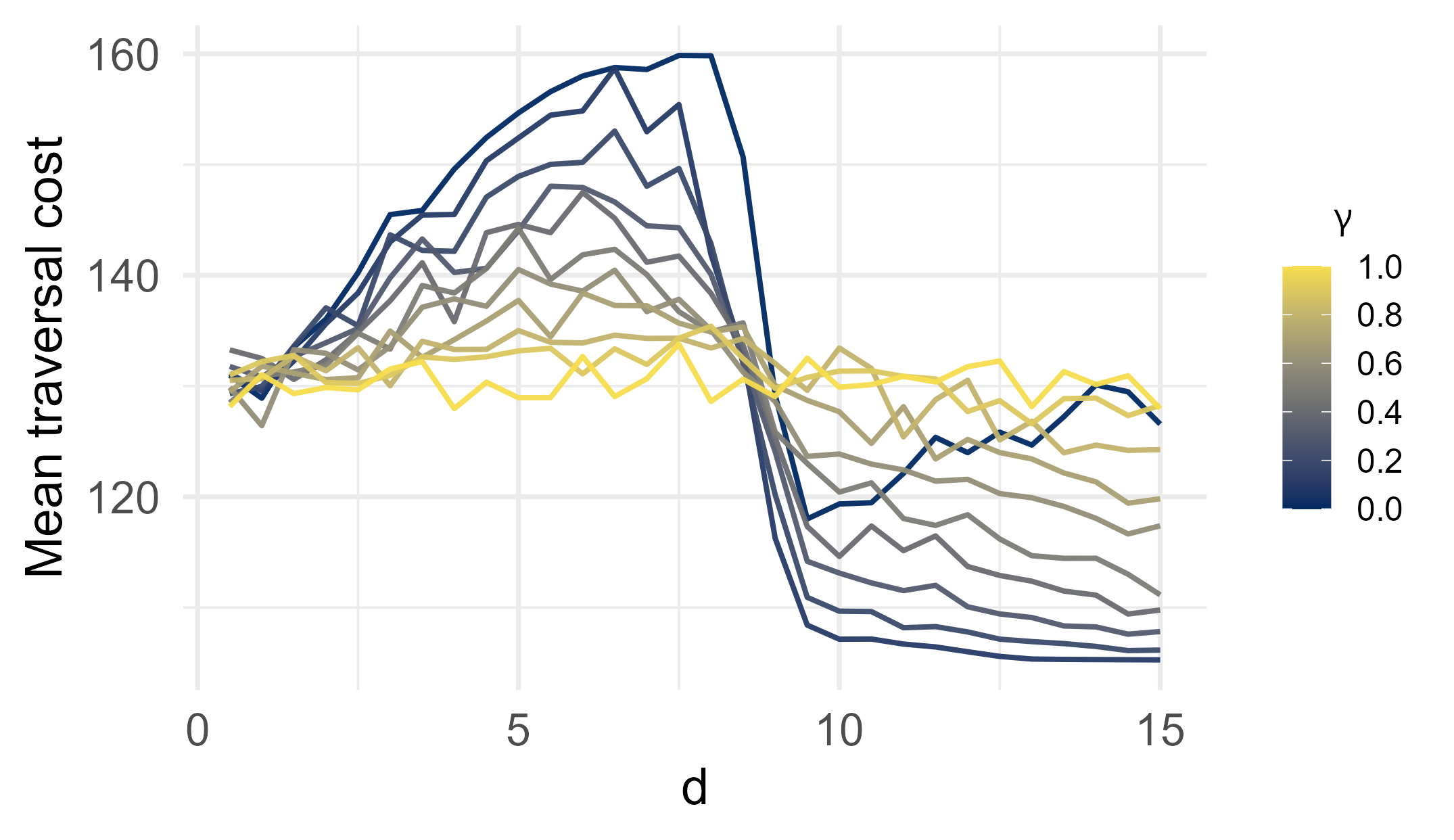}\\
    (b)
\end{minipage}

\begin{minipage}{0.45\textwidth}
    \centering
    \includegraphics[width=\linewidth]{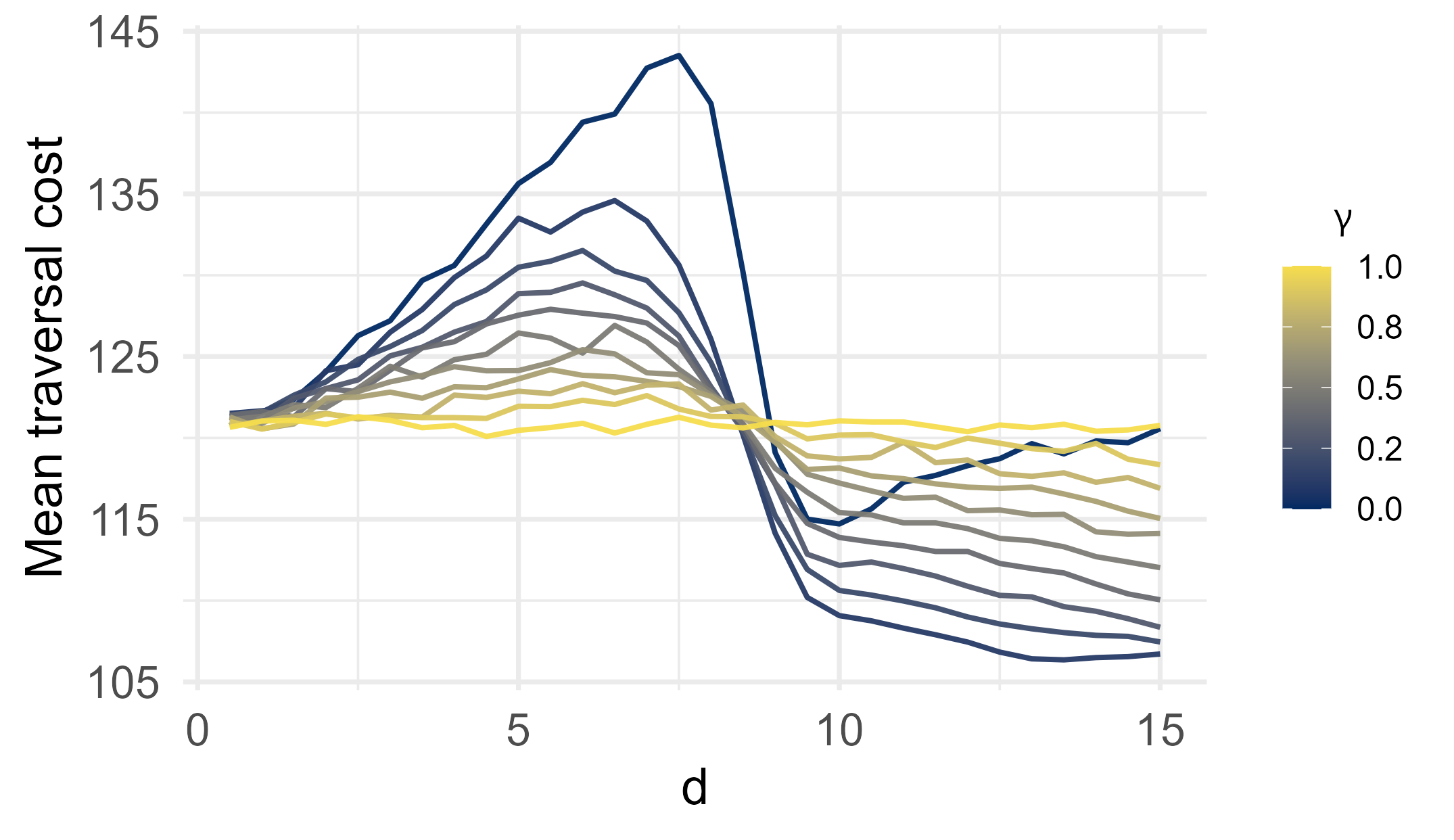}\\
    (c)
\end{minipage}

\caption{
Interaction plots with the mean traversal cost $\bar{C}$ (averaged over obstacle number levels)
versus $d$ under Strauss$(n, d, \gamma)$ regularity patterns.  
(a) False-obstacle-only case.  
(b) True-obstacle-only case.  
(c) Mixed-obstacle case.}
\label{fig:Lbarvsd-at-gamma}
\end{figure}

Figure~\ref{fig:contour-Lbarvsd-gamma}(a) reinforces these observations  
with contour plots of mean traversal cost over the $(\gamma, d)$ plane.  
For lower $n_F$ (not shown), traversal costs vary modestly,  
but for higher $n_F$ values,  
the range of cost values expands,  
indicating stronger dependence on $d$ and $\gamma$.  
Highest traversal costs occur at small $\gamma$ (e.g., $\gamma < 0.15$) and moderate $d$ values,  
with the optimal $d$ range shifting slightly lower as $n_F$ increases.  
In contrast, traversal is most efficient (lowest cost)  
when $d \gtrsim 10$ and $\gamma \approx 0.15$,  
yielding wide clutter-free regions under sparse regular arrangements.  
Interestingly, around $d = 2r = 9$, traversal cost becomes nearly invariant to $\gamma$,  
and resembles that under uniform patterns,  
highlighting a structural transition point in obstacle configuration.

\begin{figure}[!htb]
\centering
\begin{minipage}{0.45\textwidth}
    \centering
    \includegraphics[width=\linewidth]{Strauss_Contour_clut.png}\\
    (a)
\end{minipage}\hfill
\begin{minipage}{0.45\textwidth}
    \centering
    \includegraphics[width=\linewidth]{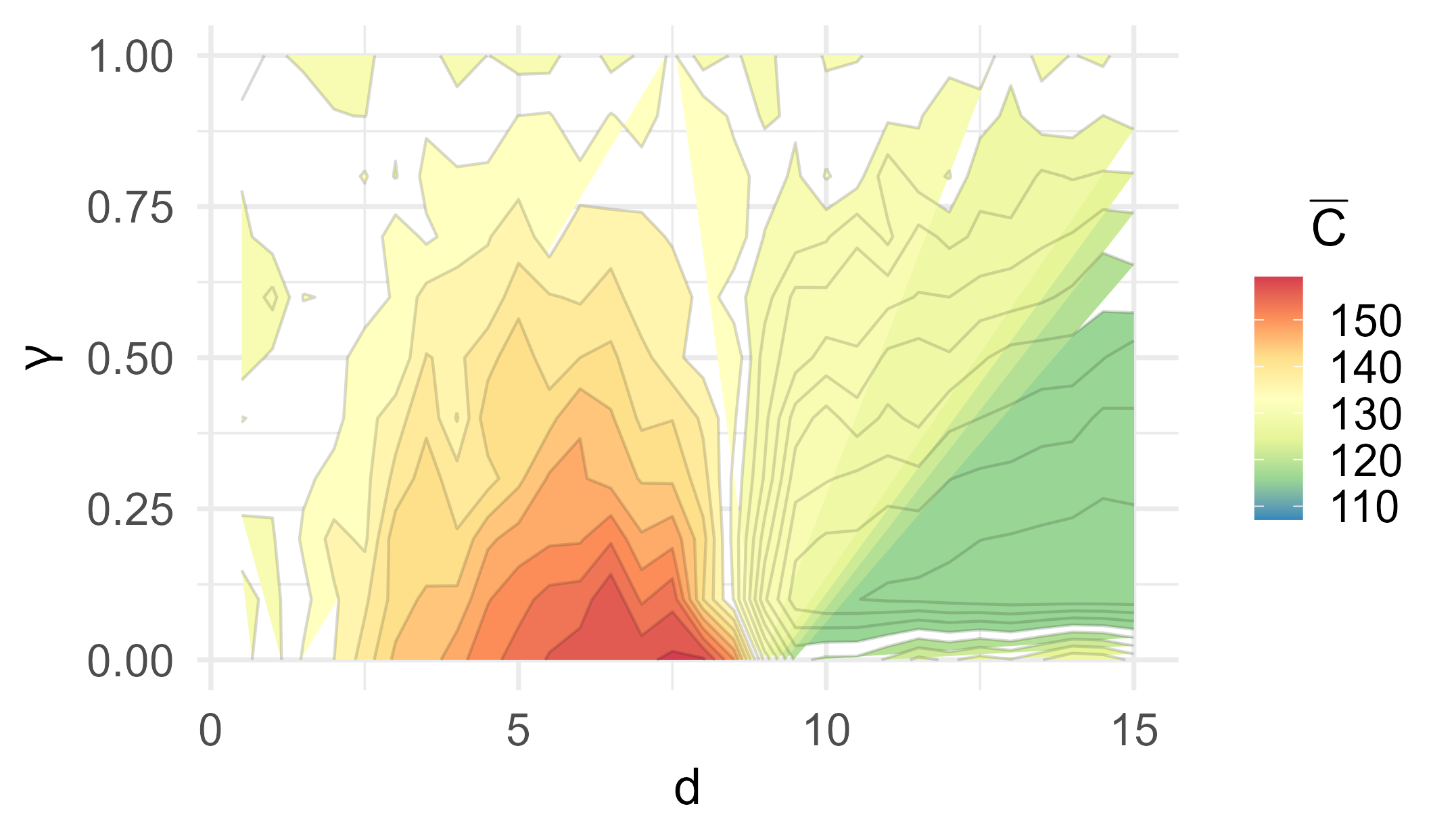}\\
    (b)
\end{minipage}

\begin{minipage}{0.45\textwidth}
    \centering
    \includegraphics[width=\linewidth]{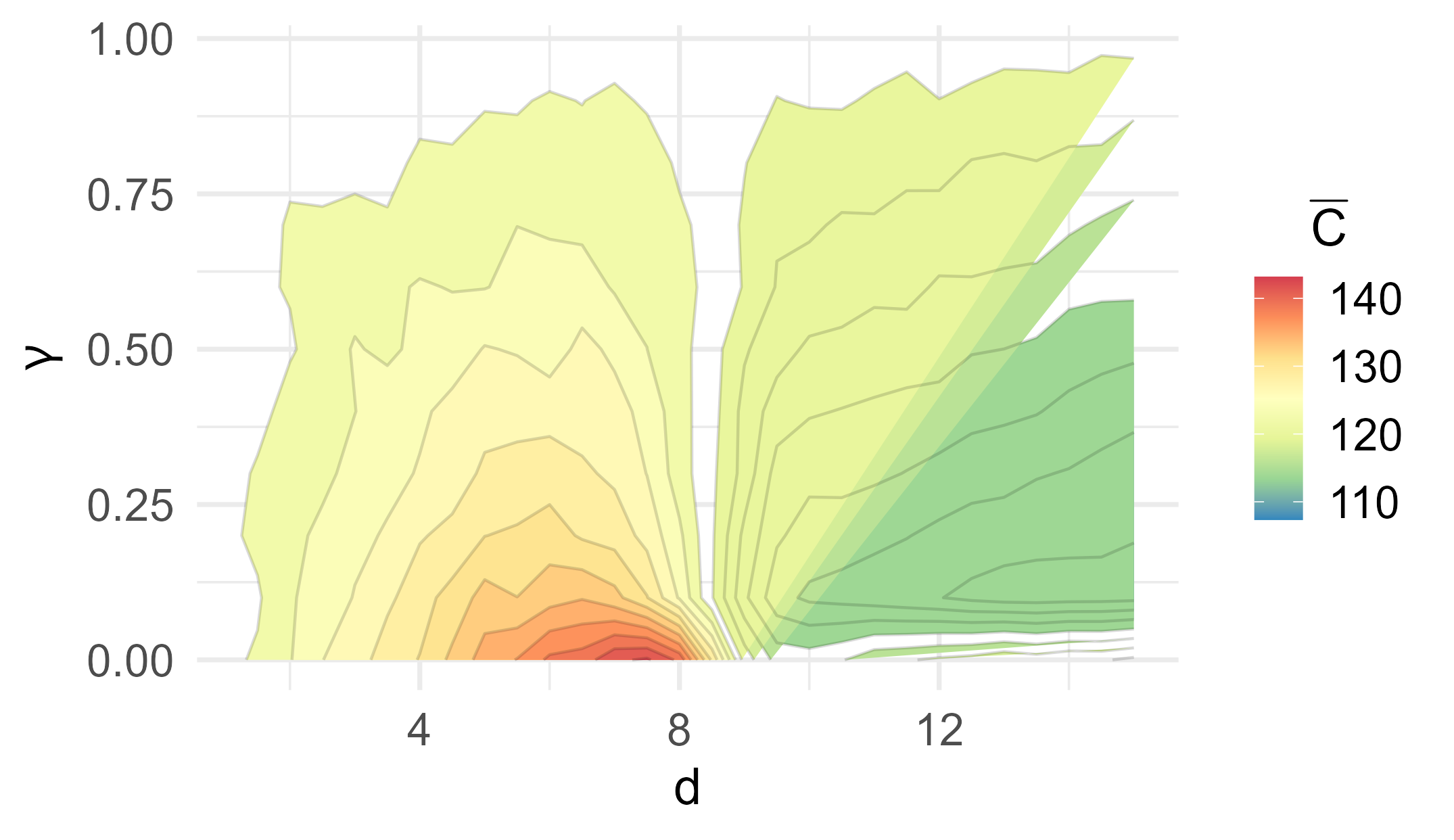}\\
    (c)
\end{minipage}

\caption{Filled contour plots of mean traversal cost (averaged over all obstacle number levels)  
for $\gamma$ and $d$ values under the Strauss$(n,d,\gamma)$ regularity pattern.  
(a) False-obstacle-only case.  
(b) True-obstacle-only case.  
(c) Mixed-obstacle case.  
The color index for increasing mean traversal cost is provided by each plot.}
\label{fig:contour-Lbarvsd-gamma}
\end{figure}

In Figures~\ref{fig:Lbar-vs-Strauss}–\ref{fig:contour-Lbarvsd-gamma},  
we observe distinct trends in mean traversal cost  
as functions of the covariates $\gamma$, $d$, and $n_F$.  
To better quantify these relationships and assess interaction effects,  
we fit a multiple linear regression model,  
treating traversal cost as the response variable  
and the remaining variables as numerical predictors.  
We exclude the number of disambiguations from the model  
since it is only observed post-traversal and is itself a function of the other covariates.  
The interaction plots suggest non-parallel trends across levels of each variable,  
indicating the presence of interaction effects.  
Because traversal cost exhibits strong right skewness and outliers—  
especially due to high disambiguation counts—  
we apply a logarithmic transformation.  
Even after transformation, the data deviate from normality (per Lilliefors' test),  
with outlier rates of approximately 30\% by Cook’s distance  
and 2\% by $z$-score criterion ($|z| > 3$) \citep{kuhn2013}.  
This motivates the use of a robust linear regression model with $M$-estimation \citep{Huber:1981}.  
We include second-order terms and all two-way interactions among $\gamma$, $d$, and $n_F$,  
but remove the $\gamma \times n_F$ interaction due to insignificance.  
This suggests traversal cost trends in $\gamma$ are parallel across $n_F$ levels.  
The robust model reduces residual standard error from 4.48 (OLS) to 3.78,  
validating the approach.  
Model fitting is done using the \texttt{rlm} function from the \texttt{MASS} package in \texttt{R},  
using Huber weights (with bisquare weights yielding similar results).

The resulting robust regression model is:
\begin{equation}
\label{eqn:rlm-false-regular}
\widehat C =  99.270 -4.21 \gamma + .33\,d +.1897n_F+  1.68\, \gamma^2-.025\,d^2 + .0006\,n_F^2  +
 .365\,\gamma\,d -.0064\, d\,n_F.
\end{equation}

While interpretation is nuanced due to interaction and quadratic terms, some dominant trends emerge:  
traversal cost decreases with $\gamma$ at small $d$ and increases at large $d$;  
it grows quadratically (concave-down) in $d$;  
and increases quadratically (concave-up) in $n_F$.

Although not intended for prediction, this model can offer exploratory insights.  
For predictive modeling under outlier presence, we also employ random forest (RF) regression \citep{breiman2001},  
using \texttt{randomForest} in \texttt{R}, with $C$ as response and $(n_F, \gamma, d)$ as predictors.  
After testing models with 100, 300, and 500 trees, we found 100 trees sufficient,  
with 1 variable tried at each split (as suggested by \cite{kuhn2013}), and default node size.  
Variable importance is assessed via average increase in out-of-bag (OOB) residuals when permuting a given predictor.  
Figure~\ref{fig:RF-varimp-plot-false-reg}(a) shows that $n_F$ is the most important predictor,  
followed by $d$ and then $\gamma$.  
The RF model explains 66.86\% of the variance with a mean squared residual of 19.04,  
indicating average prediction errors around 19 units.  
While RF is not recommended for accurate prediction in this setting,  
it is valuable for ranking variable importance and confirming regression insights.

\begin{figure}[htb]
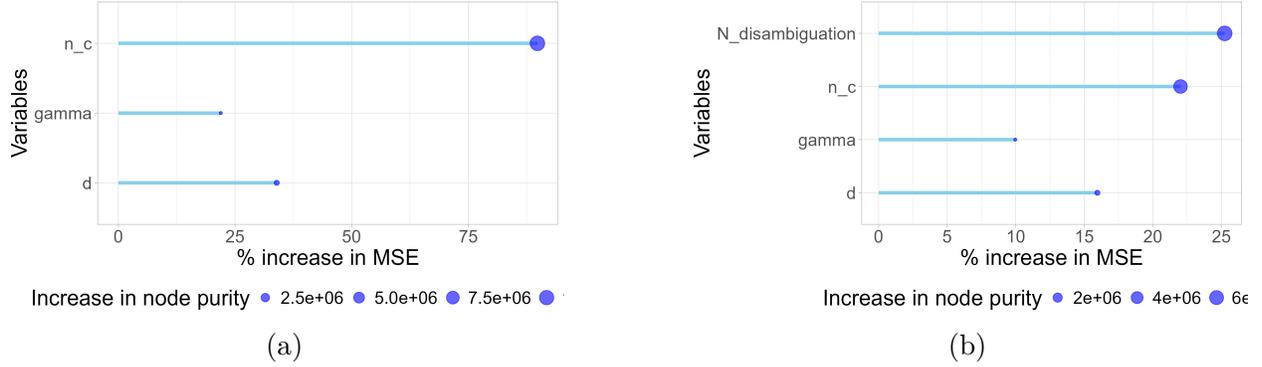

\centering
\begin{minipage}{0.45\textwidth}
    \centering
    \includegraphics[width=\linewidth]{RF_Clutter_Strauss.png}\\
    (a)
\end{minipage}\hfill
\begin{minipage}{0.45\textwidth}
    \centering
    \includegraphics[width=\linewidth]{RF_ClutterNdis_Strauss.png}\\
    (b)
\end{minipage}
\caption{Variable importance (decrease in mean square error) and node purity (residual sum of squares)  
for RF regression models under the Strauss$(n_F,d,\gamma)$ pattern.  
(a) Using $C$ as response with $(n_F, d, \gamma)$ as predictors.  
(b) Using $C$ as response with $(n_F, d, \gamma, N_{dis})$ as predictors.}
\label{fig:RF-varimp-plot-false-reg}
\end{figure}

Although $N_{dis}$ (number of disambiguations) is only observed after traversal and not usable in predictive modeling beforehand,  
we include it here for completeness.  
This post hoc model can help estimate average traversal cost once the obstacle layout and $N_{dis}$ are known.  
We fit a robust linear regression with $C$ as response and $\gamma$, $d$, $n_F$, $N_{dis}$  
(along with their squares and all two-way interactions) as predictors.  
No variables were eliminated in the selection process.  
The resulting model is:
\begin{multline}
\label{eqn:rlm-false-regular-Ndis}
\widehat C = 99.336 -2.16\,\gamma + .16\, d + .1965\,n_F + 3.135\,N_{dis} + .54\,\gamma^2 -.0125\,d^2 -.00044\,n_F^2 \\
- .0193\,N_{dis}^2 + .15\,\gamma\,d + .0185\, \gamma\,n_F - .71\,\gamma\,N_{dis} \\
- .0029\,d\,n_F + .043\,d\,N_{dis} + .0096\,n_F\,N_{dis}.
\end{multline}

We also fit a RF model using $(n_F, \gamma, d, N_{dis})$ as predictors.  
Variable importance and node purity (Figure~\ref{fig:RF-varimp-plot-false-reg}, right)  
indicate $N_{dis}$ is the most important predictor, followed by $n_F$, $d$, and $\gamma$.  
This model achieves a mean squared residual of 14.65 and explains 74.5\% of variance.

Given that $N_{dis}$ is post-traversal,  
it is more practical to model it directly as a function of pre-traversal covariates.  
To that end, we fit a Zero-Inflated Negative Binomial (ZINB) model with $N_{dis}$ as the response  
and $(\gamma, d, n_F)$ as predictors,  
using the \texttt{zeroinfl} function in \texttt{R} (\texttt{pscl} package).  
The count portion (negative binomial) uses $\gamma$ and $d$,  
while the zero-inflation (logit) part uses $n_F$.  
All predictors are statistically significant,  
and the model outperforms an intercept-only alternative (confirmed via chi-squared test).  
In the count model, coefficients for $\gamma$ and $d$ are -0.062 and -0.054, respectively,  
indicating that higher regularity and spacing reduce expected disambiguation count.  
In the logit model, the coefficient for $n_F$ is -0.0925,  
meaning that more false obstacles reduces the odds of zero disambiguations.  
Thus, increasing $n_F$ makes disambiguations more likely,  
while increasing $\gamma$ and $d$ tends to reduce them.  
These findings align with observed trends:  
traversal cost $C$ and $N_{dis}$ are strongly positively correlated,  
and both increase with decreasing $\gamma$ at moderate $d$,  
but decrease at large $d$ where obstacles are widely spaced.

Overall, for all false obstacle number levels, $d$ has the largest scale and dominates the interaction with $\gamma$:
\begin{itemize}
\item At small $d$, $\gamma$ has limited effect on $C$;
\item At moderate $d$, lower $\gamma$ (more regularity) increases $C$;
\item At large $d$, higher $\gamma$ (more uniformity) increases $C$;
\item For fixed $\gamma$, $C$ increases then decreases with $d$, 
yielding a concave-down profile most pronounced for low $\gamma$.
\end{itemize}

\subsubsection{True Obstacles from Uniform to Regular Patterns}
\label{sec:OOP-obs-unif2reg}

We now consider a more realistic scenario where OPA places 
only \emph{true obstacles} (e.g., mines) in the region.  
The simulation setup is identical to Section~\ref{sec:OOP-clut-unif2reg},  
with the number of true obstacles $n_T$ replacing $n_F$.  
We follow the same analysis approach,  
examining how mean traversal cost varies with $\gamma$, $d$, and $n_T$ under the Strauss$(n_T, d, \gamma)$ process.  
As expected, for fixed values of $\gamma$, $d$, and obstacle count,  
the mean traversal cost is higher in this case compared to the 
false obstacles only setting (Figure~\ref{fig:Lbar-vs-Strauss}(a) (middle).  
This is because when NAVA disambiguates an obstacle and finds it to be true, 
it must reroute, thereby increasing traversal cost.  
Despite this shift in magnitude, the overall trends in traversal cost and 
its correlation with $\gamma$ are consistent with those in the false obstacle case.

Figure~\ref{fig:Lbarvsd-at-gamma}(b) shows that mean traversal cost again follows 
a concave-down pattern in $d$ for each $\gamma$, as before, but at higher cost levels.  
Similarly, the contour plot in Figure~\ref{fig:contour-Lbarvsd-gamma}(b) 
reflects the same $(\gamma, d)$ dependence as in the false obstacle case, 
with elevated costs due to rerouting upon encountering true obstacles.

We also present the (filled) contour plots of mean traversal cost averaged over all $n_T$ values for  
$\gamma$ and $d$ values in Figure \ref{fig:contour-Lbarvsd-gamma}(b),  
which conveys a similar result as in the false-only case in Section \ref{sec:OOP-clut-unif2reg}.

To quantify these effects,  
we fit a robust linear regression model with $C$ as the response and $n_T$, $d$, $\gamma$, 
and their quadratic and two-way interaction terms as predictors.  
During model selection, $n_T$ and its interaction with $\gamma$ were excluded due to insignificance.  
The robust model substantially reduces residual standard error from 49.53 (OLS) to 7.46.  
The estimated model is:
\begin{equation}
\label{eqn:rlm-true-regular}
\widehat C = 100.108 -10.91 \,\gamma  + 1.082\,d + 4.73\,\gamma^2 -0.042\,d^2 + 0.0067\,n_T^2 + 0.86\,\gamma\,d - 0.03\,d\,n_T.
\end{equation}
Dominant effects mirror those observed earlier:  
cost decreases with $\gamma$ at small $d$ and increases at large $d$,  
follows a concave-down trend in $d$,  
and exhibits a quadratic increase with $n_T$.

Fitting the RF regression of \cite{breiman2001} with traversal cost $C$ as the response and $n_T$, $\gamma$, and $d$ as predictors,  
we find the variables ranked in decreasing order of importance and node purity as $n_T$, $d$, and $\gamma$.  
This mirrors the ordering observed in the false obstacles only setting.

Next, we fit a robust linear regression model using $C$ as the response and $(\gamma, d, n_T, N_{dis})$,  
their squares, and all two-way interactions as predictors.  
Only the $N_{dis}^2$ term is eliminated in model selection.  
The resulting model is:
\begin{multline}
\label{eqn:rlm-true-regular-Ndis}
\widehat C = 99.525 -6.11\,\gamma + .416\, d + .176\,n_T + 18.88\,N_{dis} + 2.15\,\gamma^2 -.024\,d^2 -.0012\,n_T^2\\
    + .43\,\gamma\,d + .03\,\gamma\,n_T  -.90\,\gamma\,N_{dis}
    -.011\,d\,n_T + .005\,d\,N_{dis} + .047\,n_T\,N_{dis}.
\end{multline}

In the corresponding RF model using $n_T$, $\gamma$, $d$, and $N_{dis}$ as predictors,  
variable importance rankings again highlight $N_{dis}$ as the most influential,  
followed by $n_T$, $d$, and $\gamma$,  
consistent with results from the false-only case.

We also model $N_{dis}$ directly using a ZINB regression with $(\gamma, d, n_T)$ as predictors.  
Here, the count model uses $\gamma$ and $d$,  
while the zero-inflation (logit) part uses $n_T$.  
All predictors are statistically significant.  
In the count portion, the expected change in $\log N_{dis}$ is $-0.08$ per unit increase in $\gamma$ and $-0.024$ per unit increase in $d$.  
In the logit portion, the log-odds of observing zero disambiguations decreases by $0.067$ per additional true obstacle,  
indicating disambiguations become more likely as $n_T$ increases.  
These trends parallel those observed in the false-only case (Section~\ref{sec:OOP-clut-unif2reg}).

The influence of $\gamma$ and $d$ is weaker when $n_T$ is small (e.g., $n_T \le 50$),  
as obstacle-free or sparsely populated paths are still accessible,  
allowing NAVA to traverse with few or no disambiguations.  
As $n_T$ increases ($n_T \ge 60$), the traversal region becomes more saturated with obstacles,  
reducing navigable paths and amplifying the effects of $\gamma$ and $d$ on traversal cost.  
For large $n_T$, the impact of $\gamma$ and its interaction with $d$ plateaus.  
With the region nearly fully obstructed, obstacle distribution becomes less critical—  
most paths are blocked regardless of regularity,  
and NAVA often resorts to traversing the obstacle-free annular region around the study window.

\subsubsection{Mixed Obstacles from Uniform to Regular Patterns}
\label{sec:OOP-mix-unif2reg}

We now consider the intermediate scenario in which OPA places a mix of true and false obstacles,  
combining aspects of the two extreme settings examined in Sections~\ref{sec:OOP-clut-unif2reg} and~\ref{sec:OOP-obs-unif2reg}.

Obstacle locations follow a Strauss$(n, d, \gamma)$ process,  
where the total number of obstacles $n_o = n_T + n_F$ varies from 20 to 100 in steps of 10,  
and $(n_T, n_F) \in \{10, 20, \ldots, 90\}$.  
As in prior sections, $\gamma$ ranges from 0 to 1 in increments of 0.1, and $d$ ranges from 0.5 to 15 in steps of 0.5.  
A representative realization is shown in Figure~\ref{fig:sample-Strauss-false-mix}(b).  
In this setting, sensors produce more varied probabilistic markings, introducing greater uncertainty in NAVA’s disambiguation decisions.  
Figure~\ref{fig:Lbar-vs-Strauss}(a) (bottom) shows that mean traversal costs lie between the corresponding values  
from the false-only and true-only cases, given the same values of $n_o$, $\gamma$, and $d$.  
Correlation patterns in Figure~\ref{fig:Lbar-vs-Strauss}(b) (bottom) reflect similar directional trends,  
positioned between the extremes.

Figure~\ref{fig:Lbarvsd-at-gamma}(c) confirms that the relationship between traversal cost and $d$ retains the concave-down form  
seen in previous cases, again suggesting a quadratic dependence.  
Cost levels remain intermediate: higher than in the false-only case and lower than in the true-only case.

The filled contour plots of mean traversal cost averaged over all $n_o$ values  
for $(\gamma, d)$ pairs in Figure~\ref{fig:contour-Lbarvsd-gamma}(c)  
also exhibit patterns consistent with those observed in Sections~\ref{sec:OOP-clut-unif2reg} and~\ref{sec:OOP-obs-unif2reg}.

To quantify these effects, we fit a robust linear model using traversal cost $C$ as the response  
and $n_o$, $n_F$, $d$, and $\gamma$ (including their squares and all two-way interactions) as predictors.  
The switch from OLS to robust modeling reduces the residual standard error from 26.97 to 6.32.  
All variables are retained in model selection:
\begin{multline}
\label{eqn:rlm-mix-regular}
\widehat C = 99.545 -11.114 \,\gamma  + .95\,d + .10\,n_o +.087\,n_F + 4.32\,\gamma^2 -.054\,d^2 + .0038\,n_o^2 + .002\,n_F^2\\
+ .845\,\gamma\,d + .0185\,\gamma\,n_o -.014\,\gamma\,n_F -.02\,d\,n_o + .01\,d\,n_F - .005\,n_o\,n_F.
\end{multline}
Dominant trends include: a concave-up relationship with $\gamma$,  
a concave-down quadratic trend in $d$,  
an overall increase with $n_o$, and a decreasing effect from $n_F$,  
modulated by several interaction terms.

An RF regression with $C$ as the response and $n_o$, $n_F$, $d$, and $\gamma$ as predictors  
yields variable importance rankings (in decreasing order): $n_o$, $n_F$, $d$, and $\gamma$.

We also fit a robust linear model using $C$ as the response and predictors $(\gamma, d, n_o, n_F, N_{dis})$,  
including their squares and all two-way interactions.  
Only the $\gamma \times N_{dis}$ and $\gamma \times n_o$ terms are removed in model selection.  
The final model is:
\begin{multline}
\label{eqn:rlm-mix-regular-Ndis}
\widehat C = 99.645 -3.05\,\gamma + .27\, d + .245\,n_o - .048\,n_F + 4.535\,N_{dis} + .953\,\gamma^2 -.02\,d^2 -.00025\,n_o^2 \\ 
+ .0043\,n_F^2 + .42\,N_{dis}^2 + .335\,\gamma\,d + .0044\, \gamma\,n_F -.0051\,d\,n_o + .0012\,d\,n_F + .024\,d\,N_{dis} \\ 
-.004\,n_o\,n_F + .1086\,n_o\,N_{dis} -.177\,n_F\,N_{dis}.
\end{multline}

The corresponding RF model using $(n_T, \gamma, d, N_{dis})$ as predictors confirms earlier trends,  
ranking variables by importance as: $N_{dis}$, $n_F$, $n_o$, $d$, and $\gamma$.  
The model achieves a mean squared residual of 164.17 and explains 82.86\% of the variance.

To further assess $N_{dis}$, we fit a ZINB model with $(\gamma, d, n_o, n_F)$ as predictors.  
The count portion includes $\gamma$, $d$, and $n_F$, while the zero-inflation part (logit) uses $n_o$.  
All predictors are statistically significant.  
In the count portion, expected changes in $\log N_{dis}$ per unit increase are:  
$-0.065$ for $\gamma$, $-0.07$ for $d$, and $-0.003$ for $n_F$.  
In the zero-inflation model, each unit increase in $n_o$ reduces the log-odds of observing zero disambiguations by 0.1,  
indicating denser obstacle fields increase the likelihood of disambiguation events.  
These results align with previous findings:  
$N_{dis}$ and traversal cost decrease with increasing regularity and spacing,  
but increase with the number of obstacles.

\begin{remark}
\label{rem:nc-vs-nt-in-models}
(\textbf{Using $n_F$ Versus $n_T$ in Models with Mixed-Type Obstacles})\\
In the mixed obstacle case, note that $n_o = n_F + n_T$.  
Thus, models in Equations~\eqref{eqn:rlm-mix-regular} and~\eqref{eqn:rlm-mix-regular-Ndis}  
can be reparameterized using $n_T$ in place of $n_F$.  
However, joint interpretation of $n_o$ and $n_F$ requires caution:  
an increase in $n_F$ (holding $n_o$ fixed) implies a decrease in $n_T$, and vice versa.  
Hence, a positive coefficient for $n_F$ on $C$ corresponds to a negative effect of $n_T$ on $C$,  
and similarly for $N_{dis}$.  
For example, the observed negative effect of $n_F$ on traversal cost  
implies that increasing $n_T$ leads to higher costs, reinforcing the asymmetry in obstacle impact.
\end{remark}

\subsection{OOP with Uniform to Clustered Obstacle Patterns}
\label{sec:OOP-unif2clust}

\subsubsection{False Obstacles from Uniform to Clustered Patterns}
\label{sec:OOP-clut-unif2clust}

We now consider the false-only case where the obstacle pattern transitions 
from uniformity to clustering, modeled using the Mat\'{e}rn$(\kappa, r_0, \mu)$ point process.

A representative realization from this setting, 
generated with parameters Mat\'{e}rn$(\kappa = 2, r_0 = 15, \mu = 10)$, 
is shown in Figure~\ref{fig:sample-Matern-false}.

\begin{figure} [!ht]
\centering
\includegraphics[width=.5\textwidth]{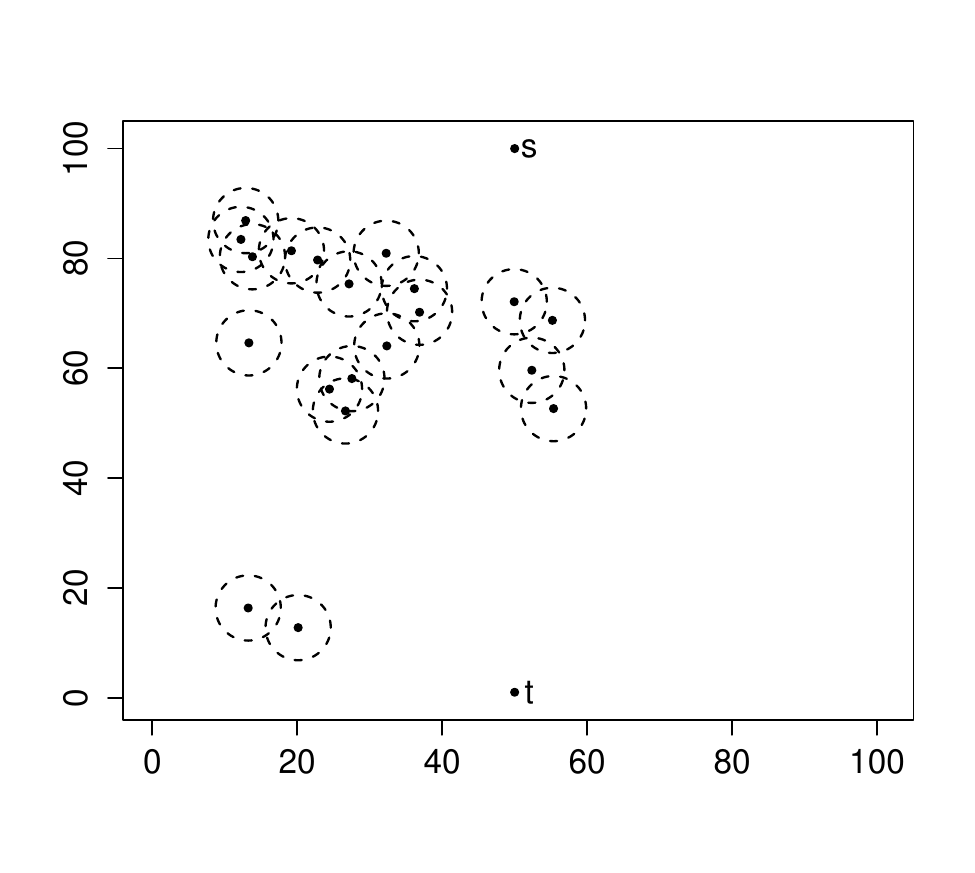}
\caption{A realization of the clustered obstacle pattern in the false-only case.
False obstacles are drawn from a Mat\'{e}rn$(\kappa=2, r_0=15, \mu=10)$ process.}
\label{fig:sample-Matern-false}
\end{figure}

We investigate trends in the mean traversal cost $\bar{C}$ as a function of 
the Mat\'{e}rn clustering parameters $\kappa$ (number of parent points) and $r_0$ (cluster radius).  
Figure~\ref{fig:L-vs-Matern}(a) (top) displays $\bar{C}$ versus $r_0$ for various values of $\kappa$, 
averaged over levels of false obstacle number $n_F$.

For moderate to large cluster radii ($r_0 \ge 15$), mean traversal cost remains relatively constant, 
as loosely clustered obstacles approximate a near-uniform pattern.  
However, for smaller values of $r_0$ (i.e., $r_0 \lesssim 15$), $\bar{C}$ decreases sharply. 
In this tightly clustered regime, 
the formation of obstacle-dense pockets creates wide navigable corridors elsewhere in the domain, 
enabling NAVA to traverse at reduced cost with fewer disambiguations.

This decrease in $\bar{C}$ is more pronounced at higher $n_F$ values, 
where obstacle density amplifies the effect of clustering.  
Furthermore, for fixed $r_0$, increasing $\kappa$ leads to higher $\bar{C}$:
more clusters result in greater dispersion of obstacles throughout the region, 
narrowing viable traversal paths.

Overall, as $r_0$ increases, the traversal cost $\bar{C}$ also tends to increase—signaling 
a transition from strongly clustered configurations to more dispersed, uniform-like layouts.  
This behavior is visualized in the filled contour plot in Figure~\ref{fig:L-vs-Matern}(b) (top), 
where $\bar{C}$ increases with both $\kappa$ and $r_0$, peaking around $\kappa \approx 12$ and $r_0 \approx 50$.

\begin{figure} [!ht]
\centering
\begin{tabular}{cc}
\begin{subfigure}[b]{0.45\textwidth}
\centering
\includegraphics[width=\textwidth]{Matern_Lbarvsradius_kappa_clut.png}
\includegraphics[width=\textwidth]{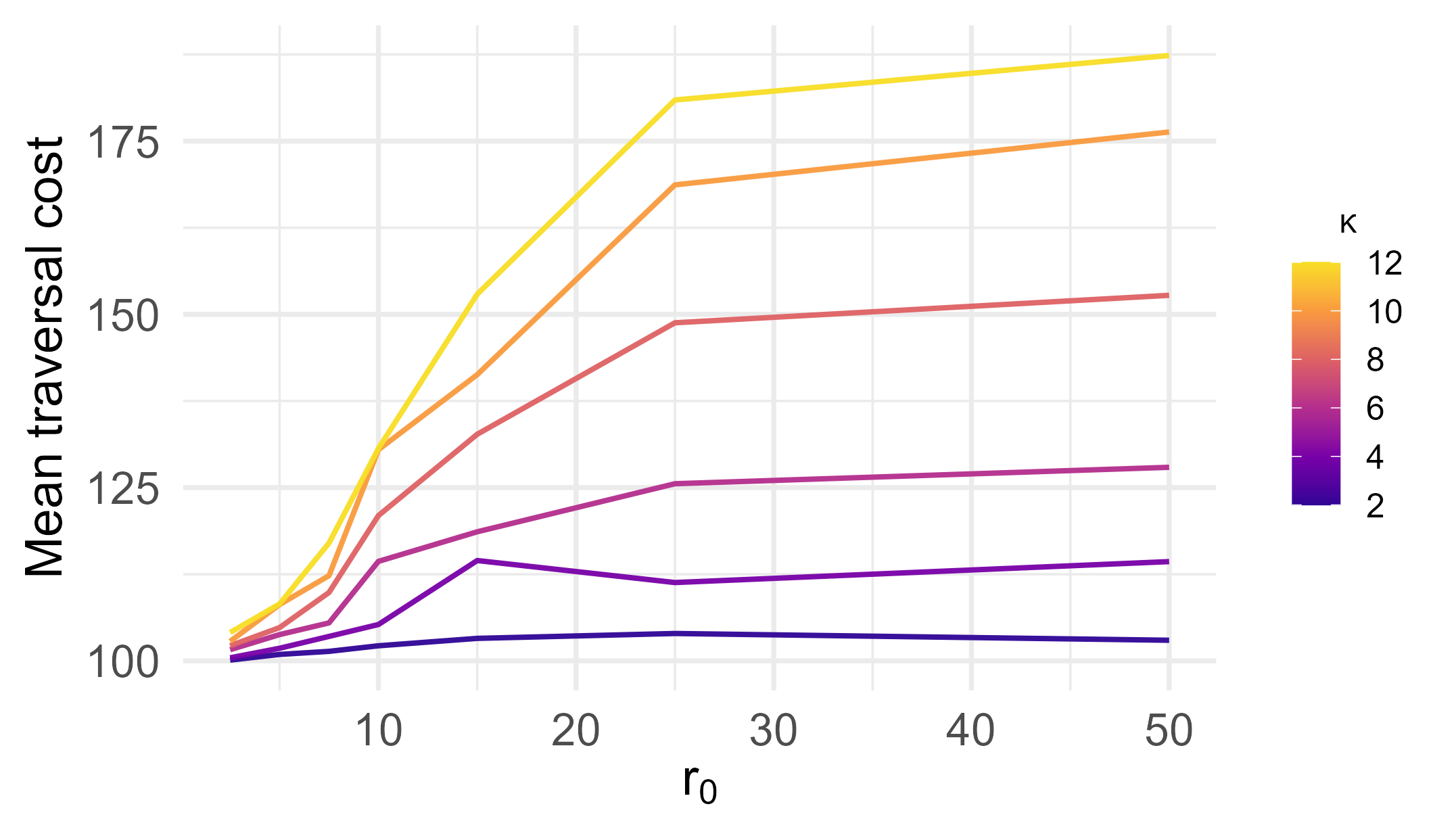}
\includegraphics[width=\textwidth]{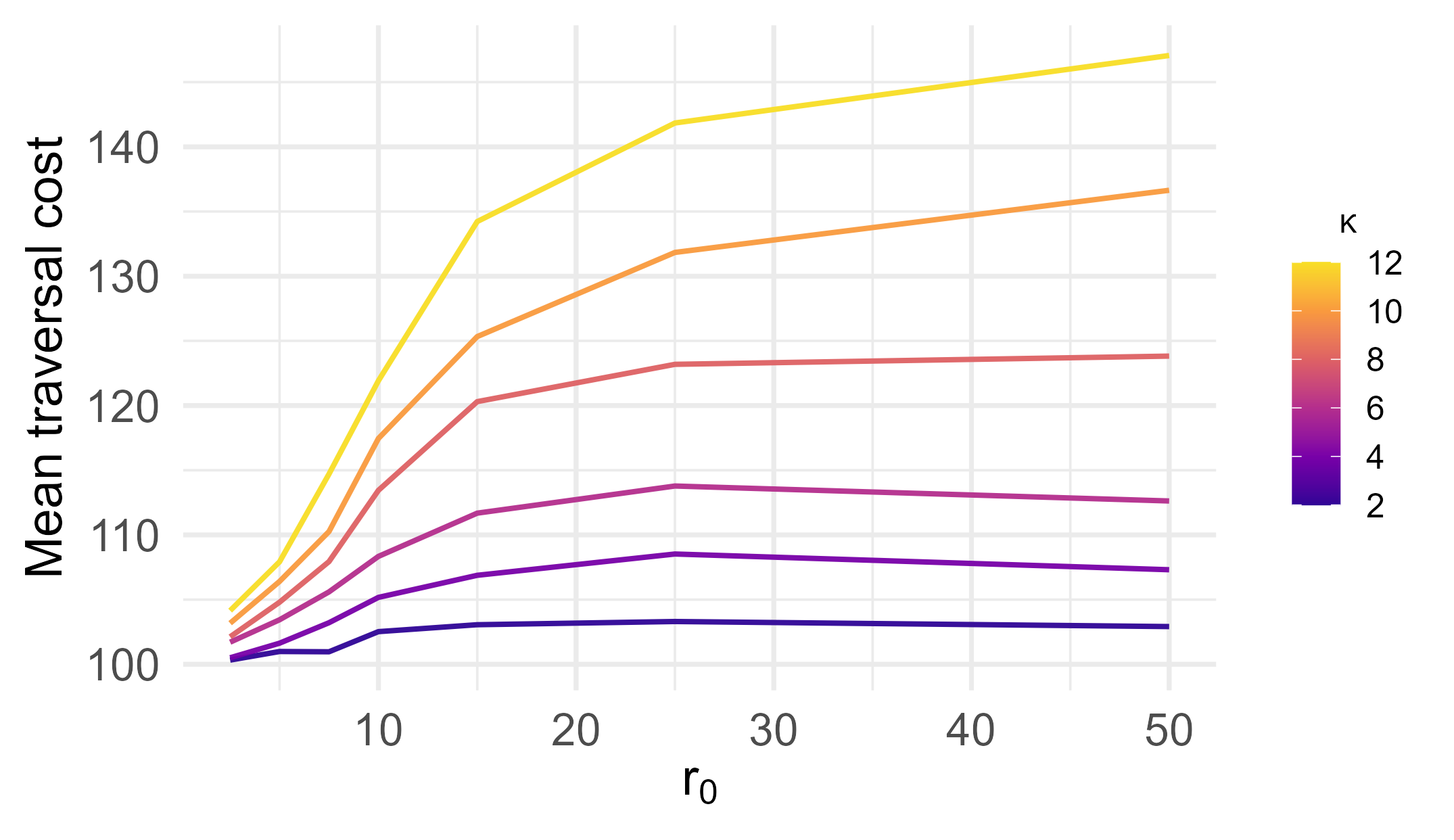}
\caption{}
\end{subfigure}
\hfill
\begin{subfigure}[b]{0.45\textwidth}
\centering
\includegraphics[width=\textwidth]{Matern_Contour_clut.png}
\includegraphics[width=\textwidth]{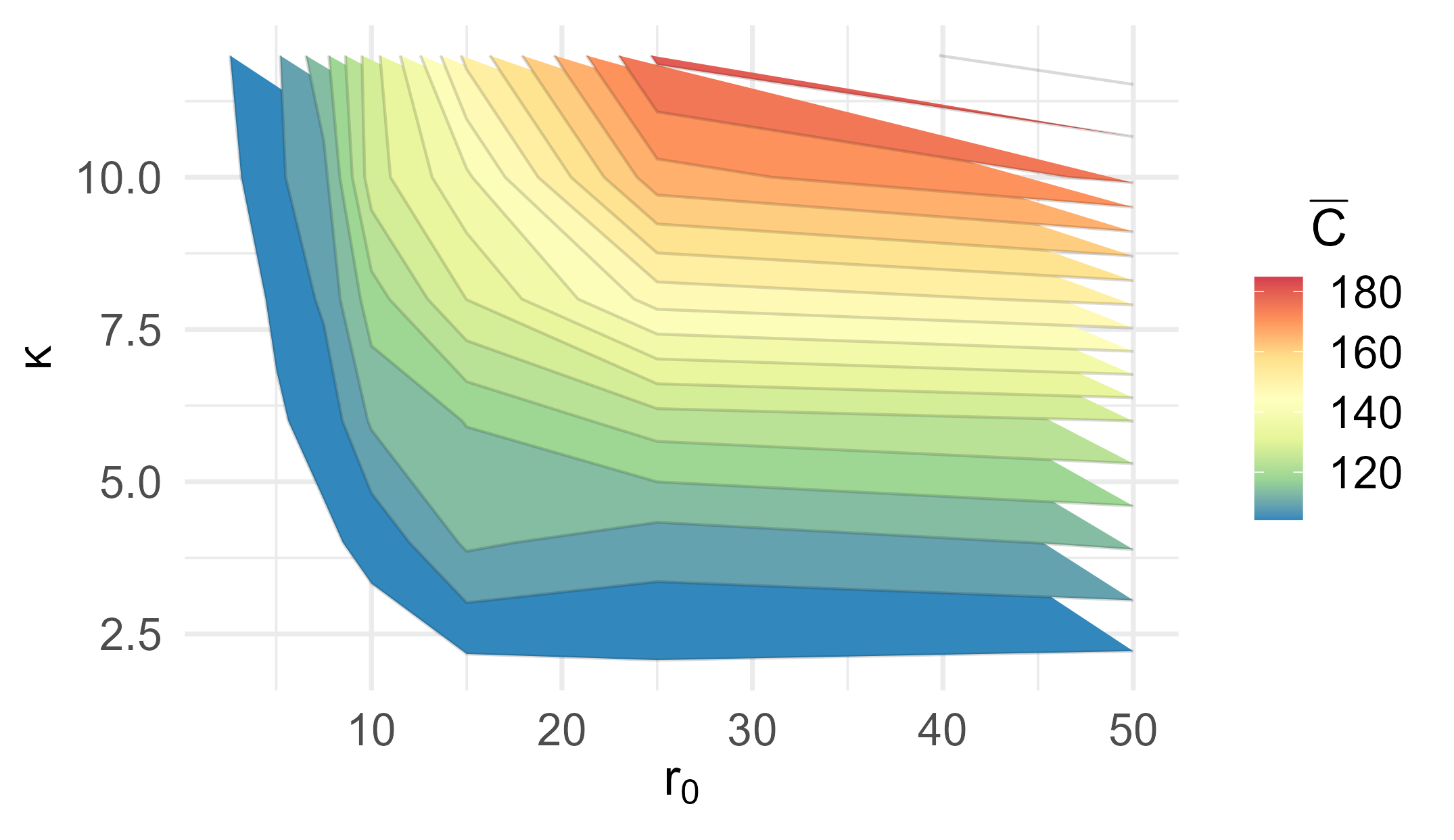}
\includegraphics[width=\textwidth]{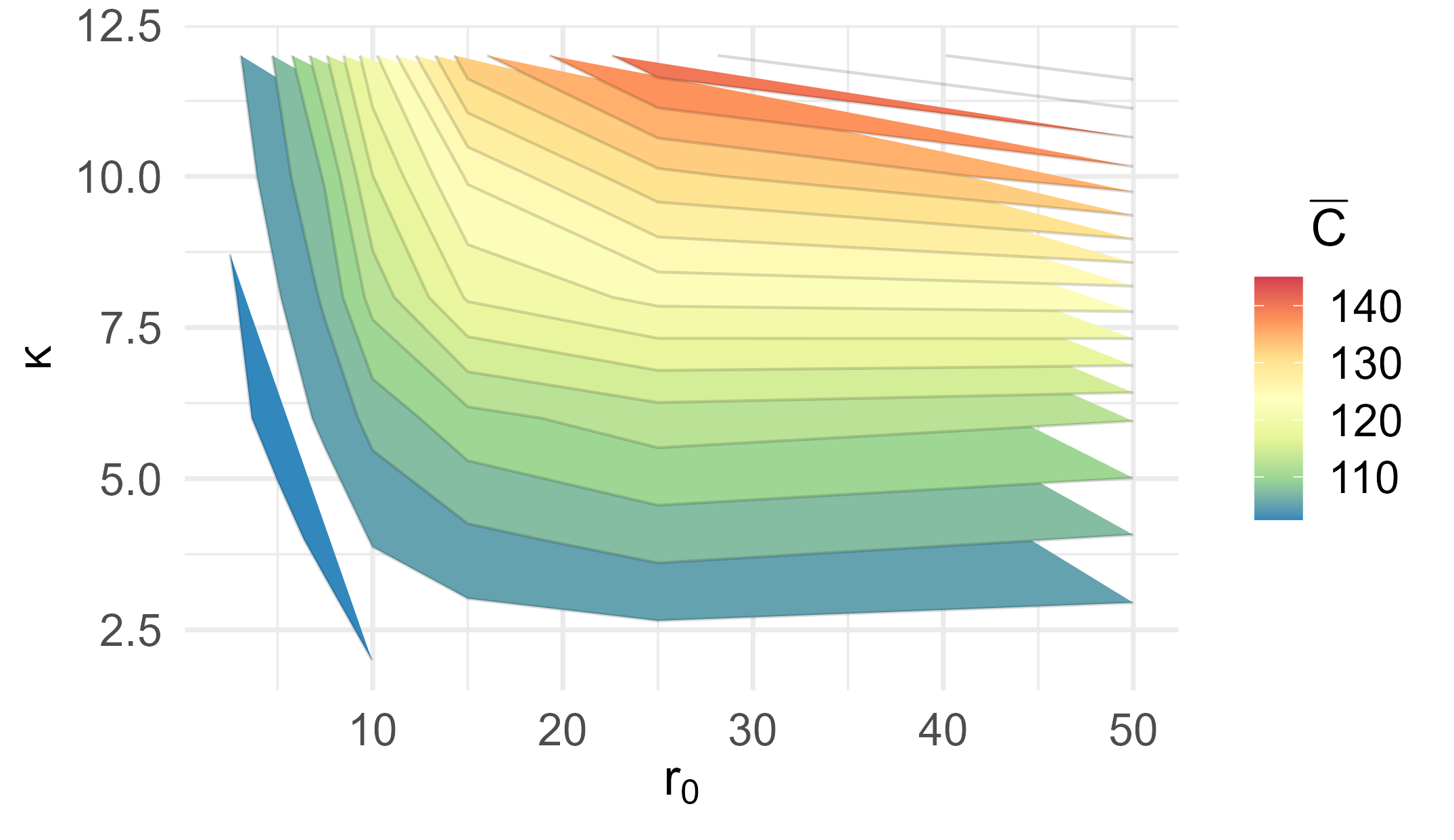}
\caption{}
\end{subfigure}
\end{tabular}
\caption{(a) Interaction Plots with the mean traversal cost $\bar{C}$ (averaged over obstacle number levels)
versus cluster radius $r_0$ values plotted
for varying values of $\kappa$ under the Mat\'{e}rn$(\kappa, r_0, \mu)$ clustering pattern, 
in the false obstacle only (top), true obstacle only (middle), and mixed obstacles (bottom) cases.
(b) Filled contour plots of mean traversal cost $\bar{C}$ (averaged over all obstacle number levels) 
for $\kappa$ and $r_0$ values under the Mat\'{e}rn$(\kappa, r_0, \mu)$ clustering pattern 
for the corresponding three cases in the left column.}
\label{fig:L-vs-Matern}
\end{figure}

As in Section~\ref{sec:OOP-clut-unif2reg}, traversal cost $C$ exhibits strong right skew,  
motivating the use of a second-order robust linear regression model.  
We include all main effects and two-way interactions of the clustering parameters  
$\kappa$ and $r_0$ as predictors.  
During variable selection, the $\kappa^2$ term is eliminated,  
indicating a non-quadratic relationship between $\kappa$ and $C$.  
The residual standard error decreases from 6.88 (OLS) to 5.62 under the robust model,  
justifying its use.

Our estimated robust linear model takes the following form:
\begin{equation}
\label{eqn:rlm-false-cluster}
\widehat C = 95.996 + 0.65\,\kappa + 0.51\,r_0 - 0.01\,r_0^2 + 0.038\,\kappa\,r_0.
\end{equation}
The dominant trend in traversal cost $C$ is a positive, concave-down relationship with $r_0$,  
and an increasing effect of $\kappa$ that is amplified as $r_0$ increases.

We also fit a RF regression model with $C$ as the response  
and $\kappa$ and $r_0$ as predictors,  
following the approach in Section~\ref{sec:OOP-clut-unif2reg}.  
Variable importance rankings are $r_0$ followed by $\kappa$ (plot omitted).  
The RF model explains 47.82\% of the variance with a mean squared residual of 48.07.  
Hence, RF is not preferred for prediction in this setting.

For completeness, we also include the post-traversal covariate $N_{dis}$  
in a second robust regression model.  
This model can be used to estimate average traversal cost  
when $\kappa$, $r_0$, and $N_{dis}$ are known post hoc.  
We include main effects, squares, and two-way interactions among these variables,  
excluding only the insignificant $\kappa\,N_{dis}$ term during selection.  
Variable importance in the corresponding RF model (not shown)  
ranks $N_{dis}$ as most important, followed by $r_0$ and then $\kappa$.

\begin{multline}
\label{eqn:rlm-false-cluster-Ndis}
\widehat C = 97.82 + 0.544\,\kappa + 0.31\,r_0 + 6.75\,N_{dis} - 0.0056\,r_0^2 - 0.209\,N_{dis}^2 \\
+ 0.016\,\kappa\,r_0 - 0.038\,r_0\,N_{dis}.
\end{multline}

The positive $\kappa\,r_0$ interaction indicates that the effect of each variable  
on traversal cost intensifies with larger values of the other.  
Likewise, the negative $r_0\,N_{dis}$ interaction suggests  
that the cost contribution of disambiguations is dampened  
when obstacle clusters are more dispersed.

Due to the high proportion of zero disambiguations and the presence of overdispersion,  
we model $N_{dis}$ as the response and $\kappa$, $r_0$, and $n_F$ as predictors  
using a ZINB regression \citep{Zeileis2008}.  
The count component includes $\kappa$ and $r_0$,  
while the zero-inflation (logit) component uses $n_F$.

All predictors are statistically significant,  
indicating a significantly improved fit over an intercept-only model.  
In the count portion, the coefficients for $\kappa$ and $r_0$ are 0.198 and 0.042, respectively,  
implying that the expected change in $\log N_{dis}$ for a one-unit increase in $\kappa$  
is approximately 0.20, and for a one-unit increase in $r_0$ is 0.042,  
holding other variables constant.

In the logit portion, the coefficient for $n_F$ is $-0.07$,  
meaning that the log-odds of observing an excessive zero  
decrease with increasing false obstacle numbers.  
In other words, as the number of false obstacles increases,  
the likelihood of observing zero disambiguations declines.

Thus, disambiguation counts tend to rise with both $\kappa$ and $r_0$,  
consistent with the observed increase in traversal cost.  
Additionally, denser false obstacles makes disambiguations more likely.

\subsubsection{True Obstacles from Uniform to Clustered Patterns}
\label{sec:OOP-obs-unif2clust}

We consider the true obstacle only case under the same setting as in Section~\ref{sec:OOP-clut-unif2clust},  
using the same $\kappa$ and $r_0$ values,  
with the number of true obstacles $n_T$ set to $10\kappa$ on average.

Figure~\ref{fig:L-vs-Matern}(a) (middle) shows the mean traversal cost versus $r_0$ for each $\kappa$ value,  
averaged over the corresponding $n_T$ levels.  
The trends closely mirror those observed in the false-only case:  
traversal cost increases with both $\kappa$ and $r_0$,  
exhibiting a concave-down quadratic relationship with $r_0$.  
As before, higher $\kappa$ values result in greater obstruction  
due to wider spatial spread of obstacles.  
The contour plot in Figure~\ref{fig:L-vs-Matern}(b) (middle) confirms this pattern,  
with elevated cost levels across the $(\kappa, r_0)$ plane.  
Notably, except for very small $r_0 \approx 2.5$,  
the mean traversal cost is higher in the true-obstacle setting,  
due to rerouting upon encountering true obstacles.

Using the same modeling strategy as in Section~\ref{sec:OOP-clut-unif2clust},  
we fit a robust second-order linear model for $C$ with $\kappa$, $r_0$, and their interaction.  
Switching from OLS to robust regression reduces residual standard error from 34.78 to 7.68.  
The selected model (excluding the $\kappa^2$ term) is:
\begin{equation}
\label{eqn:rlm-true-cluster}
\widehat C = 96.23 + 0.27\,\kappa + 0.39\,r_0 - 0.013\,r_0^2 + 0.132\,\kappa\,r_0.
\end{equation}
This model reflects a positive and concave-down relationship with $r_0$,  
and an increasing effect of $\kappa$ that intensifies with $r_0$.  
The corresponding RF regression identifies $r_0$ as the most important predictor,  
followed by $\kappa$, consistent with the linear model.

We also fit a robust regression using $\kappa$, $r_0$, $N_{dis}$,  
and their squares and interactions as predictors.  
After eliminating the main effect of $\kappa$, the final model is:
\begin{multline}
\label{eqn:rlm-true-cluster-Ndis}
\widehat C = 97.36 + 0.516\,r_0 + 25.11\,N_{dis} + 0.04\,\kappa^2 - 0.012\,r_0^2 - 0.345\,N_{dis}^2 + 0.06\,\kappa\,r_0 \\
+ 0.447\,\kappa\,N_{dis} - 0.0756\,r_0\,N_{dis}.
\end{multline}
Variable importance in the RF model ranks predictors as $N_{dis}$, $r_0$, and $\kappa$,  
again highlighting the dominant influence of post-traversal disambiguation counts.

Finally, we model $N_{dis}$ using a ZINB regression with $\kappa$ and $r_0$ as count predictors  
and $n_T$ in the zero-inflation (logit) component.  
Here, only $r_0$ is statistically significant in the count part (coefficient $= 0.069$),  
indicating a positive effect on expected disambiguations.  
The $\kappa$ effect is not significant after controlling for $r_0$ and $n_T$.  
In the logit part, $n_T$ has a significant negative coefficient ($-0.067$),  
meaning that more true obstacles reduce the odds of zero disambiguations.

In summary, disambiguations—and hence traversal cost—increase with greater $r_0$ and $n_T$,  
while $\kappa$ has a relatively weaker influence in the presence of other variables.

\subsubsection{Mixed Obstacles from Uniform to Clustered Patterns}
\label{sec:OOP-mix-unif2clust}

We consider the case of mixed obstacles being inserted by OPA into the study window  
under the same setting as in Section~\ref{sec:OOP-clut-unif2clust}.  
We let the clustering parameter $\kappa$ vary from 2 to 12 in steps of 2,  
and the clustering radius $r_0$ take values in $\{2.5, 5, 7.5, 10, 15, 25, 50\}$.  
The total number of obstacles is fixed at $n_o = 10\,\kappa$ for each $\kappa$ value,  
with obstacle type compositions $(n_T, n_F)$ ranging from 10 to 100 in steps of 10,  
such that $n_T + n_F = n_o$.

Figure~\ref{fig:L-vs-Matern}(a) (bottom) shows the plot of mean traversal cost versus $r_0$ for each $\kappa$,  
averaged over obstacle number levels.  
The trend mirrors those seen in the false-only and true-obstacle-only cases:  
traversal cost increases with both $\kappa$ and $r_0$,  
and exhibits a concave-down quadratic profile in $r_0$.  
The corresponding contour plot in Figure~\ref{fig:L-vs-Matern}(b) (bottom) reinforces this pattern.  
As in earlier settings, the highest traversal costs occur when both $\kappa$ and $r_0$ are large.  
Compared to the false-only case, mean costs are generally higher across most $(\kappa, r_0)$ pairs,  
but remain lower than in the true-only case—except for very small $r_0$ values  
(e.g., $r_0 \approx 2.5$), where cluster compactness creates wider traversable gaps.

We fit a second-order robust linear model for traversal cost $C$,  
using $\kappa$, $r_0$, and $n_F$ (along with their squares and two-way interactions) as predictors.  
The robust model reduces residual standard error from 34.78 (OLS) to 7.68.  
No variables are eliminated during selection:
\begin{multline}
\label{eqn:rlm-mixed-cluster}
\widehat C = 97.165 - 0.84\,\kappa + 0.64\,r_0 + 0.134\,n_F + 0.125\,\kappa^2 - 0.0155\,r_0^2 + 0.0009\,n_F^2 \\
+ 0.09\,\kappa\,r_0 - 0.02\,\kappa\,n_F - 0.005\,r_0\,n_F.
\end{multline}

A RF regression using $C$ as the response and $\kappa$, $r_0$, and $n_F$ as predictors  
ranks variables by importance as: $\kappa$, $r_0$, and $n_F$—consistent with the linear model structure.

We also include $N_{dis}$ as a predictor in a second robust regression  
with the same covariates and interactions.  
The variable selection procedure eliminates the main effect of $\kappa$:
\begin{multline}
\label{eqn:rlm-mixed-cluster-Ndis}
\widehat C = 101.8 + 0.37\,r_0 - 0.10\,n_F + 11.674\,N_{dis} - 0.009\,r_0^2 + 0.00055\,n_F^2 + 0.11\,N_{dis}^2 \\
+ 0.043\,\kappa\,r_0 + 0.0072\,\kappa\,n_F + 0.876\,\kappa\,N_{dis} - 0.0015\,r_0\,n_F - 0.052\,r_0\,N_{dis} - 0.194\,n_F\,N_{dis}.
\end{multline}
The corresponding RF model with predictors $\kappa$, $r_0$, $n_F$, and $N_{dis}$  
ranks variable importance as: $N_{dis}$, $n_F$, $r_0$, and $\kappa$,  
confirming the dominant role of disambiguation count in explaining traversal cost in the mixed obstacle setting.

We fit a ZINB model to the disambiguation count $N_{dis}$  
using $\kappa$, $r_0$, and $n_F$ as predictors in the count part,  
and $n_T$ in the zero-inflation (logit) part.  
As in the false-only and true-only settings,  
both parts of the model yield statistically significant results.  
In the count part, the predictors $r_0$ and $n_F$  
have positive and statistically significant effects on $N_{dis}$,  
with estimated coefficients $0.08$ and $0.018$, respectively.  
This indicates that disambiguation count increases  
as the cluster radius and the number of false obstacles increase.  
While $\kappa$ also has a positive effect,  
it is not statistically significant after accounting for $r_0$ and $n_F$.  
In the logit part, $n_T$ has a statistically significant negative coefficient of $-0.063$,  
suggesting that the probability of observing zero disambiguations  
decreases as the number of true obstacles increases.  
This aligns with intuition:  
more true obstacles lead to more post-traversal disambiguation.  
Taken together, the ZINB model shows that $N_{dis}$  
increases with cluster radius $r_0$ and the number of false obstacles $n_F$,  
while the odds of zero disambiguation events decrease with increasing true obstacle count $n_T$.  
The clustering parameter $\kappa$ has minimal influence in this setting  
once $r_0$ and obstacle type composition are controlled for.

\subsection{Summary of the Simulation Results}
\label{sec:simulation-summary}

Based on the empirical results in Sections~\ref{sec:OOP-unif2reg} and~\ref{sec:OOP-unif2clust},  
we find that traversal cost is substantially higher when true obstacles are present  
compared to false obstacles, all else being equal.  
Hence, from OPA’s perspective, inserting more true obstacles—when available—  
is generally more effective at increasing traversal cost.

\noindent
\textbf{Obstacle Pattern Changing from Uniformity to Regularity:}
\begin{itemize}
 \item[] \textbf{False Obstacle Only Case:}  
   For Strauss$(n, d, \gamma)$ patterns, traversal cost is maximized when $d$ is moderate  
   (typically between 6 and 8) and $\gamma$ is small ($\lesssim 0.1$),  
   producing strong regularity.  
   In this regime, obstacle spacing reduces navigable corridors  
   without rendering the field too sparse.  
   For large $d$ ($\gtrsim 2r$), increasing $\gamma$  
   (i.e., reducing regularity) becomes preferable to avoid excessive spacing.

 \item[] \textbf{True Obstacle Only Case:}  
   Trends in $d$ and $\gamma$ resemble the false-only case,  
   with the difference that traversal cost is consistently higher  
   due to forced rerouting upon encountering true obstacles.  
   The effect of $\gamma$ and $d$ is negligible for small $n_T$ ($\leq 50$),  
   but becomes pronounced as $n_T$ increases.  
   Again, moderate $d$ and low $\gamma$ lead to maximal obstruction.

 \item[] \textbf{Mixed Obstacle Case:}  
   The optimal configuration closely follows that of the false obstacle case,  
   but with a key recommendation: if feasible,  
   OPA should prioritize inserting more true obstacles than false ones.  
   This increases the likelihood of disambiguation and rerouting by NAVA,  
   thus amplifying traversal cost.
\end{itemize}

\noindent
\textbf{Obstacle Pattern Changing from Uniformity to Clustering:}
\begin{itemize}
  \item 
  For all three obstacle types (false-only, true-only, and mixed),  
  traversal cost increases with both cluster radius $r_0$  
  and the number of clusters $\kappa$, especially when $r_0$ is moderate to large.  
  Tight clusters (small $r_0$) create navigable corridors, reducing traversal cost.  
  Larger $\kappa$ values spread obstacles more widely, increasing obstruction.

 \item 
 In the mixed case, cost trends again fall between the false- and true-only cases.  
 To maximize traversal cost, OPA should prefer configurations with larger $r_0$,  
 higher $\kappa$, and greater numbers of true obstacles than false.
\end{itemize}

\begin{remark}
\label{rem:summary-model-use}
The regression models in this section primarily serve to quantify  
the influence of covariates on traversal cost.  
However, they can also support predictive applications if key parameters—  
such as obstacle counts and spatial pattern characteristics—  
are known or estimated from data.  
In practice, Strauss and Matérn process parameters  
can be fitted using functions like \texttt{ppm} or \texttt{clusterfit}  
in the \texttt{spatstat.model} package in \texttt{R} \citep{baddeley2010}.  
These fitted models can then be used to simulate plausible scenarios  
and forecast traversal costs under different operational settings.
\end{remark}

\section*{Data Availability Statement}
We have published the data along with the corresponding simulation and analysis code on Zenodo. 
The materials are publicly available at \url{https://doi.org/10.5281/zenodo.17074761}
    
\bibliographystyle{apalike}
\bibliography{References}

\begin{thebibliography}{}

\bibitem[Aksakalli, 2007]{aksakalli2007}
Aksakalli, V. (2007).
\newblock The {BAO*} algorithm for stochastic shortest path problem with
  dynamic learning.
\newblock In {\em Proceedings of IEEE {C}onference on {D}ecision and
  {C}ontrol}, New Orleans, LA.

\bibitem[Aksakalli and Ari, 2013]{aksakalliari2013}
Aksakalli, V. and Ari, I. (2013).
\newblock Penalty-based algorithms for stochastic obstacle scene problem.
\newblock {\em INFORMS Journal on Computing}, 26(2):370--384.

\bibitem[Aksakalli and Ceyhan, 2012]{aksakalli:2012}
Aksakalli, V. and Ceyhan, E. (2012).
\newblock Optimal obstacle placement with disambiguations.
\newblock {\em The Annals of Applied Statistics}, 6(4):1730--1774.

\bibitem[Aksakalli et~al., 2011]{aksakalli2011}
Aksakalli, V., Fishkind, D.~E., Priebe, C.~E., and Ye, X. (2011).
\newblock The reset disambiguation policy for navigating stochastic obstacle
  fields.
\newblock {\em Naval Research Logistics}, 58(4):389--399.

\bibitem[An et~al., 2024]{An2024HexAStar}
An, Z., Rui, X., and Gao, C. (2024).
\newblock Improved a* navigation path‐planning algorithm based on hexagonal
  grid.
\newblock {\em ISPRS International Journal of Geo‐Information}, 13(5):166.

\bibitem[Aoude et~al., 2013]{aoude2013probabilistically}
Aoude, G.~S., Luders, B.~D., Joseph, J.~M., Roy, N., and How, J.~P. (2013).
\newblock Probabilistically safe motion planning to avoid dynamic obstacles
  with uncertain motion patterns.
\newblock {\em Autonomous Robots}, 35(1):51--76.

\bibitem[Azizi and Seifi, 2024]{azizi2024shortest}
Azizi, E. and Seifi, A. (2024).
\newblock Shortest path network interdiction with incomplete information: a
  robust optimization approach.
\newblock {\em Annals of Operations Research}, 335(2):727--759.

\bibitem[Baddeley, 2010]{baddeley2010}
Baddeley, A. (2010).
\newblock Analysing spatial point patterns in {R}.
\newblock {\em Workshop notes Ver. 4.1.}

\bibitem[Bar-Noy and Schieber, 1991]{bar-noy:1991}
Bar-Noy, A. and Schieber, B. (1991).
\newblock The {Canadian} traveller problem.
\newblock In {\em Proceedings of the Second Annual ACM-SIAM Symposium on
  Discrete Algorithms, pages 261– 270, San Francisco, CA.}

\bibitem[Breiman, 2001]{breiman2001}
Breiman, L. (2001).
\newblock Random forests.
\newblock {\em Machine Learning}, 1(45):5--32.

\bibitem[Chung et~al., 2019]{Chung2019risk}
Chung, J.~J., Smith, A.~J., Skeele, R., and Hollinger, G.~A. (2019).
\newblock Risk-aware graph search with dynamic edge cost discovery.
\newblock {\em The International Journal of Robotics Research},
  38(2--3):182--195.

\bibitem[Diggle, 2003]{diggle2003spatial}
Diggle, P.~J. (2003).
\newblock {\em Statistical Analysis of Spatial Point Patterns}.
\newblock Edward Arnold, London, 2nd edition.

\bibitem[Dijkstra, 1959]{Dijkstra1959}
Dijkstra, E.~W. (1959).
\newblock A note on two problems in connexion with graphs.
\newblock {\em Numerische Mathematik}, 1(1):269–271.

\bibitem[Eyerich et~al., 2009]{eyerich:2009}
Eyerich, P., Keller, T., and Helmert, M. (2009).
\newblock High-quality policies for the {Canadian} traveler’s problem.
\newblock In {\em Proceedings of the 24th AAAI Conference on Artificial
  Intelligence, Atlanta, Georgia.}

\bibitem[Fishkind et~al., 2007]{fishkind2007}
Fishkind, D.~E., Priebe, C.~E., Giles, K., Smith, L.~N., and Aksakalli, V.
  (2007).
\newblock Disambiguation protocols based on risk simulation.
\newblock {\em IEEE Transactions on Systems, Man, and Cybernetics},
  37(5):814--823.

\bibitem[Howard et~al., 2002]{howard2002risk}
Howard, A., Mataric, M.~J., and Sukhatme, G.~S. (2002).
\newblock Mobile sensor network deployment using potential fields: A
  distributed, scalable solution to the area coverage problem.
\newblock In {\em Proceedings of the 6th International Symposium on Distributed
  Autonomous Robotic Systems (DARS)}, pages 299--308, Berlin. Springer.

\bibitem[Huber, 1981]{Huber:1981}
Huber, P. (1981).
\newblock {\em Robust Statistics}.
\newblock John Wiley \& Sons, New York, NY.

\bibitem[Illian et~al., 2008]{illian2008statistical}
Illian, J., Penttinen, A., Stoyan, H., and Stoyan, D. (2008).
\newblock {\em Statistical Analysis and Modelling of Spatial Point Patterns}.
\newblock John Wiley \& Sons, Chichester.

\bibitem[Israeli and Wood, 2002]{israeli2002shortest}
Israeli, E. and Wood, R.~K. (2002).
\newblock Shortest-path network interdiction.
\newblock {\em Networks: An International Journal}, 40(2):97--111.

\bibitem[Jenelius and Mattsson, 2015]{Jenelius2015Vulnerability}
Jenelius, E. and Mattsson, L. (2015).
\newblock Road network vulnerability analysis: Conceptualization,
  implementation and application.
\newblock {\em Computers, Environment and Urban Systems}, 49:136--147.

\bibitem[Kuhn and Johnson, 2013]{kuhn2013}
Kuhn, M. and Johnson, K. (2013).
\newblock {\em Applied Predictive Modeling}.
\newblock Springer-Verlag New York Inc.

\bibitem[LaValle, 2006]{lavalle2006planning}
LaValle, S.~M. (2006).
\newblock {\em Planning Algorithms}.
\newblock Cambridge University Press, Cambridge.
\newblock Available online at \url{http://planning.cs.uiuc.edu/}.

\bibitem[Lv et~al., 2024]{Lv2024OffRoad}
Lv, Z., Ni, L., Peng, H., Zhou, K., Zhao, D., Qu, G., Yuan, W., Gao, Y., and
  Zhang, Q. (2024).
\newblock Research on global off‐road path planning based on improved a*
  algorithm.
\newblock {\em ISPRS International Journal of Geo‐Information}, 13(10):362.

\bibitem[Maidana et~al., 2023]{Maidana2023}
Maidana, R.~G., Kristensen, S.~D., Utne, I.~B., and Sørensen, A.~J. (2023).
\newblock Risk-based path planning for preventing collisions and groundings of
  maritime autonomous surface ships.
\newblock {\em Ocean Engineering}, 290:116417.

\bibitem[McRae et~al., 2008]{McRae2008Circuit}
McRae, B.~H., Dickson, B.~G., Keitt, T.~H., and Shah, V.~B. (2008).
\newblock Using circuit theory to model connectivity in ecology, evolution, and
  conservation.
\newblock {\em Ecology}, 89(10):2712--2724.

\bibitem[Meng et~al., 2022]{Meng2022nrrrt}
Meng, F., Chen, L., Ma, H., Wang, J., and Meng, M. Q.-H. (2022).
\newblock Nr-rrt: Neural risk-aware near-optimal path planning in uncertain
  nonconvex environments.
\newblock {\em IEEE Transactions on Automation Science and Engineering}.
\newblock Preprint available at \url{https://arxiv.org/abs/2205.06951}.

\bibitem[Missiuro and Roy, 2006]{missiuro2006adapting}
Missiuro, P. and Roy, N. (2006).
\newblock Adapting probabilistic roadmaps to handle uncertain maps.
\newblock In {\em Proceedings of the IEEE International Conference on Robotics
  and Automation (ICRA)}, pages 1261--1267. IEEE.

\bibitem[M{\o}ller and Waagepetersen, 2004]{moller2004statistical}
M{\o}ller, J. and Waagepetersen, R.~P. (2004).
\newblock {\em Statistical Inference and Simulation for Spatial Point
  Processes}.
\newblock Chapman \& Hall/CRC, Boca Raton.

\bibitem[Nikolova and Karger, 2008]{nikolova:2008}
Nikolova, E. and Karger, D.~R. (2008).
\newblock Route planning under uncertainty: The {C}anadian traveller problem.
\newblock In {\em the 23rd AAAI Conference on Artificial Intelligence},
  Chicago, IL.

\bibitem[Papadimitriou and Yannakakis, 1991]{papadimitriou:1991}
Papadimitriou, C.~H. and Yannakakis, M. (1991).
\newblock Shortest paths without a map.
\newblock {\em Theoretical Computer Science}, 84(1):127--150.

\bibitem[Parajuli et~al., 2023]{Parajuli2023FloodEvac}
Parajuli, G., Neupane, S., Kunwar, S., Adhikari, R., and Acharya, T. (2023).
\newblock A gis‐based evacuation route planning in flood‐susceptible area
  of siraha municipality, nepal.
\newblock {\em ISPRS International Journal of Geo‐Information}, 12(7):286.

\bibitem[Priebe et~al., 2005]{Priebe-NRL:2005}
Priebe, C.~E., Fishkind, D.~E., Abrams, L., and Piatko, C.~D. (2005).
\newblock Random disambiguation paths for traversing a mapped hazard field.
\newblock {\em Naval Research Logistics}, 52(3):285--292.

\bibitem[Sadeghi and Seifi, 2024]{sadeghi2024modified}
Sadeghi, S. and Seifi, A. (2024).
\newblock A modified scenario bundling method for shortest path network
  interdiction under endogenous uncertainty.
\newblock {\em Annals of Operations Research}, pages 1--29.

\bibitem[Smith and Song, 2020]{smith2020survey}
Smith, J.~C. and Song, Y. (2020).
\newblock A survey of network interdiction models and algorithms.
\newblock {\em European Journal of Operational Research}, 283(3):797--811.

\bibitem[Sundar et~al., 2021]{sundar2021credible}
Sundar, K., Misra, S., Bent, R., and Pan, F. (2021).
\newblock Credible interdiction for transmission systems.
\newblock {\em IEEE Transactions on Control of Network Systems}, 8(2):738--748.

\bibitem[West, 2001]{west:2001}
West, D. (2001).
\newblock {\em Introduction to Graph Theory, $2^{nd}$ Edition}.
\newblock Prentice Hall, NJ.

\bibitem[Witherspoon et~al., 1995]{witherspoon1995COBRA}
Witherspoon, N., Holloway, J., Davis, K., Miller, R., and Dubey, A. (1995).
\newblock The {Coastal Battlefield Reconnaissance and Analysis (COBRA)} program
  for minefield detection.
\newblock {\em SPIE: Detection Technologies for Mines and Minelike Targets},
  2496:500--508.

\bibitem[Xu et~al., 2009]{xu:2009CTP}
Xu, Y., Hu, M., Su, B., Zhu, B., and Zhu, Z. (2009).
\newblock The {C}anadian traveller problem and its competitive analysis.
\newblock {\em Journal of Combinatorial Optimization}, 18(2):195--205.

\bibitem[Zeileis et~al., 2008]{Zeileis2008}
Zeileis, A., Kleiber, C., and Jackman, S. (2008).
\newblock Regression models for count data in {R}.
\newblock {\em Journal of Statistical Software}, 27(8):1–25.

\end{thebibliography}
\end{document}